\documentclass[letterpaper,11pt,USenglish]{scrartcl}
\usepackage[utf8]{inputenc}
\usepackage[modulo]{lineno}

\DeclareMathAlphabet\mathbfcal{OMS}{cmsy}{b}{n}

\iftrue
\RedeclareSectionCommand[%
  beforeskip=1ex plus.2ex minus.2ex
]{paragraph}
\fi
\usepackage{auxlib}
\usepackage[letterpaper,margin=1in]{geometry}
\usepackage{xspace}
\usepackage{listings}
\usepackage{framed}
\usepackage{float}
\MakeOuterQuote{"}
\makeatletter
\def\HyPsd@CatcodeWarning#1{}
\makeatother
\let\corollary\relax
\let\lemma\relax
\let\proposition\relax
\let\observation\relax
\declaretheorem[style=plain, sibling=theorem]{corollary}
\declaretheorem[style=plain, sibling=theorem]{lemma}
\declaretheorem[style=plain, sibling=theorem]{proposition}
\declaretheorem[style=plain, sibling=theorem]{observation}

\makeatletter
\newcommand{\DeclareMathActive}[2]{%
  \expandafter\edef\csname keep@#1@code\endcsname{\mathchar\the\mathcode`#1 }
  \begingroup\lccode`~=`#1\relax
  \lowercase{\endgroup\def~}{#2}%
  \AtBeginDocument{\mathcode`#1="8000 }%
}

\newcommand{\std}[1]{\csname keep@#1@code\endcsname}
\patchcmd{\newmcodes@}{\mathcode`\-\relax}{\std@minuscode\relax}{}{\ddt}
\AtBeginDocument{\edef\std@minuscode{\the\mathcode`-}}
\makeatother

\DeclareMathActive{O}{\operatorname{\std{O}}}

\newcommand*{\BigO}[1]{\ensuremath{\mathchoice{\LDAUOmicron{#1}}{\smash{\LDAUOmicron{#1}}}{\LDAUOmicron{#1}}{\LDAUOmicron{#1}}}}
\newcommand*{\BigOmega}[1]{\ensuremath{\mathchoice{\LDAUOmega{#1}}{\smash{\LDAUOmega{#1}}}{\LDAUOmega{#1}}{\LDAUOmega{#1}}}}
\newcommand*{\BigTheta}[1]{\ensuremath{\LDAUTheta{#1}}}
\newcommand*{\LittleO}[1]{\ensuremath{\LDAUomicron{#1}}}

\def\Vardef#1{%
	\expandafter\newcommand\csname #1\endcsname[1]{%
		\def\first{##1}%
		\def\second{*}%
		\def\third{}%
		\ensuremath{\mathsf{\MakeLowercase #1}\ifx\first\second\else\ifx\first\third\else[{##1}]\fi\fi}%
	}%
}
\def\Vardefxx#1#2{%
	\expandafter\newcommand\csname #1\endcsname[1]{%
		\def\first{##1}%
		\def\second{*}%
		\def\third{}%
		\ensuremath{\mathsf{#2}\ifx\first\second\else\ifx\first\third\else[{##1}]\fi\fi}%
	}%
}
\def\Vardefx#1#2{%
	\expandafter\newcommand\csname #1\endcsname[1]{%
		\ensuremath{\mathsf{#2}[{##1}]}%
	}%
}

\def\T#1{\ensuremath{\tau_{\text{\textsc{\MakeLowercase{#1}}}}}}

\Vardef{State}

\Vardef{Rounds}
\Vardef{Count}
\Vardef{Tick}
\Vardef{Clock}
\Vardefx{DecisionFlag}{decision}

\Vardefx{TickE}{tick_{E}}
\Vardef{Undecided}
\Vardef{Opinion}
\Vardef{PhaseType}
\Vardef{Round}
\def\ttrue{\ensuremath{\text{{\small\textsc{True}}}}\xspace}
\def\tfalse{\ensuremath{\text{{\small\textsc{False}}}}\xspace}

\let\epsilon\varepsilon
\let\phi\varphi

\def\Boosting{boosting\xspace}
\def\Decision{decision\xspace}
\let\Decided\psi
\newcommand*{\Sigset}[1]{\ensuremath{\mathcal{S}{\left(#1\right)}}}

\newcommand*{\concConstEps}{\ensuremath{\varepsilon^{*}}}
\newcommand*{\polyaConstSmall}{\ensuremath{\varepsilon_p}}
\newcommand*{\polyaConstLarge}{\ensuremath{c_{p}}}

\def\protocol#1{\textsc{#1}\xspace}

\def\sDelta{\xi\xspace}
\def\2whp{with probability at least $1 - \BigO{n^{-2}}$\xspace}

\floatname{algoflt}{Algorithm}
\floatstyle{boxed}
\newfloat{algoflt}{tbp}{loa}
\def\dcmnumberstyle{}
\def\dcmbasicstyle{\small}
\makeatletter
\newif\ifdcmlinenumbers
\dcmlinenumberstrue
\lstnewenvironment{lstalgo}[2][htb]{%
\lstset{%
	mathescape,%
	tabsize=4,%
	numbers=left,%
	numberstyle=\tiny\dcmnumberstyle,%
	numberblanklines=true,%
	showlines=true,%
	basicstyle=\rmfamily\dcmbasicstyle,%
	columns=fullflexible,%
	emph={%
		[1]if, else, then, and, or, for, do, while, return, exit, not, output, initialize %
		},
	emphstyle={%
		[1]\bfseries%
		},%
	escapeinside={/*}{*/},
	escapebegin={\ifmmode\else\color{gray}\em\small$\triangleright$ \fi},
	escapeend={\ifmmode\else\fi}
	}%
	\ifdcmlinenumbers
	\else
	\lstset{numbers=none}
	\fi
\setlength{\topsep}{0pt}%
\setlength{\parindent}{0pt}%
\@skiphyperreftrue%
\algoflt[#1]\caption{#2}%
}{\endalgoflt\@skiphyperreffalse}

\AfterPackageLoaded{cleveref}{
    \crefname{lst@lineno}{Line}{Lines}
    \crefname{algoflt}{Algorithm}{Algorithms}
}

\renewcommand*{\@maketitle}{%
  \null
  \vskip 2em%
  \global\@topnum=\z@ 
  \setparsizes{\z@}{\z@}{\z@\@plus 1fil}\par@updaterelative 
  \begin{center}% 
    {\usekomafont{title}{\huge \@title \par}}% 
    \vskip .5em 
    {\ifx\@subtitle\@empty\else\usekomafont{subtitle}\@subtitle\par\fi}% 
    \vskip 1em 
    {% 
      \usekomafont{author}{% 
        \lineskip .5em% 
        \begin{tabular}[t]{c} 
          \@author 
        \end{tabular}\par 
    	\vskip .5em 
        \begin{tabular}[t]{cc} 
          \@affiliation
        \end{tabular}\par 
      }% 
    }% 
  \end{center}% 
  \par 
  \vskip 2em 
}% 
\def\@affiliation{}
\def\affiliation#1{\def\@affiliation{#1}}
\renewenvironment{abstract}{\quotation\small\textbf{\abstractname.} }{\endquotation}
\makeatother

\addtokomafont{disposition}{\rmfamily}
\addtokomafont{author}{\normalsize}
\makeatletter
\def\?#1{}
\def\whp{w.h.p\@ifnextchar.{.\?}{\@ifnextchar,{.}{\@ifnextchar){.}{\@ifnextchar:{.:\?}{.\ }}}}}
\def\Whp{W.h.p\@ifnextchar.{.\?}{\@ifnextchar,{.}{.\ }}}
\makeatother

\usepackage[
		isbn=false,
		sortcites,
		maxbibnames=99,
		maxcitenames=2,
		maxalphanames=4,
		minalphanames=3,
		style=numeric,
		url=false,
		giveninits=true,
		backend=biber,
	]{biblatex}
\DeclareCiteCommand{\nobracketcite}
  {\usebibmacro{prenote}}
  {\usebibmacro{citeindex}%
   \usebibmacro{cite}}
  {\textcolor{black}{\multicitedelim}}
  {\usebibmacro{postnote}}

\makeatletter
\newcommand\fullciteauthor[1]{\def\blx@maxcitenames{99}\citeauthor{#1}\def\blx@maxcitenames{2}}
\makeatother

\defbibenvironment{midbib}
  {\list
     {}
     {\setlength{\leftmargin}{\bibhang}%
      \setlength{\itemsep}{\bibitemsep}%
      \setlength{\parsep}{\bibparsep}}}
  {\endlist}
  {\item}

\usepackage{xpatch}
\DeclareFieldFormat*{citetitle}{\itshape #1}
\DeclareFieldFormat*{title}{\itshape #1\isdot}

\DeclareFieldFormat[inproceedings]{volume}{Volume #1}
\DeclareFieldFormat*{pages}{pages #1}
\AtEveryBibitem{% Clean up the bibtex rather than editing it
 \ifentrytype{inproceedings}{%
  \iffieldundef{series}{}{%
    \savefield{volume}{\theoriginalvolume}%
    \restorefield{number}{\theoriginalvolume}%
    \clearfield{volume}%
  }%
 }{}%
 \ifentrytype{book}{%
  \iffieldundef{series}{}{%
    \savefield{volume}{\theoriginalvolume}%
    \restorefield{number}{\theoriginalvolume}%
    \clearfield{volume}%
  }%
 }{}%
 \clearlist{address}%
 \clearlist{location}%
 \ifentrytype{book}{}{% Remove publisher and editor except for books
  \clearname{editor}%
 }%
}
\AtEveryCitekey{\UseBibitemHook}
\renewcommand*{\multicitedelim}{\addcomma\space}

\renewbibmacro*{issue+date}{%
  \setunit{\addcomma\space}
    \iffieldundef{issue}
      {\usebibmacro{date}}
      {\printfield{issue}%
       \setunit*{\addspace}%
       \usebibmacro{date}}%}%
  \newunit}

\renewbibmacro*{publisher+location+date}{%
  \printlist{location}%
  \iflistundef{publisher}
    {\setunit*{\addcomma\space}}
    {\setunit*{\addcolon\space}}%
  \printlist{publisher}%
  \setunit*{\addcomma\space}%
  \addcomma\space\usebibmacro{date}%
  \newunit}

\DeclareFieldFormat*{number}{\mkbibparens{#1}}
\DeclareFieldFormat[inproceedings]{number}{\addcomma\space volume #1\addcomma}
\DeclareFieldFormat[incollection]{volume}{\addcomma\space volume #1\addcomma}
\DeclareFieldFormat[article]{volume}{volume #1}

\renewbibmacro*{volume+number+eid}{%
  \setunit{\addcomma\space}
  \printfield{volume}%
  \setunit*{\addnbthinspace}% NEW (optional); there's also \addnbthinspace
  \printfield{number}%
}

\xpatchbibdriver{incollection}{%
  \usebibmacro{maintitle+booktitle}%
  \newunit\newblock
}{%
  \usebibmacro{maintitle+booktitle}%
  \iffieldundef{volume}{
    \newunit\newblock
  }{
    \addcomma\space\newblock
  }
}{}{}

\xpatchbibdriver{book}{%
\newunit\newblock
\printfield{edition}%
}
{%
\setunit{\addcomma\space}\newblock
\printfield{edition}%
}{}{}

\DeclareCiteCommand{\citefauthor}
{\defcounter{maxnames}{99}%
\defcounter{minnames}{99}%
\defcounter{uniquename}{2}%
\boolfalse{citetracker}%
\boolfalse{pagetracker}%
\usebibmacro{prenote}}
{\ifciteindex{\indexnames{labelname}}{}%
\printnames{labelname}}
{\multicitedelim}
{\usebibmacro{postnote}}

\DeclareCiteCommand{\footnotecite}[\iffootnote\mkbibparens\mkbibfootnote]
  {}
  {\usedriver
     {\defcounter{minnames}{99}%
\renewbibmacro{in:}{. In }%
\defcounter{maxnames}{99}}%
{\thefield{entrytype}}}
  {\multicitedelim}
  {\usebibmacro{postnote}}

\DefineBibliographyStrings{english}{%
  mathesis = {Master's thesis},
}

\def\cL{\ensuremath{\mathcal{L}}\xspace}
\def\cX{\ensuremath{\mathcal{X}}\xspace}

\def\x{\ensuremath{\mathbf{x}}\xspace}
\def\y{\ensuremath{\mathbf{y}}\xspace}

\def\X{\ensuremath{\mathbf{X}}\xspace}
\def\Y{\ensuremath{\mathbf{Y}}\xspace}
\newcommand{\config}[1]{\mathbf{#1}}
\let\E\Ex

\def\Bin{\operatorname{Bin}}
\def\PE{\operatorname{PE}}

\title{Fast Consensus via the Unconstrained Undecided State Dynamics}

\author{
Gregor Bankhamer\footnotemark[1] \and
Petra Berenbrink\footnotemark[2] \and 
Felix Biermeier\footnotemark[2] \and
Robert Elsässer\footnotemark[1] \and
Hamed Hosseinpour\footnotemark[2] \and
Dominik Kaaser\footnotemark[2] \and
Peter Kling\footnotemark[2] \\
}
\affiliation{
\footnotemark[1] University of Salzburg, Austria & 
\footnotemark[2] Universität Hamburg, Germany\\
\textit{\small\{ gbank, elsa \}@cs.sbg.ac.at} &
\textit{\small firstname.lastname@uni-hamburg.de}
}

\begin{document}
\crefname{enumi}{Statement}{Statements}
\setcounter{page}{0}
\maketitle 

\renewcommand{\thefootnote}{\fnsymbol{footnote}}
\footnotetext[1]{ Supported in part by the European Union’s Horizon 2020 research and innovation programme under Grant Agreement no. 824115 (HiDALGO).}
\renewcommand{\thefootnote}{\arabic{footnote}}

\thispagestyle{empty}
\begin{abstract}
We consider the \emph{plurality consensus} problem for $n$ agents.
Initially, each agent has one of $k$ opinions.
Agents choose random interaction partners and revise their state according to a fixed transition function, depending on their own state and the state of the interaction partners.
The goal is to reach a configuration in which all agents agree on the same opinion.
If there is initially a sufficiently large bias towards some opinions one of them should prevail.

In this paper we consider a synchronized variant of the undecided state dynamics where the agents use so-called phase clocks.
The phase clocks divide the time in overlapping phases.
Each phase consists of a decision and a boosting part.
In the decision part, any agent that encounters an agent with a different opinion becomes undecided.
In the boosting part, undecided agents adopt the first opinion they encounter.
We consider this dynamics both in the sequential \emph{population model} and the parallel \emph{gossip model}.

In the population model agents interact in randomly chosen pairs, one pair per time step.
The runtime is measured in \emph{parallel time} (number of interactions divided by $n$).
We show that our protocol reaches consensus (w.h.p.) in $O(\log^2 n)$ parallel time, providing the first polylogarithmic result for $k > 2$ (w.h.p.) in this model.
If  there is an initial bias of $\Omega(\sqrt{n \log n})$, then (w.h.p.) that opinion wins.

The gossip model assumes parallel rounds.
During each round every agent is allowed  to communicate with one randomly chosen agent.  
Here it is known that consensus can be reached fast (in polylogarithmic time) if there is a bias of order $\Omega(\sqrt{n\log n})$ towards one opinion [Ghaffari and Parter, PODC'16; Berenbrink et al., ICALP'16].
Without any assumption on the bias, fast consensus has only been shown for $k = 2$ for the unsynchronized version of the undecided state dynamics [Clementi et al., MFCS'18]. 
To account for the yet unsolved general case, we show that the synchronized variant of the undecided state dynamics reaches consensus (w.h.p.) in time $O(\log^2 n)$ for every initial configuration.
Again, we guarantee that if there is an initial bias of $\Omega(\sqrt{n \log n})$, then (w.h.p.) that opinion wins.

A simple extension of our protocol in the gossip model yields a dynamics that does not depend on $n$ or $k$, is anonymous, and has (w.h.p.) runtime $O(\log^2 n)$.
This solves an open problem formulated by Becchetti et al.~[Distributed Computing,~2017].

\end{abstract}
\clearpage

\section{Introduction}

We consider \emph{plurality consensus} in a distributed system consisting of $n$ \emph{agents}.
Initially each agent is assigned one of $k$ \emph{opinions}.
Agents interact in pairs and update their opinions based on other opinions they observe.
Eventually, all agents must agree on the same opinion.
Moreover, if there is a sufficiently large \emph{bias} -- i.e., the difference between the number of agents initially assigned to the most common and second most common opinion -- the most common opinion should prevail.

Such consensus problems represent a fundamental primitive in distributed computing.
Practical applications include fault tolerant sensor arrays (where a trustworthy result must be confirmed by a majority of the sensors) or majority-based conflict resolution (e.g., for CRCW PRAMs).
They are also used in physics and biology to model huge dynamic systems of particles or bacteria, or in social sciences to study how opinions form and spread through social interactions.
See~\cite{DBLP:journals/sigact/BecchettiCN20} for references and further applications.

Our focus lies on a specific family of consensus protocols, the \emph{undecided state dynamics (USD)}~\cite{DBLP:journals/dc/AngluinAE08, DBLP:conf/soda/BecchettiCNPS15}.
Here, interaction pairs are chosen randomly.
Any agent that encounters an agent with a different opinion becomes \emph{undecided} (loses its opinion).
Subsequently, such undecided agents adopt the first opinion they observe.
Variants of this simple idea have been studied in different models and have proven to be surprisingly efficient and robust.
We consider a variant of the USD in two models, namely in
(i) the \emph{population model}, where each time step \emph{one pair} of randomly chosen agents interact; and
(ii) the \emph{gossip model}, where \emph{every} agent \emph{simultaneously} interacts in discrete rounds with another randomly chosen agent.
\Cref{sec:model} provides a formal description of both models.

\Textcite{DBLP:conf/icalp/BerenbrinkFGK16, DBLP:conf/podc/GhaffariP16a} introduce a \emph{synchronized} version of the USD in the gossip model.
Here, a synchronization mechanism known as \emph{phase clocks} is used to make the agents jointly progress through \emph{phases} of length $\BigTheta{\log n}$, alternating between \emph{\Decision parts} (where agents become undecided if they encounter a different opinion) and \emph{\Boosting parts} (where undecided agents adopt one of the remaining opinions).
The results from~\cite{DBLP:conf/icalp/BerenbrinkFGK16, DBLP:conf/podc/GhaffariP16a} showed that such a synchronized USD reaches consensus in $\ldauOmicron{\log k \cdot \log n}$ rounds \whp, a significant speedup compared to the original USD in the gossip model (for which there are initial configurations where it requires $\ldauOmega{k}$ rounds~\cite{DBLP:conf/soda/BecchettiCNPS15}).
However, the results of~\cite{DBLP:conf/icalp/BerenbrinkFGK16, DBLP:conf/podc/GhaffariP16a}~hold only if there is an initial bias of $\softOmega{\sqrt{n}}$ towards one opinion.
In fact, we do not know of any simple dynamics in the gossip model that achieves consensus for $k > 2$ in polylogarithmic time and requires no bias.

\paragraph{Results in a Nutshell}
We consider the synchronized USD protocol for both the population and the gossip model.
Our results hold for the general case of up to $n$ opinions and independently of any bias.
Our protocols reach, \whp, consensus in $\ldauOmicron{n \log^2 n}$ interactions in the population model and in $\ldauOmicron{(\log k + \log\log n) \cdot \log n}$ rounds in the gossip model.
Both protocols have the following property: if there is a \emph{plurality} opinion (support by an additive term of at least $\BigOmega{\sqrt{n \log n}}$ larger compared to any other opinion), the agents agree on that opinion.
Otherwise, they agree on what we call a \emph{significant} opinion (an opinion whose support is at most $\ldauOmicron{\sqrt{n \log n}}$ smaller than the maximum support).
The population protocol requires only $k \cdot \ldauOmicron{\log n}$ states per agent ($\log k + \BigO{\log\log n}$ bits), while in the gossip model we require $\log k + \ldauOmicron{\log\log k + \log\log\log n}$ bits per agent.
To the best of our knowledge, our protocols are the first protocols that solve the consensus problem in these two models without restrictions on the number of opinions or the initial bias in polylogarithmic time.
If the initial configuration has an additive bias of order $\BigOmega{\sqrt{n \log n}}$, our results essentially match the known results from~\cite{DBLP:conf/podc/GhaffariP16a, DBLP:conf/icalp/BerenbrinkFGK16}.

Our protocols use simple phase clock implementations for the synchronization, which allows us to focus on our main contribution – completely unconditional bounds for the synchronized USD.
However, by using more sophisticated phase clock implementations, both of our protocols can be improved.
In particular, in the case of the gossip model we can avoid any knowledge of $n$, such that the protocol becomes uniform.
This allows us to answer an open question from~\cite{DBLP:journals/dc/BecchettiCNPST17}, who asked for a simple dynamics (which should specifically be anonymous and uniform) that achieves plurality consensus in polylogarithmic time for any $k$ and any initial configuration.
Details for this uniform protocol can be found in \cref{sec:uniform-protocol}.

For both protocols the  main challenge is to handle the case without a clear bias, i.e., there is almost no
difference between the number of agents initially assigned to the most common and second most common opinion, and where the number of opinions $k$ lies in $\ldauOmega{\sqrt{n} / \log n}$ and $\ldauOmicron{\sqrt{n}}$.
In these cases the support of the opinions is too small to be tracked via concentration bounds but at the same time the support of all opinions is large enough to prevent them from vanishing fast.
The only results for the unbiased case and large $k$
known from the literature rely on a coupling with the so-called \protocol{Voter} 
process~\cite{DBLP:conf/podc/GhaffariL18, DBLP:conf/podc/BerenbrinkCEKMN17}. This approach results in a convergence time of 
%. The latter process requires 
$\tilde{O}(n^{2/3})$ instead of the polylogarithmic convergence time shown here.

The analysis for our population protocols is divided into three cases, depending on the number of opinions $k$ (see  \Cref{sec:analysis-case1, sec:analysis-case2, sec:analysis-case3}).
\Cref{sec:concentration-population-model} gathers several results which we use in all three cases.
At the beginning of each of these sections we give a detailed overview about the analysis methods  and we compare our approach with existing results.
The analysis of the Gossip protocol is based on the analysis of the population protocol and  can be found in \cref{sec:analysis-gossip-model}.

\subsection*{Related Work.}%
\label{sec:related-work}

In this overview we focus on \emph{plurality consensus}, a term that typically refers to comparatively simple consensus protocols under random, pairwise interactions that aim to agree on one of possibly many opinions.
%Moreover, if the bias is large enough, the opinion with the largest support should prevail.
A more detailed discussion of distributed consensus problems can be found in~\cite{DBLP:journals/sigact/BecchettiCN20}.

\paragraph{Undecided State Dynamics}
\Textcite{DBLP:journals/dc/AngluinAE08} introduced the undecided state dynamics (USD) in the population model and consider two opinions.
They show that their $3$-state protocol reaches consensus \whp in $\BigO{n \cdot \log{n}}$ interactions. If the bias is of order $\ldauomega{\sqrt n \cdot \log n}$ it converges towards the initial majority \whp. \Textcite{DBLP:journals/nc/CondonHKM20} reduced this required bias to $\Omega(\sqrt{n \log n})$.
It is worth mentioning that they also consider a variant of plurality consensus with $k$ opinions in a communication model in which three randomly chosen agents interact in a step. In this variant, if an agent interacts with two other agents of the same opinion then it adopts this opinion.
They show that the system converges to the initial majority within $\BigO{k n \log n}$ interactions \whp, provided the initial bias is large enough.
\Textcite{DBLP:conf/mfcs/ClementiGGNPS18} study the USD in the gossip model. They also consider the unbiased case and $k = 2$ opinions 
and show that, \whp, the protocol reaches consensus in $\BigO{\log n}$ rounds.
Moreover, they show that the plurality opinion prevails if the initial bias is $\BigOmega{\sqrt{n \log n}}$.
Their analysis partitions the configuration space into a total of seven cases, depending on the magnitude of a possible bias and on the number of undecided agents, which makes it hard to apply the approach to arbitrary values of $k$.
Another recent work by \Textcite{DBLP:conf/sirocco/DAmoreCN20}
%todo: fix the DOI! it is not working!(look at the reference)
also considers the USD with $k=2$. While they require the usual initial bias of $\Omega(\sqrt{n \log n})$ they introduce \emph{noise} in their model which may modify sent messages with a certain probability $p$. Within $\BigO{\log n}$ time their protocol reaches a bias of $\Theta(n)$ towards the initial majority.
\Textcite{DBLP:conf/soda/BecchettiCNPS15} adopt the USD to the gossip model and generalize it to $k = \ldauOmicron{{(n/\log n)}^{1/3}}$ opinions.
Given an initial multiplicative bias, their protocol requires, \whp, $\ldauOmicron{k \cdot \log n}$ rounds to achieve consensus on the plurality opinion.
Deviating slightly from the original USD definition, \textcite{DBLP:conf/podc/GhaffariP16a, DBLP:conf/icalp/BerenbrinkFGK16} consider a \emph{synchronized} version of the USD\@.
For the synchronization, both suggest basically the same protocol, which uses counters to partition time into phases of length $\BigTheta{\log k}$.
Agents can become undecided only at the start of such a phase and use the rest of the phase to obtain a new opinion.
Both protocols achieve consensus in $\ldauOmicron{\log k \log n}$ rounds \whp, using $\log k + \ldauOmicron{\log\log k}$ bits per agent.
The runtime can be slightly reduced to $\ldauOmicron{\log k \cdot \log\log_{\alpha} n + \log\log n}$, where $\alpha$ denotes the initial multiplicative bias~\cite{DBLP:conf/icalp/BerenbrinkFGK16}.
Both \cite{DBLP:conf/podc/GhaffariP16a, DBLP:conf/icalp/BerenbrinkFGK16} proceed to use (different) more sophisticated synchronization mechanisms to design protocols that require only $\log k + \ldauOmicron{1}$ bits and maintain (essentially) the same runtime bounds (the refined runtime of~\cite{DBLP:conf/icalp/BerenbrinkFGK16} becomes $\ldauOmicron{\log(n) \cdot \log\log_{\alpha} n}$). Note that neither~\cite{DBLP:conf/podc/GhaffariP16a} nor~\cite{DBLP:conf/icalp/BerenbrinkFGK16} extend to the case without bias or very large $k$: their techniques are based on chains of concentration bounds, which are no longer applicable in the general setting.

\paragraph{Consensus in the Population Model}
With respect to plurality consensus there are two lines of related work.
One considers the (exact) \emph{majority} problem, where one seeks to achieve (guaranteed) majority consensus among $k = 2$ opinions, even if the additive bias is as small as one~\cite{DBLP:journals/siamco/DraiefV12, DBLP:conf/icalp/MertziosNRS14, DBLP:conf/podc/AlistarhGV15, DBLP:conf/nca/MocquardAABS15, DBLP:journals/dc/DotyS18, DBLP:conf/soda/AlistarhAEGR17, DBLP:conf/soda/AlistarhAG18, DBLP:conf/podc/BilkeCER17, DBLP:conf/podc/KosowskiU18, DBLP:conf/wdag/BerenbrinkEFKKR18, DBLP:journals/corr/abs-1805-04586, BKKP20}.
The currently best protocol by \Textcite{doty2021time} solves exact majority with $\ldauOmicron{\log n}$ states and $\ldauOmicron{\log n}$ stabilization time, both in expectation and \whp.
%The currently best protocol by \textcite{BKKP20} solves exact majority in $\ldauOmicron{n \cdot \log^{3/2}{n}}$ interactions using $\ldauOmicron{\log{n}}$ states.
%It improves the protocol of~\cite{DBLP:conf/wdag/BerenbrinkEFKKR18}.
\Textcite{DBLP:conf/soda/AlistarhAEGR17} show that any stable majority protocol using $\log\log{n} / 2$ states requires $\BigOmega{n^2 / {\log\log n}}$ interactions in expectation.
Assuming some natural properties, \textcite{DBLP:conf/soda/AlistarhAG18} show that any majority protocol which stabilizes in expected
${n^{2-\BigOmega{1}}}$ interactions requires $\BigOmega{\log{n}}$ states.

The other line of research~\cite{DBLP:journals/tsipn/SalehkaleybarSG15, DBLP:journals/jstsp/BenezitTV11, DBLP:conf/ciac/NataleR19} is related to signal processing and studies voting on graphs.
On the complete graph, the model becomes equivalent to the population model.
%These results typically study bounds on the number of states required to achieve plurality consensus (for fixed or general $k \geq 2$) in finite (otherwise unspecified) time.
\Textcite{DBLP:journals/jstsp/BenezitTV11} provide plurality consensus protocols for $k \in \set{3, 4}$ with $15$ and $100$ states, respectively.
\Textcite{DBLP:journals/tsipn/SalehkaleybarSG15} provide a protocol for arbitrary $k$ that uses $\ldauOmicron{k \cdot 2^k}$ states.
%For the complete graph they give runtime bounds, but those may be arbitrarily bad, even for a large bias.
\Textcite{DBLP:conf/ciac/NataleR19} improve the number of states to $\ldauOmicron{k^{11}}$ and establish a lower bound of $\BigOmega{k^2}$ states to solve plurality consensus with probability $1$.
In a recent result \cite{DBLP:conf/podc/BankhamerEKK20}, the authors consider a variant of the population model where agents are activated by random clocks.
%At each clock tick, every agent may open communication channels to constantly many other agents chosen uniformly at random or from a list of at most constantly many agents contacted in previous steps.
%In this model, opening communication channels is subject to a random delay.
The authors show that consensus is reached by all but a $1/\poly\log n$ fraction of agents in $\BigO{\log\log_\alpha k \log k + \log\log n}$ time, provided a sufficiently large bias is present.

\paragraph{Consensus in the Gossip Model}
Different types of consensus processes on graphs have been considered in the gossip model.
%The most well-studied processes are \protocol{Voter}, \protocol{TwoChoices}, and \protocol{3Majority}.
In the \protocol{Voter} process (see~\cite{DBLP:journals/iandc/HassinP01, DBLP:journals/networks/NakataIY00, DBLP:conf/podc/CooperEOR12, DBLP:conf/icalp/BerenbrinkGKM16, DBLP:conf/soda/KanadeMS19}), in each round every agent adopts the opinion of a single, randomly chosen neighbor.
In the \protocol{TwoChoices} process, every agent samples two random neighbors and, if their opinions coincide, adopts their opinion.
In the \protocol{3Majority} dynamics, each agent samples three random neighbors and adopts the majority opinion among the samples, breaking ties uniformly at random.

\Textcite{DBLP:conf/podc/ElsasserFKMT17} consider the \protocol{TwoChoices} process on the complete graph.
For an additive bias of $\BigOmega{\sqrt{n \log n}}$ and $k = \BigO{n^\epsilon}$ (for a constant $\epsilon > 0$), they show that the initially largest opinion wins in $\BigO{k\cdot \log n}$ rounds.
The authors of~\cite{DBLP:conf/podc/GhaffariL18} consider the consensus problem for arbitrary initial configurations in the gossip model.
They show that both \protocol{TwoChoices} with $k = \ldauOmicron{\sqrt{n / {\log n}}}$ and \protocol{3Majority} with $k = \ldauOmicron{n^{1/3} / {\log n}}$ reach consensus in $\BigO{k\cdot \log n}$ rounds.
The latter result improves a result by~\textcite{DBLP:journals/dc/BecchettiCNPST17}.
Another result from~\cite{DBLP:conf/podc/GhaffariL18} shows that, for arbitrary $k$, \protocol{3Majority} reaches consensus \whp~in $\ldauOmicron{n^{2/3} \log^{3/2} n}$ rounds, improving upon a bound by \textcite{DBLP:conf/podc/BerenbrinkCEKMN17}.
The analysis is based on a coupling with the slow \protocol{Voter} process.
%This approach is inherently limited (due to the slow convergence of \protocol{Voter}) 
and cannot be applied to yield polylogarithmic runtime bounds~\cite{DBLP:conf/podc/BerenbrinkCEKMN17}.

\Textcite{DBLP:conf/soda/SchoenebeckY18} consider a general model for multi-sample consensus on graphs for $k = 2$ opinions.
%The authors analyze the time until all agents agree on one opinion for complete graphs and Erdős-Rényi graphs.
Another related process is the \protocol{MedianRule}~\cite{DBLP:conf/spaa/DoerrGMSS11}, which requires a total order on the $k$ opinions.
Here, in each round every agent adopts the median of its own and two sampled value.
This protocol reaches consensus \whp in $\BigO{\log k \log\log n + \log n}$ rounds.
Note that a total order on $k$ is a comparatively strong assumption and not required by any of the other protocols (including ours).

\section{Models and Results}%
\label{sec:model}

We consider a system of $n$ identical, anonymous agents.
Initially, each agent has one of $k$ possible opinions, which we represent as numbers from the set $\set{1, 2, \dots, k}$. We do not assume an order among the opinions. 
A \emph{configuration} describes the current state of the system and can be represented as an (unsorted) vector $\bm{x} = \intoo{x_i}_{i=1}^{k} \in \set{0, 1, \dots, n}^k$, where $x_i$ is the \emph{support} of opinion $i$, defined as the number of agents with opinion $i$.

The (additive) \emph{bias} of a configuration $\bm{x}$ is $x_{(1)} - x_{(2)}$, where $x_{(i)}$ denotes the \emph{support} of the $i$-th largest opinion (ties broken arbitrarily but consistently).
The \emph{multiplicative bias} is defined as $x_{(1)} / x_{(2)}$.
In the analysis we also use $x_{max} = x_{(1)}$ to denote the support of the largest opinion.
For any opinion $i$ with $x_i = x_{max}$ we say opinion $i$ \emph{provides} $x_{max}$.
Additionally, we call an opinion $i$ \emph{significant} if $x_{(i)} \geq x_{(1)} - \sDelta \cdot \sqrt{n \log n}$ (the constant $\sDelta$ is specified in \cref{def:Constant} and originates from our analysis in \cref{sec:analysis-population-model}).
An opinion that is not significant is called \emph{insignificant}.
We use $\Sigset{\bm{x}}$ to denote the set of opinions in configuration $\bm{x}$ that are significant.

\noindent\begin{minipage}[t]{0.5\textwidth-1em}
\begin{lstalgo}[H]{Population model.\label{alg:consensus}}
Actions performed when agents $(u, v)$ interact:

/* Decision Part: $\Clock u < 2\tau\log{n}$ */
if $\Clock u < 2\tau\log{n}$ and not $\DecisionFlag u$ then
	if $\Opinion u \neq \Opinion v$ then
		$\Undecided u \gets \ttrue$
	else
	    $\Undecided u \gets \tfalse$
    $\DecisionFlag u \gets \ttrue$

/* Boosting Part: {\normalfont $\Clock u \geq 2\tau\log{n}$} */
if $\Clock u \geq 2\tau\log{n}$ and $\Undecided u$ then
	if not $\Undecided v$ then
		$\Undecided u \gets \tfalse$
		$\Opinion u \gets \Opinion v$
	$\DecisionFlag u \gets \tfalse$

/* Leaderless Phase Clock \cite{DBLP:conf/soda/AlistarhAG18} */
if $\Clock u \leq_{(6\tau\log{n})} \Clock v$ then
    $\Clock u \gets (\Clock u + 1) \bmod 6\tau\log{n}$
else
    $\Clock v \gets (\Clock v + 1) \bmod 6\tau\log{n}$
\end{lstalgo}
\end{minipage}\hfill\begin{minipage}[t]{0.5\textwidth-1em}
\dcmlinenumbersfalse
\begin{lstalgo}[H]{Gossip model.\label{alg:consensus-gossip-model}}
Actions performed when agents $(u, v)$ interact:

/* Decision Part: {\normalfont $\Round u = 0$} */
if $\Round u = 0$ then
	if $\Opinion u \neq \Opinion v$ then
		$\Undecided u \gets \ttrue$
	else
	    $\Undecided u \gets \tfalse$


/* Boosting Part: {\normalfont $\Round u > 0$} */
if $\Round u > 0$ and $\Undecided u$ then
	if not $\Undecided v$ then
		$\Undecided u \gets \tfalse$
		$\Opinion u \gets \Opinion v$


/* Synchronization by counting modulo $(\T{BC} + 1)$ */
$\Round u \gets (\Round u + 1) \bmod (\T{BC} + 1)$



\end{lstalgo}
\end{minipage}

\subsection{Population Model}
\label{sec:protocol-population-model}
In the population model~\cite{DBLP:journals/dc/AngluinADFP06}, agents are finite state machines.
While the original model assumed a constant number of states per agent, more recent results allow a state space whose size depends on $n$.
In every time step, one pair of agents $(u, v)$ is chosen independently and uniformly at random to interact. 
During such an interaction, both agents update their states according to a common transition function $\delta\colon Q \times Q \rightarrow Q \times Q$, where $Q$ is the state space.
Note that we allow interactions of the form $(u,u)$, i.e., we allow agents to interact with themselves.
A population protocol is \emph{uniform} if the transition function does not depend on the number of agents $n$, although the number of states used during an execution may still be a function of $n$.
%We refer to~\cite{conf/podc/DotyE19} for a detailed discussion of the uniform population model.
The (parallel) runtime of a consensus population protocol is the number of interactions until all agents agree on the same opinion, divided by the number of agents $n$.
Note that any population protocol for plurality consensus requires at least $k$ states per agent (to store the agent's current opinion).

The \emph{state} of an agent $u$ is a tuple 
$\intoo{\Clock u, \Opinion u, \DecisionFlag u, \Undecided u}$ (see \cref{alg:consensus}).
$\Opinion u \in \set{1, 2, \dots, k}$ stores the current opinion of agent $u$.
The Boolean variable $\Undecided u$ indicates whether agent $u$ is currently undecided and $\DecisionFlag u$ indicates whether agent $u$ has already performed an interaction in the decision part. Both flags are initialized to \tfalse.

Our protocol uses the leaderless phase clock from \cite{DBLP:conf/soda/AlistarhAG18} (which runs on every agent) to divide the interactions into \emph{phases}, each consisting of $\BigO{n \log n}$ interactions.
The first part of a phase is called \emph{\Decision part} and the second part is called \emph{\Boosting part}.
In the \Decision part, every agent becomes undecided if and only if its first interaction partner has a different opinion.
In that case it sets its $\Undecided{}$ bit.
In the \Boosting part, every undecided agent adopts the opinion of a randomly sampled agent that is not undecided. 
This agent propagates its opinion to other undecided agents in subsequent steps.\footnote{For $k = \omega(n / \log n)$ it might happen with non-negligible probability that there is not a single  decided agent after the \Decision part. See \cref{sec:analysis-case3} for more details on how we resolve this case.}
The clock of agent $u$ uses the variable $\Clock u$ (initially 0) which can take values in $\set{0,\ldots, 6\tau \log n-1}$ for a suitably chosen constant $\tau$.
In each interaction, the smaller\footnote{
  Smaller w.r.t.\ the circular order modulo $m = 6\tau\log n$, defined as $a \leq_{(m)} b \equiv (a \leq b~ \text{\textsc{xor}} ~ \abs{a - b} > m/2)$.
} of the two values $\Clock{u}$ and $\Clock{v}$ is increased by one modulo $6\tau\log{n}$.
%The phase clocks divide the interactions into phases of length $\BigO{\log n}$.
For a polynomial number of interactions it guarantees \cite[][see Section~4]{DBLP:conf/soda/AlistarhAG18} that for any pair of agents $u$ and $v$ the distance\footnote{
Distance w.r.t.\ the circular order modulo $m = 6\tau\log{n}$, defined as $\abs{a-b}_{(m)} = \min\set{\abs{a-b}, m - \abs{a-b}}$.
}
between $\Clock u$ and $\Clock v$ is at most $\tau \log{n}$, and every agent participates in $\BigOmega{\log n}$ interactions in every phase. Hence, it cleanly separate the decision and boosting parts.
%such that, \whp, no agent becomes undecided (in the decision part) while other agents are still active in the preceding boosting part, and no undecided agent adopts an opinion (in the boosting part) while other agents are still active in the preceding decision part.
Our main result for the population model follows.
\begin{theorem}%
\label{thm:main-result-population-model}
Consider \cref{alg:consensus} on an initial configuration $\bm{x}$ with $k \leq n$ opinions.
The algorithm uses $k \cdot \ldauTheta{\log n}$ states per agent and has the following properties:
\begin{enumerate}[nosep]
\item All agents agree on a significant opinion in
$\BigO{\log^2{n}}$ parallel time, \whp.
\item Assume $\bm{x}$ has additive bias of $\sDelta \cdot \sqrt{n \cdot \log{n}}$ and multiplicative bias of $\alpha$.
\Whp the algorithm reaches a configuration in which all agents agree on the initial plurality opinion in time
  \begin{itemize}[nosep]
  \item $\BigO{\log{n} \cdot \log\log_{\alpha}{n}}$ if $k \leq \sqrt{n} / \log{n}$ and
  \item $\BigO{\log{n} \cdot (\log\log_{\alpha}{n} + \log\log{n})}$ if $k > \sqrt{n} / \log{n}$.
  \end{itemize}
\end{enumerate}
\end{theorem}

An alternative implementation of the protocol uses the junta-driven phase clock \textcite{DBLP:journals/dc/AngluinAE08,DBLP:conf/soda/GasieniecS18,DBLP:journals/corr/abs-1805-04586}. %instead of the leaderless phase clock from \cite{DBLP:conf/soda/AlistarhAG18}.
This reduces the number of required states to $k \cdot \ldauTheta{\log\log n}$ but results in a more complicated protocol.
Thus, we opted here for the slightly less efficient but more simple phase clock.
Using the alternative, more efficient implementation, we also get a stable protocol (i.e., eventually all agents agree on the same opinion). The analysis of  \cref{thm:main-result-population-model} is presented in \cref{sec:analysis-population-model}.

\subsection{Gossip Model}
\label{sec:gossip-model}

In the gossip model~\cite{DBLP:conf/stoc/Censor-HillelHKM12, DBLP:conf/soda/BecchettiCNPS15, DBLP:journals/sigact/BecchettiCN20} the agents interact simultaneously in synchronous rounds.
In each round every agent $u$ opens a communication channel to one other agent $v$ chosen independently and uniformly at random.
Agents have local memory (measured in bits) and may perform an arbitrary amount of local computations during each round.
Time is measured in parallel rounds.
Note that any gossip protocol for plurality consensus requires at least $\log k$ bits of local memory per agent (to store the agent's current opinion).

Our protocol is specified in \cref{alg:consensus-gossip-model}.
Because of the synchronous rounds, it is possible to synchronize agents into phases by simply counting the rounds.
We use the variable $\Rounds u$ (initially 0) for this counter and count modulo $\T{BC} + 1$.
Here, $\T{BC} = \BigO{\log{k} + \log\log {n}}$ denotes the time required for a broadcast process to succeed \whp when starting with $n/k$ many informed agents (see \cite{DBLP:conf/focs/KarpSSV00}).
Similar to our protocol from \cref{sec:protocol-population-model}, our gossip protocol alternates between a \emph{\Decision part} (a single round) and a \emph{\Boosting part} (multiple rounds).
In the \Decision part ($\Rounds u = 0$), every agent samples one other agent and becomes undecided if and only if the sample has a different opinion.
In the \Boosting part ($\Round u > 0 $) every undecided agent keeps randomly sampling other agents and adopts the first opinion it sees.
The variables $\Opinion{u}$ and $\Undecided{u}$ store the opinion of agent $u$ and whether it is undecided.
Our main result for the gossip model is as follows.
\begin{theorem}%
\label{thm:main-result-gossip-model}
Consider \cref{alg:consensus-gossip-model} on an initial configuration $\bm{x}$ with $k \leq n$ opinions.
The algorithm uses $\log k + \ldauTheta{\log\log{k} + \log\log\log{n}}$ bits of memory per agent and has the following properties:
\begin{enumerate}[nosep]
\item All agents agree on a significant opinion after $\BigO{(\log k+\log\log n) \cdot \log {n}}$ many rounds, \whp.
\item Assume $\bm{x}$ has additive bias of $\sDelta \cdot \sqrt{n \cdot \log{n}}$ and multiplicative bias of $\alpha$.
\Whp the algorithm reaches a configuration in which all agents agree on the initial plurality opinion in time
\begin{itemize}[nosep]
\item $\BigO{(\log{k} + \log \log n) \cdot \log\log_{\alpha}{n}}$ if $k \leq \sqrt{n} / \log{n}$ and
\item $\BigO{\log{k} \cdot (\log\log_{\alpha}{n} + \log\log{n})}$ if $k > \sqrt{n} / \log{n}$.
\end{itemize}
\end{enumerate}
\end{theorem}

Our bounds (essentially) match those from~\cite{DBLP:conf/icalp/BerenbrinkFGK16} but -- in contrast to~\cite{DBLP:conf/icalp/BerenbrinkFGK16} -- require no assumptions on $k$ or on the bias.
If $k$ is not known in advance, one can initialize $\T{BC} = \BigTheta{\log{n}}$~\cite{DBLP:conf/focs/KarpSSV00}.
In that case the algorithm uses $\log k + \ldauTheta{\log{n}}$ bits of local memory and follows the runtime bounds stated in  \cref{thm:main-result-population-model}.
It is also possible to replace the simple counter modulo $(\T{BC}+1)$ by a uniform clock.
The main idea is to use the approach described in \cite{DBLP:conf/soda/AlistarhAEGR17} to sample an approximation of $\Theta(\log{n})$ for $\T{BC}$.
This solves an open question from~\cite{DBLP:journals/dc/BecchettiCNPST17}. See \cref{sec:uniform-protocol} for the details.
Note that the definition of this uniform protocol does not depend on $n$ or $k$, but the number of states used by the agents do depend on $n$.
It is a more general open question whether plurality consensus can be solved in polylogarithmic time using only constantly many additional states.
The proof of  \cref{thm:main-result-gossip-model} 
follows along the same lines as the proof for the population model and can also be found in the appendix.

\section{Analysis for the Population Model}
\label{sec:analysis-population-model}

In this section we analyze our consensus protocol for the population model.
In \cref{sec:concentration-population-model} we start with fundamental concentration results that describe the evolution of opinions throughout one fixed phase.
Then, in \cref{sec:analysis-case1,sec:analysis-case2,sec:analysis-case3} we consider three different cases, depending on the initial number of opinions $k$.
The proof of the first part of \cref{thm:main-result-population-model} follows immediately from three propositions, \cref{pro:case-1} in \cref{sec:analysis-case1}, \cref{pro:case-2} in \cref{sec:analysis-case2}, and \cref{pro:case-3} in \cref{sec:analysis-case3}.
The proof of the second part of the theorem closely resembles the proofs in \cite{DBLP:conf/podc/GhaffariP16a,DBLP:conf/icalp/BerenbrinkFGK16}.
For completeness, a detailed proof of  \cref{thm:main-result-population-model} is given in \cref{sec:details-proof-of-main-result}.

In our analysis, we assume that the phase clocks properly separate the boosting and decision parts of the considered phases.
This follows from \cite{DBLP:conf/soda/AlistarhAG18,DBLP:journals/rsa/PeresTW15}, where it is shown that for a polynomial number of phases and for any pair of agents $u$ and $v$ the distance between $\Clock u$ and $\Clock v$ w.r.t.\ the circular order modulo $6\tau\log{n}$ is less than $\tau \log{n}$, \whp.
The choice of $\tau$ also ensures that every undecided agent is able to adopt an opinion in the boosting part of a phase, \whp.

The strict phase synchronization allows us to define a series of random vectors $\mathbfcal{X}=\intoo{\X(t)}_{t \in \N}$ that describe the configurations at the beginning of phase $t$ where the $i$-th entry $X_i(t)$ is the number of agents with opinion $i$.
For the analysis we also define a series of random vectors $\mathbfcal{Y} = \intoo{\Y(t)}_{t \in \N}$ where $Y_i(t)$ is the number of decided agents with opinion $i$ at the beginning of the \emph{boosting} part of phase $t$.
Finally, the series $\cX_{max} = \intoo{X_{max}(t)}_{t \in \N}$ describes the size of the support of the largest opinion.
In general, we use bold font to denote vectors, non-bold font to denote vector components, and capital letters for random variables.
When we fix the value of a random variable at the beginning of a phase $t$, we use lowercase letters, e.g., $\X(t) = \x(t)$.
When it is clear from the context, we omit the parameter $t$ in the proofs.

The following observation shows that the opinion distribution after the decision part can be described by a binomial distribution. Note that $\norm{\Y(t)}_1$ denotes the number of decided agents at the beginning of the $t$-th \Boosting part.

\begin{observation}[Decision Part]
\label{obs:bin}
\label{obs:decided}
Assume $\X(t) = \x(t)$ is fixed and let $\Y(t)$ be the configuration at the beginning of the \Boosting part of phase $t$.
Then, for $1 \leq i \leq k$, the $Y_i(t)$ have an independent binomial distribution with $Y_i(t) \sim \Bin(x_i(t), x_i(t) / n)$. Additionally, for $\Decided(t): = \Ex{\norm{\Y(t)}_1} = \sum_{i=1}^{k} x_i(t)^2 / n$ we have $n/k \leq \Decided(t) \leq x_{max}(t)$.
\end{observation}

The opinion distribution after the boosting part can be modeled by a so-called \emph{Pólya-Eggenberger distribution}.
The Pólya-Eggenberger process is a simple urn process that runs in multiple steps.
Initially the urn contains $a$ red and $b$ blue balls, where $a,b \in \mathbb{N}_0$.
In each step of the process, one ball is drawn from the urn uniformly at random, its color is observed, and the ball is then returned together with one additional ball of the same color.
The corresponding \emph{Pólya-Eggenberger distribution} $\PE(a,b,m)$ describes the number of \emph{total} red balls that are contained in the urn after $m$ steps.
In order to bound $X_i(t+1)$ we use tail inequalities (Theorem 1 and Theorem 47) shown in the full version of \cite{DBLP:conf/podc/BankhamerEKK20}. For convenience, we state these bounds in \cref{sec:polya-eggenberger}.

\begin{observation}[Boosting Part]
\label{obs:polya}
Assume $\Y(t) = \y(t)$ is fixed and $\norm{\y(t)}_{1} \geq 1$. Let $\X(t+1)$ be the configuration at the beginning of the \Decision part of phase $t+1$.
Then, for $1 \leq i \leq k$, $X_{i}(t+1)$ has Pólya-Eggenberger distribution $X_{i}(t+1) \sim \PE(y_{i}(t), \norm{\y(t)}_{1} - y_{i}(t), n - \norm{\y(t)}_{1})$ with $\Ex{X_i(t+1)} = y_i(t) \cdot (n  / \norm{\y(t)}_{1})$.
\end{observation}
\begin{proof}
This follows from an easy coupling of the boosting part with the Pólya-Eggenberger process defined as follows.
Let $\ell_0 \geq 1$ be the number of decided agents at the beginning of the boosting part.
For $\ell_0< i\le n$ the process picks an arbitrary one of the undecided agents.
This agent chooses one of the $\ell_{i-1}$ decided agents uniformly at random and adopts its opinion, resulting in $\ell_i:=\ell_{i-1}+1$.
The coupling of our process with this process is now straight-forward by discarding all interactions which do not change the number of decided agents.
\end{proof}

Finally we introduce some important constants that we use throughout our analysis. 
\vspace{-0.5ex}
\begin{definition}
\label{def:Constant}
We define $\sDelta := (160\cdot c_w)^2+(148\cdot \polyaConstLarge)^2$, $c_w := 8\sqrt{1+2/\concConstEps}$, $\concConstEps := \polyaConstSmall / 192$ and $c_k := 4 \cdot (2625 + c_p)^2$. The constants $1>\polyaConstSmall>0$ and $\polyaConstLarge>1$ originate from the Pólya-Eggenberger concentration results of \cite{DBLP:conf/podc/BankhamerEKK20}. %(stated in  \cref{thm:polya_bound_total_balls} and  \cref{thm:polya_bound_small_support}). 
\end{definition}
\begin{definition}
Opinion $i$ in configuration $\x$ is called \emph{super-weak} iff $x_i \leq c_w \cdot \sqrt{n \log n}$, \emph{weak} iff $c_w \cdot \sqrt{n \log n} < x_i < 0.9 \cdot x_{max}$, and \emph{strong} iff $x_i \geq 0.9 \cdot x_{max}$.
\end{definition}

\subsection{Analysis of a Single Phase}
\label{sec:concentration-population-model}
\label{sec:analysis-maximum}
In this section we analyze the evolution of opinions throughout some fixed phase $t$.
The following lemma gives Chernoff-like guarantees for a large range of deviations and opinion sizes. 
\begin{restatable}{lemma}{lemOnePhaseConc}
\label{lemma:one-phase-conc}
Fix $\X(t) = \x(t)$ and an opinion $i$ with support $x_i(t)$. Furthermore, let $\Decided = \sum_{j=1}^{k} x_{j}(t)^2/n$.
Then, for any $0 < \delta  < {x_i(t)}/{\sqrt{n}}$ and a suitable small constant $\concConstEps>0$
\[
\Pr\left[ X_i(t+1) < \frac{x_i(t)^2}{\Decided} - \frac{x_i(t)}{\Decided}\sqrt{n} \delta\right] \le 7e^{-\concConstEps \delta^2} \text{ and } \Pr\left[ X_i(t+1) > \frac{x_i(t)^2}{\Decided} + \frac{x_i(t)}{\Decided}\sqrt{n} \delta \right] \le 7e^{-\concConstEps \delta^2}.
\]
\end{restatable}
\begin{proof}[Proof Sketch]
First we bound $Y_i(t)$, the support of opinion $i$ at the end of the decision part of phase $t$. By \cref{obs:decided} $Y_j(t)$ follows a binomial distribution for every opinion. This allows us to bound $Y_i(t)$ as well as $\norm{\Y(t)}_1$ with Chernoff bounds.
Then we continue to analyze the boosting part of phase $t$. Following \cref{obs:polya} we can model  $X_i(t+1)$ via a Pólya-Eggenberger distribution as a function of $Y_i(t)$ and $\norm{\Y(t)}_1$. Conditioned on the outcome of the decision part we apply a concentration bound for the Pólya-Eggenberger distribution to bound $X_i(t+1)$ (see \cref{thm:polya_bound_total_balls}).
\end{proof}

\Cref{lemma:one-phase-conc} does not give high probability for opinions with small support. In \cref{sec:details-concentration-population-model}  we provide a coarse bound for this regime (\cref{lemma:weak-phase-conc}).
 
Recall that opinion $j$ is insignificant in configuration $\x$ if $x_j < x_{max} - \sDelta \sqrt{n \log n}$, and $\Sigset{\x}$ is the set of significant opinions in configuration $\x$.
Note that in our setting any significant opinion can win, and the largest opinion (which provides $\cX_{max}$) can change over time.
The next lemma shows that if an opinion becomes insignificant it cannot become significant again \whp.

\begin{restatable}{lemma}{lemInitialBias}\label{lem:initial-bias}
Fix $\X(t) = \x(t)$. Then, $\Sigset{\X(t+1)} \subseteq \Sigset{\x(t)}$ \whp.
\end{restatable}

The following lemma is key to the analysis of the most challenging case where $\sqrt{n}/\log{n} \leq k \leq \sqrt{n}/c_k$.
The proof is structured into two cases.
If almost all opinions have a support close to that of the largest opinion, then none of them grow much \emph{in expectation}.
Interestingly, the dynamic of the process allows that at least one of these opinions increases its support significantly, and at the end will win the battle between these large opinions in polylogarithmic time.
To prove this, we develop a novel anti-concentration result which models exactly this phenomenon. 

\begin{restatable}{lemma}{lemKtoSqrtlog}
\label{lem:k-to-sqrtlog}
Fix $\X(t) = \x(t)$. Let $c=\concConstEps / 625 > 0$ be a small constant.
\begin{enumerate}[nosep]
\item If $\sqrt{n} \leq x_{max}(t) \leq \sqrt{n \log n}$ then
\[ \Pr[ X_{max}(t+1) >  (1 + 1/60) \cdot  x_{max}(t)]\ge 1 - 7\exp\left(- c\cdot (x_{max}(t))^2 / n \right) \].
\item
If $\sqrt{n \log n} <  x_{max}(t) < \sqrt{n}\log^{3/2} n$ then
\[
\Pr[ X_{max}(t+1) > x_{max}(t) + (1/60) \cdot \sqrt{n \log n}] \ge 1 - 7\exp\left({- c\cdot \log n}\right)
\].
\end{enumerate}
\end{restatable}
\begin{proof}[Proof Sketch]
We only consider the first statement of \cref{lem:k-to-sqrtlog}.
The proof for the second statement is similar.
First consider the easy case where we have a small number of large opinions that battle to win the majority.
Let $i$ be the opinion with the largest support at the beginning of some phase $t$. In this case we track this opinion throughout the phase and use \cref{lemma:one-phase-conc} to show that its support grows by a factor of at least $1+1/60$.
Note that at the end of the phase another opinion may provide $X_{max}(t+1)$. Still, in both cases $X_{max}(t+1)\ge (1+1/60) \cdot x_{max}(t)$.

\medskip

In the second case we have a set $\cL$ of more than $\sqrt{n} / \log^4 n$ \emph{large} opinions that have support at least $0.9 \cdot x_{max}(t)$.
First, we consider the decision part of phase $t$.
Recall that $Y_i(t)$ is defined as the number of decided agents with opinion $i$ at the beginning of the boosting part of phase $t$, and $Y_i(t) \sim \Bin(x_i(t),x_i(t) / n)$ (\cref{obs:decided}). Note that for all $i \in \cL$ we have $\Ex{Y_i(t)} > 0.81 \cdot x_{max}^2(t) / n$ and $\sqrt{\Var{Y_i(t)}} \approx \sqrt{\Ex{Y_i(t)}}$.
First we will show that \whp there exists an opinion $i \in \cL$  $Y_i(t) > \Ex{Y_i(t)} + (1/2) \cdot \sqrt{\log n \cdot  \Var{Y_i(t)}}$.
To prove this we use an anti-concentration result for the binomial distribution (stated in \cref{lem:reverse-chernoff-simple} in \cref{sec:auxiliary-results}).

Next we  track opinion $i$ through the boosting part of phase $t$.
Recall that $X_j(t)$ is defined as the number of decided agents with opinion $j$ at the beginning of phase $t$. $X_j(t+1)$ follows the P\'olya-Eggenberger distribution with $\Ex{X_j(t+1)} = Y_j(t) \cdot (n / \norm{\Y(t)}_1)$ (\cref{obs:polya}). We estimate
\[
    \Ex{\norm{\Y(t)}_1} := \sum_{i=1}^{k} \frac{x_i^2(t)}{n} \leq \sum_{i=1}^{k} \frac{x_i(t) \cdot x_{max} (t)}{n} = x_{max} (t) \cdot \sum_{i=1}^{k} \frac{x_i (t)}{n} = x_{max} (t).
\]
Using Chernoff bounds, we show that $\norm{\Y(t)}_1 < x_{max} (1+ o(1))$ \whp.
Since $Y_i(t)\ge \Ex{Y_i(t)} + (1/2) \cdot \sqrt{\log n \cdot  \Var{Y_i(t)}}$ we have
\begin{align*}
    \Ex{X_i(t+1)} & \ge \left( \Ex{Y_i(t)} + \frac{1}{2} \cdot \sqrt{\log n \cdot  \Var{Y_i(t)}} \right) \cdot  \frac{n}{\norm{\Y(t)}_1} \\
    &\geq \left( 0.81 \cdot \frac{x_{max}^2(t)}{n} + \frac{1}{2}  \sqrt{ \log n \cdot 0.81 \cdot \frac{x_{max}^2(t)}{n}} \right) \cdot \frac{n}{x_{max}(t) (1 + o(1))} \\
   &\geq \left( 0.81 \cdot x_{max} (t) + \frac{1}{2} \sqrt{0.81 n \log n} \right) \frac{1}{1 + o(1)}  %\left(\frac{\tau^2}{n} +  \frac{1}{2} \sqrt{\log n}\frac{\tau}{\sqrt{n}} \right) \cdot \frac{n}{x_{max}(t)(1 + o(1))} = \frac{0.81 \cdot  x_{max}(t) +  (1/2) 0.9 \sqrt{n \log n}}{1+ o(1)} >
    \geq 1.2 \cdot x_{max}(t)
\end{align*}
The last step follows from $x_{max}(t) \le \sqrt{n \log n}$. While this is just an expected value, we employ a Pólya-Eggenberger concentration result (see Theorem 1.1 of \cite{DBLP:conf/podc/BankhamerEKK20}) in our detailed analysis and show that a similar bound indeed holds for $X_j(t+1)$ with probability $1 - \exp(-\Omega(x_{max}(t)^2 / n))$.
\end{proof}

\subsection{Consensus for $k \leq \sqrt{n} / \log n$}%
\label{sec:analysis-case1}
In this case the analysis uses the general approach from \cite{DBLP:conf/podc/GhaffariL18} where the authors analyze the majority process for $k$ opinions.
Opinions are classified as \emph{strong}, \emph{weak} or \emph{super-weak}, depending on their support.
The authors of \cite{DBLP:conf/podc/GhaffariL18} divide time into epochs of length $\BigO{ (\frac{5}{6})^i \cdot k \log n}$ and show that at the end of the $i$-th epoch the support of the largest opinion grows by a constant factor and the fraction of \emph{non}-super-weak opinions decreases by a constant factor.
As super-weak opinions remain super-weak this implies that eventually consensus is reached in time $\BigO{k \log n}$.

Our approach is different and exploits the properties of the \emph{undecided state dynamics (USD)} which, for example, allows us to avoid both epochs of different length as well as a total runtime that is linear in $k$.
Throughout our analysis we consider all pairs of opinions.
We show that during $\BigO{\log n}$ phases at least one opinion in each pair becomes weak and, eventually, super-weak.
If both opinions in a pair are initially strong we apply (similar to \cite{DBLP:conf/podc/GhaffariL18}) the drift result of \cite{DBLP:conf/spaa/DoerrGMSS11} to show that their support drifts apart.
Hence, only one of the strong opinion prevails, which will be adopted by every agent within a constant number of additional phases. 

\enlargethispage{\baselineskip}
\begin{restatable}{proposition}{proCaseOne}
\label{pro:case-1}
Assume $\X(t)$ is a configuration with $k < \sqrt{n} / \log n $ opinions. 
Then, after $\BigO{\log{n}}$ phases, all agents agree on some opinion $i \in \Sigset{X(t)}$,  \whp.
\end{restatable}

\subsection{Consensus for $\sqrt{n}/\log{n} <k \leq \sqrt{n}/c_k$}
\label{sec:analysis-case2}
This is the most interesting part of our analysis. 
In \cite{DBLP:conf/podc/GhaffariL18} this range of $k$ was analyzed via a coupling with the \protocol{Voter} process, which results in a runtime polynomial in $n$.
This case contains many configurations which are hard to handle, such as the following one: $\sqrt{n} / {\log\log n}$ opinions have a support of $\sqrt{n} \log\log{n}$ each.
According to our definition of strong, weak and super-weak, all the opinions are strong and super-weak at the same time. 
The support of the opinions is too small to be tracked via concentration bounds (which typically requires a support of $\ldauOmega{\sqrt{n \log n}}$ -- see the constraint on $\delta$ in \cref{lemma:one-phase-conc}) but at the same time the support of all opinions is large enough to prevent them from vanishing fast (which typically happens for opinions with support $\ldauOmicron{\sqrt{n}}$).

Another problem we have to deal with in our analysis is that the opinions which provide the maximum support $\cX_{max}$ can change over time.
This prevents us from using the approach from \cite{DBLP:conf/icalp/BerenbrinkFGK16,DBLP:conf/podc/GhaffariP16a} who show that the support of the (unique, unchanging) maximum opinion (or the bias) grows over time.
In our case, there might be many opinions with a support close to $X_{max}(t)$ in a certain phase~$t$, preventing us from tracking the growth of one fixed opinion.
Moreover, it is not clear that in such situations $\cX_{max}$ increases sufficiently in every phase with high probability.

Our novel proof strategy is to show that, due to the variance of the process, there is (at least) one opinion which gains a support of size $\ldauOmega{\sqrt{n} \cdot \log^{3/2} n}$ within $\BigO{\log n}$ phases (\cref{lem:Drift-second-case-k}). 
Since we cannot track one fixed opinion (as any one of them might fall behind or even die out), we analyze how $\cX_{max}$ changes over time via a drift result from~\cite{DBLP:conf/spaa/DoerrGMSS11}.
We show that in a certain period of time a sequence of \emph{successful} phases exists, in each of which $\cX_{max}$ grows by a constant factor, and one opinion eventually wins this \emph{battle of candidates}.

As soon as $X_{max}(t) = \BigOmega{\sqrt{n} \log^{3/2} n}$ for some phase $t$ we distinguish two cases (\cref{lem:log32-to-end}):
Either we have only a couple of opinions whose support is close to $X_{max}(t)$.
Then these opinions will increase their support by a constant factor in every phase.
Otherwise, if sufficiently many opinions have support close to $X_{max}(t)$, we use a counting argument to show that most opinions are small (have support of less than $\sqrt{n} \log n$).
In this case the large opinions will take over and cause the small opinions to vanish in a constant number of phases.
In both cases this brings us to the situation considered in \cref{sec:analysis-case1}.

\begin{proposition}
\label{pro:case-2}
Assume $\X(t)$ is a configuration with $ \sqrt{n} / \log n < k < \sqrt{n} / c_k$ opinions. 
Then, after $t'=\BigO{\log n}$ phases, the process reaches a configuration $\X(t+t')$
such that the following holds \whp: $\X(t+t')$ has at most $\sqrt{n}/\log{n}$ opinions,
$\vert \Sigset{\X(t+t')} \vert >0$, and
$\Sigset{\X(t+t')} \subseteq \Sigset{\X(t)}$.
\end{proposition}

\begin{proof}
We only show that the number of opinions is reduced to $\sqrt{n}/\log n$. The statement about the set of significant opinions follows the proof of \cref{pro:case-1} in \cref{sec:details-analysis-case1}.
In \cref{lem:Drift-second-case-k} we show that \whp there is a phase $\tilde{t}\in [t, t+O(\log n)]$ such that $\X(t + \tilde{t})$ contains an opinion with support at least $\sqrt{n} \cdot \log^{3/2} n$.
In \cref{lem:log32-to-end} we show that at time $t+\tilde{t}+ 3i$, for all $1\le i \le O(\log n)$, one of the following statements holds \whp:
(1) $X_{max}(t + \tilde{t}+3i) > (5/4)^i \cdot X_{max}(t + \tilde{t})$, or
(2) the number of opinions in $\X(t + \tilde{t}+3i)$ is smaller than $\sqrt{n} / \log n$.
From this it follows that there exists a $t'\in \BigO{\log n}$ such that $X_{max}(t+t')\ge n$ (in which case we are done) or $\X(t+t')$ has less than $\sqrt{n}/\log{n}$ opinions. 
Note that as soon as the second condition holds, we are in the situation analyzed in \cref{sec:analysis-case1}. 
\end{proof}

\begin{restatable}{lemma}{lemDriftSecondCaseK}
\label{lem:Drift-second-case-k}
Assume $\X(t)$ is a configuration with $ \sqrt{n} / \log n < k < \sqrt{n} / c_k$ opinions.
\Whp $X_{max}(t+t') \geq \sqrt{n}  \cdot \log^{3/2} n$ for some $t' = \BigO{\log n}$.
\end{restatable}
\begin{proof}[Proof Sketch]
We call a phase $t$ successful if and only if $X_{max}(t+1) \geq (61/60) \cdot X_{max}(t)$.
From \cref{lem:k-to-sqrtlog} it follows that, as long as $X_{max}(t) \leq \sqrt{n \log n}$, the probability that a phase is not successful is exponentially small in $X_{max}^2(t) / n$.
This enables us to use the drift result of \cite{DBLP:conf/spaa/DoerrGMSS11}, and we obtain that in $\BigO{\log n}$ many phases there exists \whp a sequence of $\BigO{\log\log n}$ consecutive successful phases.
Since after any unsuccessful phase $X_{max}(t)$ is still larger than $n/k = \BigOmega{\sqrt{n}}$, the maximum support reaches a size of $\sqrt{n \log n}$ within $\BigO{\log n}$ phases.
From there on we can repeatedly apply the second statement of \cref{lem:k-to-sqrtlog}, showing that the support of the largest opinion increases additively by $\BigOmega{\sqrt{n\log n}}$ in every phase \whp until it reaches $\sqrt{n} \log^{3/2} n$.
\end{proof}

\begin{restatable}{lemma}{lemLogToEnd}
\label{lem:log32-to-end}
	Fix $\X(t) =\x(t)$. Assume that $x_{max}(t) \geq \sqrt{n} \log^{3/2} n$ and in configuration $\x(t)$ we have $\sqrt{n} / \log n < k < \sqrt{n} / c_k$ opinions.
	At the beginning of the \Decision part of phase $t+3$ \whp at least one of the following statements holds:
	\begin{itemize}[nosep]
	    \item $X_{max}(t+3) > (5/4) \cdot x_{max}(t)$ 
	    \item at most $\sqrt{n}/ \log n$ opinions have non-zero support.
	\end{itemize}
\end{restatable}
\begin{proof}[Proof Sketch]
Recall that $ \Decided(t) = \Ex{\norm{\Y(t)}_1}$ is the expected number of decided agents at the beginning of the boosting part of phase $t$. Let us consider two cases:

The first case is $\Decided(t) < x_{max}(t) /2$, which covers the situation that there are only a couple of opinions whose support is close to $X_{max}(t)$.
In this case we track the currently largest opinion.
By \cref{lemma:one-phase-conc} we have that $X_{max}(t+1)$ is tightly concentrated around $x_{max}(t)^2 / \Decided(t) \geq 2x_{max}(t)$, implying that the opinion with maximum support grows by a constant factor.

The second case considers $\Decided(t) \geq x_{max}(t) /2$, which covers the situation that sufficiently many opinions have support close to $X_{max}(t)$.
Note that, by a simple counting argument, at most $\sqrt{n} / \log n$ opinions have a support of $\geq \sqrt{n} \log n$.
That means there must be many small opinions with support $\leq \sqrt{n} \log n$.
Using the assumption $\Decided(t) \geq x_{max}(t) /2$ we show that (\whp) all these many small opinions vanish in a constant number of phases.
To this end fix such a small opinion $j$ and note that $\Ex{X_j(t+1)} \approx x_j(t)^2 / \Decided$ and $x_j(t)^2 / \Decided(t) \leq x_j(t)\cdot (2x_j(t) / x_{max}(t)) \leq \sqrt{n} \log n \cdot (2/\sqrt{\log n})$.
We show that $X_j(t+1) = \BigO{\sqrt{n \log n}}$ \whp, implying that $j$ shrinks by a $\sqrt{\log n}$ factor. Applying a very similar argument two more times, opinion $j$ is removed with probability $1-o(1)$.
This, in turn, implies that \whp most of these small opinions disappear.
Tightening the arguments above and adapting the analysis accordingly, we show that finally at most $\sqrt{n}/\log n$ opinions remain \whp.  
\end{proof}

\subsection{Consensus for $k > \sqrt{n}/c_k$}
\label{sec:analysis-case3}

In this remaining part of our analysis we first consider a special case:
for $k = \omega(n / \log n)$ it might happen that $|\Y(t)| = 0$, i.e., there is not a single decided agent after the \Decision part of phase $t$. 
In this case, however, no agent changes its opinion throughout the \Boosting part of this phase (see \cref{alg:consensus}). Hence, the opinion distribution does not change throughout phase $t$. 
It is easy to see that the probability for $\Y(t) = 0$ is maximized if $k=n$.
However, since we allow agents to interact with themselves, we can bound this probability by $(1-1/n)^{n} < 1/e$.
Therefore, in a sequence of $c \cdot \log n$ phases for some large enough constant $c$ there will be at least one phase in which at least one agent (and at most $\BigO{\log n}$ agents) remain decided in the \Decision part, \whp.
The remaining analysis of this case is now as follows. As long as $k\geq \sqrt{n}/c_k$ at least $k - \sqrt{n}/c_k$ opinions must have support smaller than $c_k \sqrt{n}$.
Each of these opinions vanishes with constant probability throughout a single phase (provided $|\Y(t)| \geq 1$).
Hence, after a total of $\BigO{\log n}$ phases we are back to the case $k \leq \sqrt{n}/c_k$.

\begin{restatable}{proposition}{proCaseThree}
\label{pro:case-3}
Assume $\X(t)$ is a configuration with $k > \sqrt{n} / c_k$ opinions. Then,
after $t'=\BigO{\log n}$ phases, the process reaches  a configuration $\X(t+t')$
such that the following holds \whp: $\X(t+t')$ has at most $\sqrt{n}/c_k$ opinions,
$\vert \Sigset{X(t+t')} \vert >0$, and
$\Sigset{X(t+t')} \subseteq \Sigset{X(t)} $.
\end{restatable}

\section{Conclusion}\label{sec:concl}
In this paper we analyze a synchronized version of the undecided state dynamics, both for the population and the gossip model. Our result holds for  up to $n$ initial opinions. For both models we show a consensus time of $\BigO{\log^2 n}$. Furthermore we show that the agents agree on the majority opinion if such an opinion exists; otherwise they agree on a significant opinion having a relative large support. 
 One open question is to see if our results are tight. The main reason for a running time 
of $\BigO{\log^2 n}$ is that our algorithm needs $\BigO{\log n}$ phases of length $\BigO{\log n}$ for breaking the ties in the case of several opinions with roughly the same support.
It might be possible to work with a phase length as a function of $k$ resulting in a refined running time of $\BigO{\log k \log n}$.
Moreover, it may be possible to interleave consecutive phases in order to reduce the running time even further.
For the gossip model it is known that USD does solve plurality consensus but is much slower than the synchronized version.
It would be interesting to show a similar result for the population model.
Another open question is to bound the expected running time of the USD.

\printbibliography

\clearpage
\allowdisplaybreaks
\def\zerodisplayskips{}
\normalsize
\section*{Appendix}
\appendix
\section{Technical Details for the Population Model}
\label{sec:details-population-model}
In this appendix we give the omitted proofs from \cref{sec:analysis-population-model}.
Note that we omit the parameter $t$ in the proofs when we fix a variable such that $\x(t) = \x$.
\subsection{Analysis of a Single Phase}
\label{sec:details-concentration-population-model}

\lemOnePhaseConc*

\begin{proof}[Proof]
We focus on the upper bound and consider some fixed opinion $i$ throughout phase $t$.
At first we analyze the evolution of $i$ throughout \Decision part $t$ and consider $Y_i(t)$ and $\norm{\config{Y}(t)}$.
Recall (see \cref{obs:polya}) that $Y_i(t) \sim \Bin(x_i, x_i /n)$ and $\norm{\config{Y}(t)}_1 = \sum_{j=1}^k Y_j(t)$.
We define the event $\mathcal{E}_i$ as follows
	\[
	    \mathcal{E}_i = \left\{ Y_i(t) < \frac{x_i^2}{n}\cdot \left(1 + \frac{\delta \sqrt{n}}{8\cdot x_i}\right) \mbox{ and } \norm{\config{Y}(t)}_1 > \Decided\cdot \left(1 - \frac{\delta \sqrt{n}}{8 \cdot x_i} \right) \right\}.
	\]
	First we bound the probability of complement of $\mathcal{E}_i$. To do so, we apply Chernoff Bounds (\cref{lemma:chernoff_poisson_trials}) to $Y_i(t)$ and $\norm{\config{Y}(t)}_1$. Hence, for $\delta'= \frac{\delta \cdot\sqrt{n}}{8 \cdot x_i}<1$ and using $\Decided \ge \frac{x_i^2}{n}$
\[
\Pr\left[Y_i(t)\ge \frac{x_i^2}{n}\cdot\left(1+\delta'\right)\right]\le
\exp\left(-\frac{x_i^2\cdot {\delta'}^2}{3\cdot n}\right) \le \exp\left(-\frac{x_i^2\cdot n \cdot \delta^2}{192\cdot x_i^2\cdot n}\right) \le \exp\left(-\frac{{\delta}^2}{192}\right)
\]

\[
\Pr\left[\norm{\config{Y}(t)}_1 \le \Decided\cdot\left(1-\delta'\right)\right] \le \exp\left(-\frac{{\delta'}^2 \cdot \Decided}{2}\right) \le \exp\left(-\frac{x_i^2\cdot n \cdot \delta^2}{128\cdot x_i^2\cdot n}\right) \le \exp\left(-\frac{\delta^2}{192}\right)
\]
 An application of the union bound yields
	\begin{equation}\label{eq:3}
	\Pr\left[\bar{\mathcal{E}_i}\right]\le  2\exp\left(-\frac{{\delta}^2}{192} \right).
	\end{equation}

Now we deal with the outcome of the \Boosting part conditioned on the event $\mathcal{E}_i$.
We fix $Y_i(t) = y_i$ and define $d:=\norm{\config{Y}(t)}_1 = \sum_{j=1}^k y_j$.
As mentioned in \cref{obs:polya} we model $X_i(t+1) \sim \PE(y_i,d - y_i , n - d)$. Applying the  tail bound for the Pólya Eggenberger distribution from
\cref{thm:polya_bound_total_balls}  we get for $0 < \frac{\delta}{8} < \sqrt{y_i}$ and some constant $1>\polyaConstSmall>0$ that
\begin{equation}
\label{eq:4}
	\Pr \left[ X_i(t+1) > \frac{y_i}{d}\cdot n + \frac{\sqrt{y_i}}{d} \cdot n \cdot \frac{\delta}{8} ~\Big|~ \mathcal{E}_i \right] < 4 \exp\left(-\polyaConstSmall \cdot \frac{\delta^2}{64}\right).
\end{equation}
Since $y_i < \frac{x_i^2}{n}\cdot \left(1 + \frac{\delta \sqrt{n}}{8 x_i}\right) \mbox{ and } d > \Decided\cdot \left(1 - \frac{\delta\sqrt{n}}{8 x_i}\right)$ we get
\begin{align*}\label{eq:5}
     \frac{y_i}{d}\cdot n + \frac{\sqrt{y_i}}{d} \cdot n \cdot \frac{\delta}{8}
     &< \frac{x_i^2}{\Decided} \cdot \frac{\left(1 + \frac{\delta \cdot\sqrt{n}}{8 \cdot x_i} \right)}{ \left(1 - \frac{\delta \cdot \sqrt{n}}{8 \cdot x_i}\right)} +  \frac{\sqrt{\frac{x_i^2}{n}\cdot \left(1 + \frac{\delta \cdot\sqrt{n}}{8 \cdot x_i} \right)}}{\Decided\cdot \left(1 - \frac{\delta \cdot \sqrt{n}}{8 \cdot x_i}\right)}\cdot n \cdot \frac{\delta}{8}\\
     &=  \frac{x_i^2}{\Decided} \cdot \frac{\left(1 + \frac{\delta \cdot \sqrt{n}}{8 \cdot x_i} \right)}{ \left(1 - \frac{\delta \cdot \sqrt{n}}{8 \cdot x_i}\right)} +  \frac{ x_i}{\Decided}\cdot \frac{\sqrt{ \left(1 + \frac{\delta\cdot\sqrt{n}}{8 \cdot x_i} \right)}}{ \left(1 - \frac{\delta \cdot \sqrt{n}}{8 \cdot x_i}\right)}\cdot \sqrt{n} \cdot \frac{\delta}{8}\\
     &\overset{(*)}{<} \frac{x_i^2}{\Decided}\cdot \left(1 +3\cdot \frac{\delta\cdot\sqrt{n}}{8 \cdot x_i} \right) +  \frac{x_i\cdot \sqrt{n}\cdot \delta}{8 \Decided} \cdot\left(1 +3\cdot \frac{\delta\cdot\sqrt{n}}{8 \cdot x_i} \right) \\
     &< \frac{x_i^2}{\Decided}\cdot \left(1 +3\cdot \frac{\delta\cdot\sqrt{n}}{8 \cdot x_i} \right) +  2\cdot \frac{x_i\cdot \sqrt{n}\cdot \delta}{8 \cdot \Decided} \\
     &= \frac{x_i^2}{\Decided} +  5\cdot \frac{x_i\cdot \sqrt{n}\cdot \delta}{8 \cdot \Decided} <  \frac{x_i^2}{\Decided} + \frac{x_i\cdot \sqrt{n}\cdot \delta}{\Decided}.
     \end{align*}
     In (*) we apply the inequality $(1+a)/(1-a)\le 1+3a$ which holds for all $ a\le 1/3$.
Hence, a combination of this and Inequality (\ref{eq:4}) results in
\begin{align*}
\Pr \left[ X_i(t+1) >\frac{x_i^2}{\Decided} + \frac{x_i}{\Decided}\cdot \sqrt{n}\delta ~\Big|~ \mathcal{E}_i \right] &<
      \Pr \left[ X_i(t+1) > \frac{y_i}{d} n + \frac{\sqrt{y_i}}{d}\cdot n \cdot \frac{\delta}{8} ~\Big|~ \mathcal{E}_i \right]\\
      &<
      4 \exp\left(-\polyaConstSmall \cdot\frac{\delta^2}{64}\right).
     \end{align*}
At last we combine this with Inequality (\ref{eq:3}) via an application of the law of total probability.
Then we get that
\begin{align*}
    \Pr \left[X_i(t+1) > \frac{x_i^2}{\Decided} + \frac{x_i^2}{\Decided}\cdot \sqrt{n} \cdot \delta\right]
    &=
    \Pr \left[ X_i(t+1) >\frac{x_i^2}{\Decided} + \frac{x_i}{\Decided}\cdot \sqrt{n}\cdot \delta ~\Big|~ \mathcal{E}_i
    \right]\cdot \Pr\left[\mathcal{E}_i\right]
    \\&\phantom{{}={}}+
     \Pr \left[ X_i(t+1) >\frac{x_i^2}{\Decided} + \frac{x_i}{\Decided}\cdot \sqrt{n}\cdot\delta ~\Big|~ \bar{\mathcal{E}_i}
    \right]\cdot \Pr\left[\bar{\mathcal{E}_i}\right]
    \\
    &< 4\exp\left(-\polyaConstSmall \cdot \frac{\delta^2}{64}\right) + 2\exp\left(-\frac{\delta^2}{192}\right)  < 7 \exp\left(-\concConstEps \cdot \delta^2\right)
\end{align*}
for some suitably chosen constant $\concConstEps=\polyaConstSmall/192>0$. As a matching lower bound can be developed by a symmetric approach, we omit the detailed proof.
\end{proof}

\begin{lemma}
\label{lemma:weak-phase-conc}
Fix $\X(t) = \x(t)$ and an opinion $i$ with  $x_i(t) \leq c \sqrt{n \log n}$. For any constant $c>0$ 
and $\Decided = \sum_{j=1}^{k} x_{j}(t)^2/n$, it holds that 
$\Pr [ X_i(t+1)  >  \left(12c^2 + 74\polyaConstLarge\right)\cdot n \cdot \log n/\Decided] < 4 n^{-2}$.
\end{lemma}
\begin{proof}[Proof]
	We track opinion $i$ with support $x_i(t) \leq c \sqrt{n \log n}$ throughout the \Decision and \Boosting part $t$.
		Similar to the proof of \cref{lemma:one-phase-conc} first we analyze $Y_i(t)$ and $\norm{\config{Y}(t)}$.
	Again we model $Y_i(t) \sim \Bin(x_i, x_i /n)$ and $\norm{\config{Y}(t)}_1 \sim \sum_{j=1}^k \Bin(x_j , x_j / n)$.
Let $c'=\max\{c,\sqrt{6}\}$. 
Then our goal is to bound
$\Pr [ X_i(t+1)  > n/\Decided \cdot (12{c'}^2+2\polyaConstLarge) \cdot \log n ]$.

First, note that if $ \Decided < (12{c'}^2+2\polyaConstLarge) \cdot \log n $ then $n/\Decided \cdot (12{c'}^2+2\polyaConstLarge) \cdot \log n > n$ and

	\[
		\Pr \left[ X_i(t+1)  > \frac{n}{\Decided} \cdot (12{c'}^2+2\polyaConstLarge) \cdot \log n \right] =0
	\]
	and the statement of the lemma follows immediately.
	Hence, in the following we can assume that
	 $ \Decided \ge (12{c'}^2+2\polyaConstLarge) \cdot \log n $.
	We define the event $\mathcal{E}_i$ as follows
\[
	    \mathcal{E}_i = \left\{ Y_i(t) < 2{c'}^2 \cdot \log n \land \norm{\config{Y}(t)}_1 > \frac{\Decided}{2} \right\}.
	\]
	First we bound the probability of $\bar{\mathcal{E}_i}$.
	We apply general Chernoff upper Bound (\cref{General_Upper_Chernoff_bound}) and get for $\delta'=1$ that
\[
\Pr\left[  Y_i(t) \ge 2{c'}^2 \cdot \log n \right] \leq
\exp\left(-\frac{{c'}^2 \cdot \log n}{3}\right) \le \exp\left(-\frac{6 \cdot\log n}{3}\right) \le  n^{-2}.
\]
Due to the definition of $c'$ we have $\Decided\ge (12{c'}^2+2\polyaConstLarge) \log n\ge 10\log n$. Applying  Chernoff bounds with (\cref{lemma:chernoff_poisson_trials})  $\delta'=\sqrt{(4 \log n)/\Decided} \le 1/2$ we get
\[
\Pr\left[\norm{\config{Y}(t)}_1 \le \frac{\Decided}{2}\right]
\le \Pr\left[\norm{\config{Y}(t)}_1 \le \Decided \cdot \left(1-\sqrt{\frac{4\cdot\log n}{\Decided}}\right)\right] \le
\exp\left(-\frac{4\cdot\Decided\cdot \log n}{2\cdot\Decided}\right)
\le n^{-2}.
\]
An application of the union bound yields
\begin{equation}\label{eq:73}
    \Pr\left[\bar{\mathcal{E}_i}\right]\le 2n^{-2}.
\end{equation}

It remains to consider the outcome of boosting part $t$.
We fix $Y_i(t) =y_i$ and define $d:=\norm{\config{Y}(t)}_1$.

Similar to \cref{lemma:one-phase-conc} we model $X_i(t+1) \sim \PE(y_i, d- y_i, n - d)$ and apply \cref{thm:polya_bound_small_support},  which states a tail bound for this Pólya Eggenberger distribution, to deduce that
	\begin{align}
	\label{eq:74}
	    \Pr \left[X_i(t+1) > \frac{n}{d}\cdot \left( 3y_i + \polyaConstLarge\cdot\log{n}\right) ~\Big|~ \mathcal{E}_i \right]< 2n^{-2}.
	\end{align}
	Conditioned on $\mathcal{E}_i$ we have $y_i < 2{c'}^2 \log n$ and $d > \frac{\psi}{2}$. 
    Therefore,

	\[
	\frac{n}{d}\cdot \left( 3y_i + \polyaConstLarge\cdot\log{n}\right) < \frac{n}{\Decided}\cdot \left( 6y_i + 2\polyaConstLarge\cdot\log{n}\right) < \frac{n}{\Decided} \cdot\left( \left(12{c'}^2 + 2\polyaConstLarge\right)\cdot\log n\right)
	\] Together with Inequality (\ref{eq:74}), this yields
	\[
	   \Pr \left[ X_i(t+1) > \frac{n}{\Decided}\cdot ((12{c'}^2+2\polyaConstLarge) \cdot  \log n) ~\Big|~ \mathcal{E}_i \right] < \Pr \left[X_i(t+1) > \frac{n}{d}\cdot \left( 3y_i + 2\polyaConstLarge\log{n}\right) ~\Big|~ \mathcal{E}_i \right]  <2n^{-2}    .
	\]
Similar to \cref{lemma:one-phase-conc} the statement follows by an application of the law of total probability.
Note that $\polyaConstLarge>1$, therefore $12{c'}^2+2\polyaConstLarge< 12c^2+74\polyaConstLarge$.
\end{proof}

%The following lemma shows that the bias between two opinions roughly squares if both their size and the difference between them are $\BigOmega{\sqrt{n \log n}}$.
\begin{lemma}%
\label{lemma:ratio} 
Fix $\X(t) = \x(t)$
and consider two opinions $i,j$ with support $x_i(t) - x_j(t) \geq \sDelta \sqrt{n\log{n}}$ and $x_j(t) \geq c_w \sqrt{n\log{n}}$, then
$
\Pr[{X_i(t+1)}/{X_j(t+1)} \geq (x_i(t)/x_j(t))^{1.5}]\ge 1-2n^{-2}.
$
\end{lemma}
\begin{proof}[Proof] First we lower bound the support of opinion $i$  and upper bound the support of opinion $j$. To do so, we apply \cref{lemma:one-phase-conc} to both opinions with $\delta = \sqrt{(\ln 7+2\log n)/\concConstEps}$.

Note that $c_w\ge \frac{8\delta}{\sqrt{\log n}}$. In this way, we get
    \[
      \Pr\left[X_i(t+1)\ge \frac{x_i^2}{\Decided}- \frac{x_i}{\Decided}\cdot \sqrt{n}\cdot \delta
    \right] \ge  1-n^{-2} \quad \text{ and } \quad
    \Pr\left[X_j(t+1)\le \frac{x_j^2}{\Decided}+ \frac{x_j}{\Decided}\cdot \sqrt{n}\cdot \delta
      \right] \ge  1-n^{-2}.
    \] 
 
    An application of the union bound yields with probability at least $1-2n^{-2}$
    \begin{align*}
        \frac{X_i(t+1)}{X_j(t+1)}
        \geq \left(\frac{\frac{x_i^2}{\Decided}- \frac{x_i}{\Decided}\cdot \sqrt{n}\cdot \delta}{\frac{x_j^2}{\Decided}+ \frac{x_j}{\Decided}\cdot \sqrt{n}\cdot \delta}\right)
            \geq \left( \frac{x_i}{x_j} \right)^2 \cdot \left( \frac{1-\frac{\sqrt{n} \delta}{x_i}}{1+\frac{\sqrt{n} \delta}{x_j}} \right) \geq \left( \frac{x_i}{x_j} \right)^2 \cdot \left(1-\frac{2 \sqrt{n} \delta}{x_j} \right)
    \end{align*}
    where we applied the inequality $(1-a)/(1+a)\ge 1-2a$ for any $a$. For the moment let us assume  $1- (2 \sqrt{n}\cdot \delta/x_j) \ge  (x_i/x_j)^{-1/2}$, then we have
    \[
      \frac{X_i(t+1)}{X_j(t+1)}
            \geq \left( \frac{x_i}{x_j} \right)^2 \cdot \left(1-\frac{2 \sqrt{n} \delta}{x_j} \right)
            \geq \left( \frac{x_i}{x_j} \right)^{1.5}.
    \] and the statement follows immediately.
    Hence, in the remaining part of the proof we prove the following observation.
    \begin{observation*}
    \label{eq:ratio_case_study}
     If $x_i-x_j\ge\sDelta\sqrt{n \log n}$ and $x_j\ge c_w\sqrt{n\log n}$ then, \[  1-\left( \frac{x_i}{x_j} \right)^{-1/2} \geq \frac{2 \sqrt{n} \delta}{x_j}.\]
    \end{observation*}
    First assume that $x_i / x_j \geq 2$.
  
    Then
    \[
        1-\left( \frac{x_i}{x_j} \right)^{-1/2} \geq 1/4 \geq \frac{2\delta}{c_w\sqrt{\log n}} \geq \frac{2\sqrt{n} \delta}{x_j}
    \]

    Next assume that $x_i/x_j < 2$.
    Let $1>\varepsilon>0$ be defined s.t. $x_i / x_j = 1 + \varepsilon$.
    Observe that
    \begin{equation}
    \label{eq:squaring_eps}
        1-\left(\frac{x_i}{x_j}\right)^{-1/2}  =  1-\left(1+\varepsilon\right)^{-1/2} \geq \varepsilon/10     \end{equation}

    By using the assumptions,
 
     it follows that
    \begin{equation}
    \label{eq:squaring_eps_2}
        \varepsilon/10
        = \frac{x_i-x_j}{10 x_j}
        \geq \frac{\sDelta \sqrt{n\log{n}}}{10 x_j}
        \geq \frac{c_w \sqrt{n\log{n}}}{x_j}
        \geq \frac{2\sqrt{n} \delta}{x_j}.
    \end{equation}
     Then, it follows by a combining (\ref{eq:squaring_eps}) and (\ref{eq:squaring_eps_2}).
\end{proof}

\lemInitialBias*
\begin{proof}[Proof]
Recall that an opinion $j$ is insignificant if $x_{max}-x_j > \sDelta \cdot \sqrt{n\cdot\log{n}}$.

Let $A(t)$ be the set of insignificant opinions with support larger than $4\cdot c_w \cdot \sqrt{n\cdot\log{n}}$ and let $B(t)$ be the set of insignificant opinions with support smaller than
$4\cdot c_w \cdot \sqrt{n\cdot\log{n}}$ and larger than zero.
We show that every member of $A(t)$ and $B(t)$ remains insignificant at the start of phase $t+1$.
For this we consider several cases depending on the largest opinion with support $x_{max}$.
\paragraph{Case 1: $x_{max}\geq \sqrt{n}\cdot \log^{3/2}n$}
First we deal with $A(t)$ by fixing an opinion $j \in A(t)$. Note that $A(t) = \emptyset$ implies the statement in this case directly. 
We lower bound the support of opinion with maximum support and upper bound the support of opinion $j$. To do so, we apply \cref{lemma:one-phase-conc} to both opinions with $\delta = \sqrt{(\ln 7+2\log n)/\concConstEps}$ and yield
\begin{equation}\label{eq:upper_bound_lem:initial-bias}
    \Pr\left[X_{max}(t+1)\ge \frac{x_{max}^2}{\Decided}- \frac{x_{max}}{\Decided}\cdot \sqrt{n}\cdot \delta \right] \ge  1-n^{-2}
    \end{equation}
    and
\[
    \Pr\left[X_j(t+1)\le \frac{x_j^2}{\Decided}+ \frac{x_j}{\Decided}\cdot \sqrt{n}\cdot \delta\right] \ge  1-n^{-2}.
\]
By a simple union bound we have with probability at least $1-2\cdot n^{-2}$ that
\begin{align}
\label{eq:bias-00}
	X_{max}(t+1) - X_j(t+1)
	&>  \frac{x_{max}^2 -x_j^2}{\Decided} -\left(\frac{(x_{max}+x_j)\cdot \sqrt{n}\cdot\delta}{\Decided}\right)
	\nonumber \\ 
	 &= \frac{x_{max} + x_j}{\Decided}\cdot \left( (x_{max} - x_j) - \delta\right).
\end{align}
Since $x_{max} > \Decided$ (see  \cref{{obs:polya}}) and $x_{max}>x_j$ we have
\begin{align*}
\label{eq:bias-uno}
	 \frac{x_{max} + x_j}{\Decided}\cdot \left( (x_{max} - x_j) - \delta\right)
	 &\geq \left(1+ \frac{x_j}{x_{max}}\right) \cdot (x_{max}-x_j) - \left(1+\frac{x_j}{x_{max}}\right) \cdot \delta \\
	 &= (x_{max}-x_j) + \frac{x_j}{x_{max}} \cdot (x_{max}-x_j) - \left(1+\frac{x_j}{x_{max}}\right) \cdot \delta \\
	 &\geq  (x_{max}-x_j) + \frac{x_j}{x_{max}} \cdot (x_{max}-x_j) - 2 \cdot \delta.
\end{align*}
In the observation below we show $x_j/x_{max}\cdot (x_{max} - x_j) - 2\delta>0$. From that follows $X_{max}(t+1)-X_j(t+1) \geq(x_{max}-x_j)\geq \sDelta\cdot\sqrt{n\cdot\log{n}}$ where the last inequality holds since opinion j is insignificant.

\begin{observation*}
Assume  $x_j > 4 \cdot c_w \cdot \sqrt{n \cdot \log n}$ and $x_{max} - x_j\ge \sDelta\sqrt{n\log n}$ then $\frac{x_j}{x_{max}}\cdot (x_{max} - x_j) - 2\delta>0$.
\end{observation*}
\begin{proof}

\noindent To ease the calculations let $\Delta'=(x_{max} - x_j)/\sqrt{n \log n}$ and $c_1 = \sqrt{(\ln 7+2\log n)/(\concConstEps\cdot \log n)}$.
Recall that $\delta = c_1\cdot \sqrt{\log n}$, then
\begin{align*}
    \frac{x_j}{x_{max}}\cdot (x_{max} - x_j) - 2\delta
    &\ge \sqrt{n\log n}\cdot \left( \frac{x_j\cdot (x_{max}-x_j)}{x_{max} \cdot \sqrt{n\cdot\log{n}}} -2c_1\right) \\
    &\ge \sqrt{n\log n}\cdot \left( \frac{x_j\cdot \Delta'}{x_j+x_{max}-x_j} -2c_1\right) \\
	&= \sqrt{n\log n}\cdot \left( \frac{ \Delta'}{1+\frac{x_{max}-x_j}{x_j}}-2c_1 \right) \\
    &\overset{(a)}{\ge} \sqrt{n\log n}\cdot \left( \frac{ \Delta'}{1+\frac{\Delta'}{4c_w}}-2c_1 \right) \\
    &= \sqrt{n\log n} \cdot \left( \frac{4c_w \cdot \Delta' - 2c_1 \cdot (4c_w + \Delta')}{4c_w + \Delta'} \right) \\ 
    &= \sqrt{n\log n} \cdot \left( \frac{ 4c_w(\Delta'-2c_1) +\Delta'(4c_w-2c_1) }{4c_w + \Delta'} \right)
    \overset{(b)}{>} 0    
\end{align*}
where we use (a) $(x_{max}-x_j)/{x_j}\leq \frac{x_{max}-x_j}{4\cdot c_w \cdot \sqrt{n\cdot\log{n}} } =\frac{\Delta'}{4c_w}$ due to $x_j\ge 4c_w\sqrt{n\log n}$ and (b) $\Delta' =(x_{max}-x_j)/\sqrt{n\log n} \geq \sDelta > 2\cdot c_1$ and $c_w=8\sqrt{1+2/\concConstEps} > c_1/2$ (\cref{def:Constant}).
\end{proof}
Another union bound application on (\ref{eq:bias-00}) over all opinions $j \in A(t)$ yields with probability at least $1-2\cdot n^{-1}$ that all opinions $j\in A(t)$ remains insignificant at the start of phase $t+1$.

Next we deal with $B(t)$ by fixing an opinion $j \in B(t)$ (assuming $B(t) \neq \emptyset$).
Again We lower bound the support of the largest opinion  and upper bound the support of opinion $j$.
For the largest opinion we have (\ref{eq:upper_bound_lem:initial-bias}) where for opinion $j$ we get by \cref{lemma:weak-phase-conc} for $c = 4\cdot c_w$ that
\begin{equation}\label{eq:case2-lem:initial-bias}
\Pr[X_j(t+1)  < \frac{n}{\Decided} \cdot (192\cdot c_w^2 + 74\polyaConstLarge) \cdot \log{n}]\ge 1-4n^{-2} .
\end{equation}
An application of the union bound yields with probability at least $1-5n^{-2}$ that
\begin{align}\label{eq:large_max_small_insignificant_difference}
    X_{max}(t+1)-X_j(t+1)
    &\geq \frac{x_{max}^2}{\Decided}- \frac{x_{max}}{\Decided}\cdot \sqrt{n}\cdot \delta - \frac{n}{\Decided} \cdot (192\cdot c_w^2 + 74\polyaConstLarge) \cdot \log{n} \\
    &= \frac{x_{max}^2}{\Decided} \cdot \left( 1 - \frac{\delta \cdot \sqrt{n}}{x_{max}} - \frac{(192\cdot c_w^2 + 74\polyaConstLarge) \cdot n\cdot\log{n}}{x_{max}^2} \right) \nonumber\\
    &\geq x_{max} \cdot \left( 1 - \frac{\delta \cdot \sqrt{n}}{\sqrt{n} \cdot \log^{3/2}{n}} - \frac{(192\cdot c_w^2 + 74\polyaConstLarge) \cdot n\cdot\log{n}}{n\cdot\log^{3}{n}} \right) \nonumber \\
    &\geq x_{max} \cdot \left( 1 - \frac{\delta }{\log^{3/2}{n}} - \frac{(192\cdot c_w^2 + 74\polyaConstLarge)}{\log^{2}{n}} \right) \nonumber \\
    &\geq \sDelta \cdot \sqrt{n\log{n}} \nonumber
\end{align}
where we use $x_{max} > \Decided$ (see \cref{obs:polya}) and $x_{max} \geq \sqrt{n}\cdot \log^{3/2}{n}$.
Another union bound application on (\ref{eq:large_max_small_insignificant_difference}) over all opinions $j \in B(t)$ yields with probability at least $1-5\cdot n^{-1}$ that all opinions $j\in B(t)$ remains insignificant at the start of phase $t+1$.

\medskip

Next we do a similar analysis as before on $A(t)$ and $B(t)$ but with another assumption on the largest opinion.
\paragraph{Case 2: $x_{max} < \sqrt{n}\cdot \log^{3/2}n$}
Again, first we deal with $A(t)$.
Assuming $A(t) \neq \emptyset$ (otherwise there is nothing to show) we fix an opinion $j \in A(t)$.
Note that by the definition of insignificant opinions $x_{max} \geq \sDelta\cdot\sqrt{n\cdot \log n}$. 
This is suffices to apply the analysis of the first case on the largest opinion and $j\in A(t)$ such that all opinions $j \in A(t)$ remains insignificant at the start of phase $t+1$ with probability at least  $1-2\cdot n^{-1}$.

Next we deal with $B(t)$ by fixing an opinion $j \in B(t)$ (assuming $B(t) \neq \emptyset$).
Again we lower bound the support of opinion with maximum support and upper bound the support of opinion $j$.
Note that by the definition of insignificant opinions $x_{max} \geq \sDelta\cdot\sqrt{n\cdot \log n}$. 
This is suffices to use (\ref{eq:upper_bound_lem:initial-bias}) for the largest opinion whereas we use (\ref{eq:case2-lem:initial-bias}) for opinion $j$.
Note that $c_w>8$. Recall that $\delta = c_1\cdot \sqrt{\log n}$ for $c_1 = \sqrt{(\ln 7+2\log n)/(\concConstEps\cdot \log n)}$.
Then, an application of the union bound yields with probability at least $1-5n^{-2}$ that
\begin{align}\label{eq:ratio_max_small_insignificant_small}
    \frac{X_{max}(t+1)}{X_j(t+1)}
    &\geq \frac{ \frac{x_{max}^2}{\Decided} \cdot \left( 1- \frac{c_1 \cdot \sqrt{n\log n}}{x_{max}} \right) }{ \frac{n}{\Decided} \cdot (192\cdot c_w^2 + 74\polyaConstLarge) \cdot \log n } = \frac{x_{max}^2 \cdot \left( 1- \frac{c_1 \cdot \sqrt{n\log n}}{x_{max}} \right) }{ n\log{n} \cdot (192\cdot c_w^2 + 74\polyaConstLarge)} \nonumber \\
    &\overset{(a)}{\geq} \frac{ \sDelta^2 \cdot n\log n \cdot \left( 1- \frac{c_1 \cdot \sqrt{n\log n}}{\sDelta\cdot\sqrt{n\log n}} \right) }{ n\log{n} \cdot (192\cdot c_w^2 + 74\polyaConstLarge) } = \frac{ \sDelta^2 \cdot \left(1-\frac{c_1}{\sDelta}\right) }{(192\cdot c_w^2 + 74\polyaConstLarge)} 
\end{align}
where we use (a) $x_{max}\geq \sDelta\cdot \sqrt{n\log n}.$ 
Another union bound application  over all opinion $j \in B(t)$ yields with probability at least $1-5n^{-1}$ that (\ref{eq:ratio_max_small_insignificant_small}) holds for all $j \in B(t)$.
Additionally, $ x_{max} > \sDelta \cdot \sqrt{n \log n}$ allows us to apply the second statement of \cref{lem:k-to-sqrtlog}, 
which implies with probability at least $1-n^{-\BigOmega{1}}$  that
\[
	X_{max}(t+1) \geq \min\{ \sqrt{n}\cdot \log^{3/2} n ~,~ \sDelta \cdot\sqrt{n \log n}  + \frac{1}{60} \cdot \sqrt{n \log n}  \} > \sDelta \cdot \sqrt{n\log n} \cdot\left( 1+ \frac{1}{60 \cdot \sDelta} \right).
\]
Combining this with the previous result we get that, with probability at least $1-(5n^{-1} + n^{-\BigOmega{1}})$, for every $j \in B(t)$ it holds that
\begin{align*}
	X_{max}(t+1) - X_j(t+1) 
	&= X_{max}(t+1) \cdot \left( 1- \frac{X_j(t+1)}{X_{max}(t+1)} \right) \\
	&>  \sDelta \cdot \sqrt{n \log n} \cdot \left( 1+ \frac{1}{60\cdot \sDelta} \right) \cdot \left(1 - \frac{192\cdot c_w^2 + 74\polyaConstLarge}{\left(1-\frac{c_1}{\sDelta}\right) \cdot \sDelta^2}\right) \\
	&=  \sDelta \cdot \sqrt{n \log n} \cdot \left( 1+ \frac{1}{60\cdot \sDelta} \right) \cdot \left(1 - \frac{192\cdot c_w^2 + 74\polyaConstLarge }{(\sDelta-c_1) \cdot \sDelta}\right)
	\\
	&\overset{(a)}{>}  \sDelta \cdot \sqrt{n \log n} \cdot \left( 1+ \frac{1}{60\cdot \sDelta} \right) \cdot \left(1 - \frac{2(192\cdot c_w^2 + 74\polyaConstLarge) }{\sDelta \cdot \sDelta}\right)
	\\
	&\overset{(b)}{>} \sDelta \cdot \sqrt{n \log n} \cdot \left( 1+ \frac{1}{60\cdot \sDelta} \right) \cdot \left(1 - \frac{1}{62 \cdot \sDelta}\right)
	> \sDelta \cdot\sqrt{n \log n}
\end{align*}
where in (a) we use $\sDelta-c_1> \sDelta/2$ and in (b) we utilize $\sDelta = (160\cdot c_w)^2+(148\cdot \polyaConstLarge)^2 > 124\cdot(192c_w^2+74\polyaConstLarge)$ (see \cref{def:Constant}).
 This implies that all opinions $j\in B(t)$ remains insignificant at the start of phase $t+1$.

\end{proof}

\lemKtoSqrtlog*
\begin{proof}[Proof]
    We start with the proof of the first statement and therefore assume $\sqrt{n} \leq x_{max} \leq \sqrt{n \log n}$.
    The main ingredients of this proof will be \cref{lem:case-a}. We set $\tau = 0.9 \cdot x_{max} $ and distinguish between two cases. 
    \paragraph{Case 1} At most $\sqrt{n} / \log^{4} n$ opinions have size at least $\tau$ at the start of \Decision part $t$. \\
    This case is covered by the first statement of \cref{lem:case-a}, which immediately yields that 
    \[
        X_{max}(t+1) > x_{max} + \frac{1}{10} \cdot  x_{max} - \frac{1}{12} \cdot x_{max} = x_{max} \cdot \left( 1 + \frac{1}{60}\right)
    \]
     with probability $1 - 7\exp(- ( \concConstEps / 625) \cdot x_{max}^2 / n)$.
     
   \paragraph{Case 2} At least $\sqrt{n} / \log^{4} n$ opinions have size at least $\tau$ at the start of \Decision part $t$. \\
    We employ the second case of \cref{lem:case-b}, which yields that
    \[
        X_{max}(t+1) > (1 - n^{-\BigOmega{1}}) \cdot \left( \frac{\tau^2}{x_{max}} + \frac{1}{4} \cdot \frac{\tau}{x_{max}}  \sqrt{n \log n} \right)
    \]
    with probability $1- 7\exp(-(\polyaConstSmall / 25) \cdot \tau^2 /  n) = 1 - 7 \exp (-(81\polyaConstSmall / 2500) \cdot x_{max}^2 / n)$. When substituting $\tau$ by $(9/10) \cdot x_{max}$ we can simplify this inequality to 
    \begin{equation*}
        X_{max}(t+1) > (1 - n^{-\BigOmega{1}}) \left( x_{max} \cdot \frac{81}{100} + \frac{9}{40} \cdot \sqrt{n \log n}\right) > (1 - n^{-\BigOmega{1}}) \cdot x_{max} \cdot \frac{207}{200}
    \end{equation*}
    where we used in the second step that $x_{max} \leq \sqrt{n \log n}$. For large enough $n$, we have $(1 -n^{-\BigOmega{1}}) \cdot (207 / 200) > 1 + 1/60$. Therefore, we have $X_{max}(t+1) > x_{max}(1 + 1/60)$ with probability at least  $1 - 7 \exp (-(81\polyaConstSmall / 2500) \cdot x_{max}^2 / n)$ in this case.

    Hence, in both cases, the lemmas statement follows with probability at least
    \[
        1 - \max\Big\{ 7\exp(- ( \concConstEps / 625) \cdot x_{max}^2 / n) ~,~ 7 \exp (-(81\polyaConstSmall / 2500) \cdot x_{max}^2 / n)\Big\}
        \geq 1 - 7\exp(- ( \concConstEps / 625) \cdot x_{max}^2 / n),
    \]
    where we used that $\concConstEps \geq \polyaConstSmall / 192$ (see \cref{def:Constant}) and the first statement of \cref{lem:k-to-sqrtlog} follows.
    
    \medskip
    
    The second statement of \cref{lem:k-to-sqrtlog} follows by a similar argument. We assume that  $\sqrt{n \log n} <  x_{max} < \sqrt{n}\log^{3/2} n$, this time set  $\tau= x_{max} - \frac{1}{10}\sqrt{n \log n}$, and again distinguish two cases. 
    
    \paragraph{Case 1} At most $\sqrt{n} / \log^{4} n$ opinions have size at least $\tau$ at the start of \Decision part $t$. \\
    In this case we again apply the first case of \cref{lem:case-a}, which because of $x_{max} > \sqrt{n \log n}$ yields that
    \[
        X_{max}(t+1) > x_{max} + \frac{1}{10} \cdot \sqrt{n \log n} - \frac{1}{12} \cdot \sqrt{n \log n} = x_{max} +  \frac{1}{60} \cdot \sqrt{n \log n}
    \]
     with probability $1-7\exp(-\concConstEps\cdot \log n/625)$.
    
     \paragraph{Case 2}
    At least $\sqrt{n} / \log^{4} n$ opinions have size at least $\tau$ at the start of \Decision part $t$. \\
    We employ the second case of \cref{lem:case-b}, which immediately yields that
    \begin{equation}
    \label{eq:k-to-sqrtlog-2}
        X_{max}(t+1) > (1 - n^{-\BigOmega{1}}) \cdot \left( \frac{\tau^2}{x_{max}} + \frac{1}{4} \cdot\frac{\tau}{x_{max}} \cdot \sqrt{n \log n} \right)
    \end{equation}
    with probability $1-7\exp\left(-\frac{\polyaConstSmall}{25} \cdot \min \{\frac{\tau^2}{n}, \log n\} \right).$ We first substitute $\tau$ by $x_{max} - \frac{1}{10}\sqrt{n \log n}$ to further bound the right-hand side
    \begin{align*}
        (1 - n^{-\BigOmega{1}}) &\cdot \left( \frac{\left(x_{max} - \frac{1}{10} \cdot\sqrt{n \log n}\right)^2}{x_{max}} + \frac{1}{4} \cdot \frac{x_{max} - \frac{1}{10} \cdot\sqrt{n \log n}}{x_{max}}  \sqrt{n \log n} \right) \\
       & > (1 - n^{-\BigOmega{1}}) \cdot \left( x_{max}  - \frac{2}{10} \cdot\sqrt{n \log n} + \frac{1}{4} \cdot\left( 1 - \frac{\sqrt{n \log n}}{10 \cdot x_{max}}\right) \cdot\sqrt{n \log n}  \right) \\
       & > (1 - n^{-\BigOmega{1}}) \cdot \left( x_{max}  - \frac{2}{10} \cdot\sqrt{n \log n} + \frac{9}{40} \cdot\sqrt{n \log n}\right)\\
       & = (1 -n^{-\BigOmega{1}}) \cdot \left( x_{max} + \frac{1}{40} \cdot \sqrt{n \log n} \right) > x_{max} + \frac{1}{60} \cdot \sqrt{n \log n}.
    \end{align*}
    In the second and last step we used that $\sqrt{n \log n} < x_{max} < \sqrt{n} \log^{3/2} n$. 
    Additionally, note that  $\tau\ge 9/10\sqrt{n\log n}$. Hence, the probability of (\ref{eq:k-to-sqrtlog-2}) can be bounded by $1-7\exp(- (81\polyaConstSmall / 2500) \cdot \log n)$ and the second statement follows.
    
    Just as in the proof of the first statement, we finish by noting that the probability in either case is at least 
    \[
        1 - \max\Big\{ 7\exp(-\concConstEps\cdot \log n/625) ~,~ 7\exp(- (81\polyaConstSmall / 2500) \cdot \log n)\Big\}
        \geq 1 - 7\exp(-\concConstEps\cdot \log n/625),
    \]
    where we used $\concConstEps \geq \polyaConstSmall / 192$.
    
\end{proof}

In what follows we list the technical statement we used in the proof of \cref{lem:k-to-sqrtlog}.

\begin{lemma}
\label{lem:case-a}
\label{lem:case-b}
    Fix $\X(t) = \x(t)$. Let $\sqrt{n} \leq x_{max}(t) < \sqrt{n} \log^{3/2} n$ and  $0 < \tau < x_{max}(t)$.
    \begin{enumerate}
        \item If at most $\sqrt{n} / \log^4 n$ opinions have larger support than $\tau$, then with probability at least $1-7\exp\left(-(\concConstEps / 625) \cdot \min\{x_{max}(t)^2 / n \,,\, \log n \}\right)$ it holds that 
    \[
          X_{max}(t+1) > x_{max}(t) + (x_{max}(t) - \tau) - \frac{1}{12}\min\{x_{max}(t), \sqrt{n \log n}\}.
    \]
    \item If at least $\sqrt{n} / \log^4 n$ opinions have larger support than $\tau$, then with probability at least $1-7\exp\left(-(\polyaConstSmall/25) \cdot \min \{\tau^2 / n \,,\, \log n\} \right)$ it holds
    \[
       X_{max}(t+1) > (1 - n^{-\BigOmega{1}}) \cdot \left( \frac{\tau^2}{x_{max}(t)} + \frac{1}{4}\frac{\tau}{x_{max}(t)}  \sqrt{n \log n} \right).
    \]
    \end{enumerate}
\end{lemma}
\begin{proof}[Proof]
We start with the first statement. Let $i$ be an opinion which provides the $x_{max}$. The intuition is that many opinions lie below $\tau$ and therefore  $x_i$ will grow in expectation. Observe that in this setting $\Decided = \Ex{\norm{\config{Y}(t)}_1}$ is maximized if $\sqrt{n} / \log^{4} n$ opinions have support $x_{max}$, while as many remaining opinions as possible have support exactly $\tau$.
    It follows that
    \begin{equation}
    \label{eq:case-a}
        \Decided=\Ex{\norm{\config{Y}(t)}_1} = \sum_{j=1}^{k} \frac{x_j^2}{n} \leq \frac{\sqrt{n}}{\log^{4}{n}}\cdot \frac{x_{max}^2}{n} + \tau = \tau + \LittleO{\sqrt{n}} =: \tau'.
    \end{equation}
     We now employ the result of \cref{lemma:one-phase-conc} together with $\delta$ such that $25 \cdot \delta = \min\{x_i / \sqrt{n} , \sqrt{\log n}\}$, which directly yields
     \begin{equation}
     \label{eq:case-a-2}
        \Pr\left[X_i(t+1) > \frac{x_i^2}{\Decided} - \frac{x_i}{\Decided} \sqrt{n} \delta \right] > 1 - 7 \exp(-\concConstEps  \delta^2).
     \end{equation}
     We now lower bound the expression in (\ref{eq:case-a-2}) by substituting $\delta$ and applying the bound on $\Decided$ (\ref{eq:case-a}) in the second step.
     \[
        \frac{x_i^2}{\Decided} - \frac{x_i}{\Decided} \sqrt{n} \delta =  \frac{x_i}{\Decided}\cdot \left( x_i - \sqrt{n} \delta \right) \geq \frac{x_i}{\tau'}\cdot \left(x_i - \frac{1}{25} \min\{x_i, \sqrt{n \log n}\} \right).
     \]
    Using that $x_i / \tau' \geq (1 + \frac{x_i -\tau'}{x_i})$ (follows from $(x_i - \tau')^2 \geq 0$) we continue this inequality chain as follows
    \[
        \frac{x_i}{\tau'} \left(x_i - \frac{1}{25} \min\{x_i, \sqrt{n \log n}\} \right) \geq x_i + (x_i - \tau') - \frac{1}{25} (1 + \frac{x_i - \tau'}{x_i})\cdot \min \{x_i , \sqrt{n \log n}\}.
    \]
    When using that $(x_i - \tau') / x_i < 1$ as well as  $\tau'-\tau = \LittleO{\sqrt{n}} = \LittleO{\min\{x_i, \sqrt{n \log n}\}}$, it follows that
    \[
        x_i + (x_i - \tau') - \frac{1}{25} \left(1 + \frac{x_i - \tau'}{x_i}\right) \cdot \min \set{x_i , \sqrt{n \log n}} >
        x_i + (x_i - \tau) - \frac{1}{12}  \min \set{x_i , \sqrt{n \log n}}.
    \]
    This concludes the proof of the first statement.

    \medskip
    We continue with the proof of the second statement. We assume that at least $\sqrt{n} / \log^4 n$ lie above $\tau$. Let $\cL$ denote the set of these opinions. We will show that at least one of these opinions will grow by roughly $\frac{\tau}{x_{max}} \sqrt{n \log n}$ more than expected.
         For $i \in \cL$, we model $Y_i(t) \sim \Bin(x_i, x_i / n)$ as usual. As $x_i \geq \tau$, it is easy to see that  $\Bin(x_i, x_i / n)$ stochastically majorizes $\Bin(\tau, \tau / n)$. When further applying the anti-concentration result \cref{lem:reverse-chernoff-simple} to $\Bin(x_i,x_i / n)$ and setting $\delta = \sqrt{\log n \cdot n} / (2\tau)$ we get the following inequality chain
         \begin{align}
         \label{eq:case-b-anticonc}
              \Pr \left[Y_i(t) >  \frac{\tau^2}{n} + \frac{1}{2} \frac{\tau}{\sqrt{n}} \sqrt{\log n} \right]  \overset{major.}{\geq}
            \Pr \left[ \Bin(\tau, \tau / n) > \frac{\tau^2}{n} + \frac{1}{2} \frac{\tau}{\sqrt{n}} \sqrt{\log n}\right]\nonumber \\
            \overset{ \cref{lem:reverse-chernoff-simple}}{>} \frac{1}{6 \sqrt{(1+\delta)\cdot\frac{\tau^2}{n}}} \cdot \exp(-\delta^2 \cdot \frac{\tau^2}{n}) = \frac{1}{\polylog n} \cdot  \exp(-\log n / 4) > n^{-1/3}.
         \end{align}
         We now define an indicator random variable $Z_i$ for every $i \in \cL$. We set $Z_i = 1$ iff $Y_i(t) > \frac{\tau^2}{n} + \frac{1}{2} \frac{\tau}{\sqrt{n}} \sqrt{\log n}$, and $Z_i = 0$ otherwise. Note that for two opinions $i,i' \in \cL$ with $i \neq i'$ the variables $Z_i$ and $Z_{i'}$ are independent. This way the random variable $Z := \sum_{i \in |\cL|} Z_i$ is the sum of $|\cL| > \sqrt{n} / \log^4 n$ independent Poisson trials. In (\ref{eq:case-b-anticonc}) we established that $\Pr[Z_i = 1] > n^{-1/3}$. Therefore, $\Ex{Z} > |\cL| \cdot n^{-1/3} > n^{1/6}$ and a Chernoff bound application yields that, with probability $(1-n^{-\omega(1)})$ we have $Z \geq 1$. This implies that
         \begin{equation}
         \label{eq:case-b-single-exists}
            \Pr \Big[\exists j \in  \cL: Z_j = 1 \Big] = \Pr \left[ \exists j \in \cL: Y_j(t) > \frac{\tau^2}{n} + \frac{1}{2} \frac{\tau}{\sqrt{n}} \sqrt{\log n} \right] > 1 - n^{-\omega(1)}.
         \end{equation}
         Additionally, we model $\norm{\config{Y}(t)}_1$ as the sum of Poisson trials where $\Ex{\norm{\config{Y}(t)}_1} \leq x_{max}$. As $x_{max} \geq \sqrt{n}$, another Chernoff bound application with $\delta = n^{-1/5}$ yields that, with probability $(1 - n^{-\omega(1)})$, we have  $\norm{\config{Y}(t)}_1 < x_{max} (1 + n^{-\BigOmega{1}})$. Consider the following event defined for every opinion $i \in \cL$
         \[
            \mathcal{E}_i \Leftrightarrow \Big\{ Y_i(t) > \frac{\tau^2}{n} + \frac{\tau}{2\sqrt{n}} \sqrt{\log n} \land \norm{\config{Y}(t)}_1  < x_{max} \cdot(1+ n^{-\BigOmega{1}}) \Big\}.
        \]
        From (\ref{eq:case-b-single-exists}) and our bound on $\norm{\config{Y}(t)}_1$ above it follows that, with probability $(1-n^{-\omega(1)})$, there must be a $j \in \cL$ such that $\mathcal{E}_j$ is fulfilled. We fix this opinion $j$ after the \Decision part and fix $Y_j(t)=y_j$ define $d=\norm{\config{Y}(t)}_1 $. Remember, we model the outcome of the \Boosting part $t$ as $X_{j}(t+1) \sim  \PE(y_j,d-y_j,n-d)$. By \cref{thm:polya_bound_total_balls} we get for any $0 < \delta < \frac{\tau}{\sqrt{n}}$ that
    \begin{equation}
    \label{eq:case-b}
        \Pr\Big[X_j(t+1) < \frac{n}{d}\cdot \left( y_j - \sqrt{y_j} \cdot \delta \right) ~\Big|~ \mathcal{E}_j \Big] < 4\exp(-\polyaConstSmall \cdot\delta^2).
    \end{equation}
        In the following part of the proof, our goal is to simplify the bound on $X_j(t+1)$ in (\ref{eq:case-b}). We define $y' := \frac{\tau^2}{n} + \frac{\tau}{2\sqrt{n}} \sqrt{\log n}$. Conditioned on $\mathcal{E}_j$ we have $y_j > y'$ and $d < x_{max}\cdot(1 + n^{-\BigOmega{1}})$. When using these bounds we get that
            \begin{equation}
    \label{eq:lem22-1}
         \frac{n}{d}\cdot \left( y_j - \sqrt{y_j} \cdot \delta \right) > \frac{n}{x_{max}}\cdot (1 - n^{-\BigOmega{1}})\cdot \left( y' - \sqrt{y'} \cdot \delta \right).
    \end{equation}
         Furthermore, when fixing $\delta$ s.t. $5 \cdot \delta = \min\{ \tau/\sqrt{n}~,~ \sqrt{\log n}\}$ in the penultimate step of the following inequality chain, we get that
    \begin{align}
    \label{eq:caseb-minmax}
        \sqrt{y'} \cdot \delta &= \sqrt{\frac{\tau^2}{n} + \frac{1}{2} \frac{\tau}{\sqrt{n}} \sqrt{\log n}} \cdot \delta =
        \sqrt{\frac{\tau}{\sqrt{n}} \cdot \left(\frac{\tau}{\sqrt{n}} + \frac{1}{2}\sqrt{\log n} \right)} \cdot \delta \nonumber\\
        &\leq \sqrt{\max\{ \frac{\tau}{\sqrt{n}}, \sqrt{\log n}\} \cdot \left(\frac{\tau}{\sqrt{n}} + \frac{1}{2}\sqrt{\log n} \right)} \cdot \delta \nonumber \\
         &\leq \sqrt{\max\{ \frac{\tau}{\sqrt{n}}, \sqrt{\log n}\} \cdot \left( 1 + \frac{1}{2} \right) \cdot \max\{ \frac{\tau}{\sqrt{n}}, \sqrt{\log n}\}} \cdot \delta \nonumber \\
        & = \sqrt{1 + \frac{1}{2}} \cdot  \max \{\frac{\tau}{\sqrt{n}}, \sqrt{\log n} \} \cdot \frac{1}{5} \min \{\frac{\tau}{\sqrt{n}}, \sqrt{\log n}\} < \frac{1}{4} \frac{\tau}{\sqrt{n}} \sqrt{\log n}
    \end{align}
    We use this to continue the inequality chain in (\ref{eq:lem22-1}) and substitute $y'$ in the second step to derive the following intermediate result
    \begin{align*}
        \frac{n}{d} (y_i - \sqrt{y_i} \delta) &\overset{(\ref{eq:lem22-1})}{>}   \frac{n}{x_{max}}\cdot (1 - n^{-\BigOmega{1}})\cdot \left( y' - \sqrt{y'} \cdot \delta \right) \\
       & > \frac{n}{x_{max}} \left( 1-n^{-\BigOmega{1}}\right) \left(\frac{\tau^2}{n} + \frac{\tau}{2\sqrt{n}} \sqrt{\log n} - \sqrt{y'} \delta\right)\\ &\overset{(\ref{eq:caseb-minmax})}{>} \frac{n}{x_{max}}\cdot (1 - n^{-\BigOmega{1}}) \cdot \left( \frac{\tau^2}{n} + \frac{\tau\cdot \sqrt{\log n}}{4 \sqrt{n}}\right) \\
       & = (1 - n^{-\BigOmega{1}}) \cdot \left( \frac{\tau^2}{x_{max}} + \frac{1}{4}\frac{\tau}{x_{max}}  \sqrt{n \log n} \right)
    \end{align*}
    Note that, for an arbitrary random variable $X$ and $x > x'$, it holds that $\Pr[X < x] \geq \Pr[X < x']$. Therefore, the above implies for $\delta = (1/5) \cdot \min\{ \tau/\sqrt{n}~,~ \sqrt{\log n}\}$  that
    \[
        \Pr\Big[X_j(t+1) < \frac{n}{d}\left( y_j - \sqrt{y_j} \cdot \delta \right) ~\Big|~ \mathcal{E}_j \Big] \geq
         \Pr \Big[ X_j(t+1) < (1 - n^{-\BigOmega{1}})  \Big( \frac{\tau^2}{x_{max}} + \frac{1}{4}\frac{\tau}{x_{max}} \sqrt{n  \log n} \Big)
        ~\Big|~
         \mathcal{E}_j \Big].
    \]
    Remember, in (\ref{eq:case-b}) we upper-bounded the probability on the left-hand side. Therefore, for $\delta = (1/5) \cdot \min\{ \tau/\sqrt{n}~,~ \sqrt{\log n}\}$ this bound also carries over to the probability on the right-hand side and we get
    \[
       \Pr \Big[ X_j(t+1) < (1 - n^{-\BigOmega{1}}) \cdot \Big( \frac{\tau^2}{x_{max}} + \frac{1}{4}\frac{\tau}{x_{max}} \sqrt{n  \log n} \Big) ~\Big|~ \mathcal{E}_j \Big] <  4\exp(-\frac{\polyaConstSmall}{25} \cdot(\min\{ \tau^2/n~,~ \log n\}).
    \]
    The result follows as we already established that some opinion $j \in \cL$ fulfills event $\mathcal{E}_j$ \whp.
\end{proof}

The next lemma is a simple statement, which shows that in cases which are not covered by \cref{lem:k-to-sqrtlog} the maximum may only shrink by an $(1 - o(1))$ factor.

\begin{lemma}
\label{lem:max-not-shrinking}
    Fix $\X(t) = \x(t)$ and let $x_{max}(t) \geq \sqrt{n} \log n$.
    Then,
    \[
       Pr\left[ X_{max}(t+1) > x_{max}(t) \cdot (1  - \frac{2}{\sqrt{\concConstEps\cdot \log n}})\right] \ge 1-n^{-2}.
    \]
\end{lemma}
\begin{proof}[Proof]
    We lower bound the support of the largest opinion.
    To do so, we apply \cref{lemma:one-phase-conc} to the largest opinion with $\delta = \sqrt{(2\log{n} +\ln{7})/\concConstEps}$.
    In this way, we get
    \[
        \Pr\Big[ X_{max}(t+1) \geq \frac{x_{max}^2}{\Decided} - \frac{x_{max}}{\Decided} \cdot \sqrt{n} \cdot \delta \Big] \geq 1 - n^{-2} .
    \]
    We can further relax this by using the relation between $x_{max}$ and $\Decided$ (see \cref{sec:analysis-population-model}) and the assumption $x_{max} \geq \sqrt{n}\cdot \log n$.
    It follows that
    \[
       \frac{x_{max}^2}{\Decided} - \frac{x_{max}}{\Decided} \cdot \sqrt{n} \cdot \delta
       \geq \frac{x_{max}^2}{x_{max}} \cdot \left( 1 -  \frac{\delta}{\log{n}}\right)
       \geq x_{max} \cdot \left( 1 - \frac{\sqrt{2\log{n} + \ln{7}}}{\sqrt{\concConstEps} \cdot \log{n}} \right)
       \geq x_{max} \cdot \left( 1 - \frac{2}{\sqrt{\concConstEps\cdot \log{n}}} \right) .
    \]

\end{proof}

\subsection{Consensus for $k \leq \sqrt{n} / \log n$}
\label{sec:details-analysis-case1}

\proCaseOne*
\begin{proof} 
    \label{proof:pro-case-one-part-2}
    First we show that all agents agree on one opinion. %, the proof that this opinion is significant can be found in \cref{sec:details-analysis-case1}.
    We fix $\X(t) = \x(t)$ and consider two arbitrary opinions $i$ and $j$.
    If both opinions are strong \cref{thm:one-weak} shows that one of them becomes weak or super-weak within $\BigO{\log{n}}$ phases with probability at least $1-\BigO{n^{-1.9}}$.
    As soon as either $i$ or $j$ are weak \cref{lem:weak-to-super-weak} shows that the weak opinion becomes super-weak within the next $\BigO{\log \log n}$ phases and remains super-weak for the rest of the process  (\cref{lem:weak-to-super-weak}). This happens again with a probability of  $1-\BigO{n^{-1.9}}$.
    Hence, after $\BigO{\log n}$  phases either $i$ or $j$ are super-weak. 
    Since we have at most $k^2 \leq (\sqrt{n}/\log n )^2=o(n)$ pairs of distinct opinions we can apply the union bound over all such pairs to show that all but a single opinion are super-weak within $t' = \BigO{\log n}$ phases \whp.
    In \cref{lem:all-super-weak} we show that the single remaining non-super-weak opinion wins within two additional phases.

    It remains to show that the winning opinion is one of the initially significant opinions. 
    Recall that $\Sigset{X(t)}$ denotes the set of all significant opinions at the start of phase $t$.
    We show for $t' = \BigO{\log n}$ that $\Sigset{X(t+t')} \subseteq \Sigset{X(t)}$.
    This means that no opinion which is insignificant in phase $t$ can become significant in phase $t+t'$. 
    As the subset relation is transitive, we have that
    \begin{align*}
        \Pr \left[ \Sigset{X(t+t')} \subseteq \Sigset{X(t)} \right] &\geq \Pr \left[ \forall t<t_1\leq t+t': \Sigset{X(t_1)} \subseteq \Sigset{X(t_1-1)}\right]\\ \nonumber
     &= 1 - \Pr \left[ \exists t<t_1\leq t+t': \Sigset{X(t_1)} \not \subseteq \Sigset{X(t_1-1)} \right].
    \end{align*}
    Furthermore, from \cref{lem:initial-bias} we have that for any $t_1\ge 0$ it holds $\Pr[ \Sigset{X(t_1)}\subseteq \Sigset{X(t_1-1)}]\ge 1-n^{-\BigOmega{1}}$. Together with union bound application, this implies that
    \begin{align*}
        1 &- \Pr \left[ \exists t<t_1\leq t+t': \Sigset{X(t_1)} \not \subseteq \Sigset{X(t_1-1)} \right] \\
       & \geq 1 - \sum_{t_1=t+1}^{t+t'} \Pr \left[ \Sigset{X(t_1)}  \not \subseteq \Sigset{X(t_1-1)} \right] \\
        &\geq 1-t'\cdot n^{-\BigOmega{1}}\ge 1-n^{-\BigOmega{1}}.
    \end{align*}
    Therefore, we have \whp that $\Sigset{X(t+t')} \subseteq \Sigset{X(t)}$.
    In the first part of the proof we also established that, \whp, only a single opinion $i$ remains in  phase $t +t'$.
    Clearly this single remaining opinion $i$ is significant in $\X(t + t')$, or in other words $i \in \Sigset{X(t + t')}$. As $\Sigset{X(t+t')} \subseteq \Sigset{X(t)}$ \whp, this implies that $i \in \Sigset{X(t)}$ and the result follows.
\end{proof}

% The idea behind the proof of \cref{thm:one-weak} is the following. We fix two strong opinions $i$ and $j$. The first steps is to establish a difference in support of $\Omega(\sqrt{n \log n})$ between $i$ and $j$. As long as the difference $|X_i(t) - X_j(t)| = \BigO{\sqrt{n \log n}}$, we call a phase $t$ successful if and only if $|X_i(t+1) - X_j(t+1)| > (5/4) \cdot |X_i(t) - X_j(t)|$. In \cref{lemma:drift_first_part}, we show that the probability for not being successful is exponentially small in $(X_i(t) - X_j(t))^2 / n$.
% In case the initial difference lies in $\Omega(\sqrt{n})$, it is easy to see that a sequence of $\BigO{\log \log n}$ successful phases suffices to increase the difference to $\Omega(\sqrt{n \log n})$. In case the initial difference is too small or an unsuccessful phase occurs, then the first statement of \cref{lemma:drift_first_part} shows that a difference of $\Omega(\sqrt{n})$ can quickly be recovered.
% This enables us to use the drift result of \cite{DBLP:conf/spaa/DoerrGMSS11}, and we obtain that in $\BigO{\log n}$ many phases there exists \whp a sequence of $\BigO{\log\log n}$ successful phases.

% As soon as the difference between $i$ and $j$ is larger than $\Omega(\sqrt{ n \log n})$ we repeatedly apply \cref{lemma:ratio} which implies that the ratio between opinions $i$ and $j$ grows until one opinion becomes weak (or even super-weak).

\begin{lemma}
\label{thm:one-weak}
Fix $\X(t) = \x(t)$ and any two distinct strong opinions $i$ and $j$.
Then at least one of them will become weak or super-weak within \BigO{\log{n}} phases with probability at least $1-\BigO{n^{-1.9}}$.
\end{lemma}
\begin{proof}[Proof]
    First we show that at least one of opinions $i$ and $j$ will become weak or super-weak.
    We consider the difference between both opinion $i$ and $j$ via a case study.
    If the difference is $\LittleO{\sqrt{n\log{n}}}$,
    we apply the drift result from \cref{lemma_drift_markov_chain} to increase the difference up to $\BigOmega{\sqrt{n\log{n}}}$. To be more precise, we map the difference $\lvert X_i(t)-X_j(t) \rvert$ to the state space of $W(t) =  \lfloor \lvert X_i(t)-X_j(t) \rvert/(c_a \cdot \sqrt{n}) \rfloor \in \{0,\dots,(\sDelta/c_a) \cdot \sqrt{\log{n}}\}$ where $c_a$ originates from \cref{lemma:drift_first_part}.
    Observe that, for some $t'$, $W(t')= (\sDelta/c_a) \cdot \sqrt{\log{n}}$ implies that $\lvert X_i(t')-X_j(t') \rvert \geq \sDelta \cdot \sqrt{n\log{n}}$.
    Now we deal with the two requirements within the drift result with the help of \cref{lemma:drift_first_part}.
    The first requirement is fulfilled by the first result in \cref{lemma:drift_first_part}.
    That is,
    \begin{align*}
        \Pr\left[W(t+1) \geq 1\right] 
        & \geq \Pr\left[\lfloor \lvert X_i(t+1)-X_j(t+1) \rvert \rfloor \geq c_a \cdot \sqrt{n}\right] = \BigOmega{1} .
    \end{align*}
    The second requirement is fulfilled by the second result in \cref{lemma:drift_first_part}.
    Assuming $c_a \cdot \sqrt{n} \leq \lvert X_i(t)-X_j(t) \rvert \leq \sDelta \cdot \sqrt{n\log{n}}$, it holds for a suitable constant $c_2 > 0$ that
    \begin{align*}
        \MoveEqLeft \Pr\left[W(t+1) \geq \min\{(5/4) \cdot W(t), (\sDelta/c_a) \cdot \sqrt{\log{n}}\}\right] \\
        &\geq \Pr\left[ \lvert X_i(t+1)-X_j(t+1) \rvert \geq \min\{ (1+\varepsilon)\cdot \lvert x_i - x_j \rvert , \sDelta \cdot \sqrt{n\log{n}}\} \right] \\
        &\geq 1 - 14\cdot \exp\left(-\concConstEps \cdot (\lvert x_i - x_j \rvert)^2 / 16n\right) \\
        &\geq 1 - 14\cdot \exp\left(-\concConstEps/16 \cdot c_a \cdot  (\lvert x_i - x_j \rvert) / \sqrt{n} \cdot (\lvert x_i - x_j \rvert)/c_a \sqrt{n} \right) \\
        &\geq 1 - \exp\left(-c_2 \cdot W(t)\right)
    \end{align*}
    
    Thus, due to the drift result, the difference between opinion $i$ and $j$ is at least $\sDelta \cdot \sqrt{n\log{n}}$ in $\BigO{\log{n}}$ phases. 

    \medskip
    Now, we assume that $\lvert X_i(t) - X_j(t) \rvert \geq \sDelta \cdot \sqrt{n \log n}$.
    Since both opinions are strong  and their difference is sufficiently large, we apply \cref{lemma:ratio}, which yields \2whp that
    \begin{equation}
        \label{eq:iterate-bias-increase}
        \frac{X_i(t+1)}{X_j(t+1)} \geq \left( \frac{x_i}{x_j}\right)^{1.5}.
    \end{equation}
    Our goal is to repeatedly apply above result over multiple phases. We start by establishing that \cref{lemma:ratio} also may be applied in the following  phase. To this end, we need to check the two conditions which fulfill the requirements of \cref{lemma:ratio}: (i) $X_i(t+1) - X_j(t+1) \geq \sDelta \cdot \sqrt{n \log n}$ holds, and (ii) opinions $i$ and $j$ remain strong opinions in $\X(t+1)$. Note that the lemmas statement immediately follows in case (ii) is violated. In the following, we will establish that (i) indeed holds. We apply \cref{lemma:one-phase-conc}  to both $i$ and $j$ and set $\delta = \sqrt{(\ln 7+2\log n)/\concConstEps}$. This way, we get, \2whp, that 
    \begin{align*}
        X_i(t+1)-X_j(t+1)
            &\geq \frac{1}{\Decided} \cdot \left( x_i^2 - x_j^2 - (x_i + x_j) \cdot \sqrt{(\ln 7+2\log n)/\concConstEps} \cdot \sqrt{n} \right) \\
            &\geq \frac{x_i+x_j}{\Decided} \cdot \left( (x_i - x_j) -  \sqrt{(\ln 7+2\log n)/\concConstEps} \cdot \sqrt{n} \right)
    \end{align*}
    Since both opinions are strong and $\Decided \leq x_{max}$ (see \cref{obs:decided}), it follows that $(x_i + x_j) / \Decided > 9/5$. This, together with above inequality chain implies that, indeed, $X_i(t+1)-X_j(t+1) > \sDelta \cdot \sqrt{n \log n}$. Above argument can easily be translated into an induction, which yields that $X_i(t + t') - X_j(t + t') \geq  \sDelta \cdot \sqrt{n \log n}$ until a round $t +t'$ is reached where opinion $j$ becomes weak (i.e., condition (ii) above is violated).
    Note that, \whp, $t' = \BigO{\log n}$ must hold as otherwise it follows by a repeated application of (\ref{eq:iterate-bias-increase}) that 
    \[
        \frac{X_i(t+t')}{X_j(t+t')} \geq \left(1 + \frac{\sDelta \cdot \sqrt{n \log n}}{x_j}\right)^{1.5^{t'}} \gg  \left(1 + \frac{1}{\sqrt{n}}\right)^{1.5^{t'}} > n.
    \]
    Hence opinion $j$ will become either weak or super-weak within $\BigO{\log n}$ phases with probability at least $1-\BigO{n^{-1.9}}$.
\end{proof}

\begin{lemma}
\label{lemma:drift_first_part}
Fix $\X(t) = \x(t)$ and two distinct strong opinions $i$ and $j$.
Let $c_a = \max\{1/\polyaConstSmall ~,~ 100\}$. 
\begin{enumerate}
    \item If $\lvert x_i(t) - x_j(t) \rvert < c_a \cdot \sqrt{n}$, then
$
    \Pr\left[\lvert X_i(t+1)-X_j(t+1) \rvert \geq c_a \cdot \sqrt{n}\right] = \BigOmega{1} .
$
\item If $c_a \cdot \sqrt{n} \leq \lvert x_i(t) - x_j(t) \rvert < \sDelta \cdot \sqrt{n\log{n}} $, then
\[
    \Pr\left[\lvert X_i(t+1)-X_j(t+1) \rvert \geq (5/4)\cdot (x_i(t) - x_j(t))\right] \geq  1 - 14\cdot \exp\left(-\concConstEps \cdot (\lvert x_i(t) - x_j(t) \rvert)^2 / 16n\right) .
\]
\end{enumerate}
\end{lemma}
\begin{proof}[Proof]
    We track the difference between two strong opinions $i$ and $j$ throughout a single phase.
    We  start with the first statement and assume, w.l.o.g., that $x_i \geq x_j$.
    The idea is to apply anti-concentration results  during the \Decision part in order to establish a sufficiently large difference between the support of both opinions and to roughly maintain this difference  throughout the following \Boosting part. \\
    At first we analyze the \Decision part and consider $Y_i(t), Y_j(t)$ and  $\norm{\config{Y}(t)}_1$.
    Recall (see \cref{obs:polya}) that $Y_i(t) \sim \BinDistr(x_i,x_i/n)$, $Y_j(t) \sim \BinDistr(x_j,{x_j}/{n})$ and $\norm{\config{Y}(t)}_1$ can be modeled as a sequence of $n$ Poisson trials.
    We define the event $\mathcal{E}$ as follows 
    \begin{align*}
	    \mathcal{E} = \Big\{ Y_i(t) \geq \frac{x_i^2}{n} +\frac{10}{7} \cdot c_a \cdot\frac{x_i}{\sqrt{n}} \mbox{ and } 
	     \left(1-\frac{2}{\sqrt{\log n}}\right)\cdot \frac{x_j^2}{n} \leq Y_j(t) \leq \frac{x_j^2}{n} -\frac{10}{7} \cdot c_a \cdot\frac{x_j}{\sqrt{n}} \\ \mbox{ and }
	    \norm{\config{Y}(t)}_1 = \Decided\cdot \left(1 \pm \frac{6}{\sqrt{\log n}} \right) \Big\}.        
    \end{align*}

	Now we bound probability of $\Bar{\mathcal{E}}$.
    We apply \cref{lemma:reverse_chernoff_v2} to both opinions $i$ and $j$ with  $\delta_i = ((10/7) \cdot c_a \sqrt{n})/x_i$ and $\delta_j = ((10/7) \cdot c_a \sqrt{n})/x_j$ and yield  
    \begin{align}
        \Pr\left[Y_i(t) \geq \frac{x_i^2}{n} +\frac{10}{7} \cdot c_a \cdot\frac{x_i}{\sqrt{n}}\right] &\geq \exp\left(-9\cdot ((10/7) \cdot c_a)^2 \right)
        \geq \exp\left(-20\cdot c_a^2 \right)
        \label{eq:drift_strong_anti_positve}\\
        \Pr\left[Y_j(t) \leq \frac{x_j^2}{n} -\frac{10}{7} \cdot c_a \cdot\frac{x_j}{\sqrt{n}}\right] &\geq \exp\left(-9\cdot ((10/7) \cdot c_a)^2 \right)
        \geq \exp\left(-20\cdot c_a^2 \right)
        \label{eq:drift_strong_anti__negative}
    \end{align}
    Next we apply Chernoff Bound(\cref{lemma:chernoff_poisson_trials}) to  $\norm{\config{Y}(t)}_1$ and $Y_j(t)$ where in the latter case we only need an additional lower bound. 
    Note that $\Ex{\norm{\config{Y}(t)}_1} \geq x_{max}^2 / n = \BigOmega{\log^2 n}$ due to  $x_{max} \geq n/k$ and $k \leq \sqrt{n}/\log{n}$.
    Hence, for $\delta' = 6/\sqrt{\log{n}}$ we get 
    \[
        \Pr\Big[\norm{\config{Y}(t)}_1 \le \Decided\cdot\left(1-\delta'\right)\Big] 
        \leq n^{-2} 
   \text{ and }
        \Pr\Big[\norm{\config{Y}(t)}_1 \ge \Decided\cdot\left(1+\delta'\right)\Big] 
        \leq n^{-2} . 
    \]
    In case of $Y_j(t)$ we use the fact that opinion $j$ is strong and hence, similar to the previous case,  $x_j \geq 0.9 \cdot x_{max} \geq 0.9 \cdot \sqrt{n}\log{n}$.
    Again we apply Chernoff Bound(\cref{lemma:chernoff_poisson_trials}) for $\delta' = \sqrt{2\cdot n \cdot \log{n}}/x_j \leq 4/\sqrt{\log n}$ and yield
    \[
        \Pr\left[ Y_j(t) \leq \left(1-\frac{4}{\sqrt{\log n}}\right)\cdot \frac{x_j^2}{n} \right] 
        \leq \Pr\left[ Y_j(t) \leq \left(1-\frac{\sqrt{2\cdot n\cdot\log n}}{x_j}\right)\cdot \frac{x_j^2}{n} \right]
        \leq n^{-2} .
    \]
    An application of the union bound on the bounds of $Y_j(t)$ yields 
    \begin{align*}
        \Pr\left[ Y_j(t) \notin \left(\left(1-\frac{4}{\sqrt{\log n}}\right)\cdot \frac{x_j^2}{n}, \frac{x_j^2}{n} -\frac{10}{7} \cdot c_a \cdot \frac{x_j}{\sqrt{n}}\right)  \right] 
        &\leq 1- \exp\left(-9\cdot ((10/7) \cdot c_a)^2 \right) + n^{-2} \\ 
        &\leq 1-\left(\exp\left(-10\cdot  c_a^2 \right) - n^{-2}\right) .
    \end{align*}
    Since $Y_{i}(t)$ and $Y_j(t)$ are independent,
    an another application of the union bound yields
    \begin{equation}
        \label{eq:drift_prep_union_bound_decisionPhase}
        \Pr\left[\Bar{\mathcal{E}}\right] \leq 1- \left( \left(\exp\left(-10\cdot  c_a^2 \right) - n^{-2}\right) \cdot \exp\left(-20\cdot  c_a^2 \right) \right) + 2n^{-2} = p < 1 
    \end{equation}
    where $p<1$ is a constant probability.  \\
    Now we deal with the outcome of the \Boosting part conditioned on the event $\mathcal{E}$.
    We fix $Y_i(t) = y_i$, $Y_j(t) = y_j$ and define $d=\norm{\config{Y}(t)}_1$.
    Again, recall (see \cref{obs:bin,obs:polya}) that $X_i(t+1) \sim \PE(y_i,d-y_i,n-d)$ and $X_j(t+1) \sim \PE(y_j,d-y_j,n-d)$.
    We apply the tail bound for this Pólya Eggenberger distribution from \cref{thm:polya_bound_total_balls}. 
    Note that conditioned on the event $\mathcal{E}$ and the previous observations about strong opinions we have $\sqrt{y_i} \geq \sqrt{y_j} \geq 0.9\cdot  \sqrt{1-4/\sqrt{\log n}} \cdot \log{n}$.
    Thus, clearly $\delta = (2 \cdot c_a)/7 < \sqrt{y_j}$ and hence we get
    \begin{align}
        \Pr\left[X_i(t+1) < y_i \cdot \frac{n}{d} - \sqrt{y_i} \cdot\delta \cdot \frac{n}{d} ~\Big|~ \mathcal{E}\right] &\leq 4\exp\left({-\polyaConstSmall \cdot \delta^2}\right) \leq 4\exp\left(-4\cdot \polyaConstSmall   \cdot c_a^2 / 49\right)  \label{eq:drift_strong_polya_positive} \\
        \Pr\left[X_j(t+1) > y_j \cdot \frac{n}{d} + \sqrt{y_j} \cdot \delta \cdot \frac{n}{d} ~\Big|~ \mathcal{E}\right] &\leq 4\exp\left({-\polyaConstSmall \cdot\delta^2}\right) \leq 4\exp\left(-4\cdot \polyaConstSmall   \cdot c_a^2 / 49\right) \label{eq:drift_strong_polya_negative}
    \end{align}

    Now we show that the difference between $X_i(t+1)$ and $X_j(t+1)$ is still sufficiently large.
    Conditioned on the event $\mathcal{E}$ we get that 
    \begin{align*}
        X_i(t+1)-X_j(t+1) 
        &\geq \frac{n}{d} \cdot \left( y_i -y_j -\delta \cdot (\sqrt{y_i}+\sqrt{y_j}) \right) \\
        &\geq \frac{n}{d} \cdot \Bigg( \frac{x_i^2 -x_j^2}{n} + \frac{10}{7}\cdot c_a \cdot \frac{x_i+x_j}{\sqrt{n}} \\ 
            &\qquad - \delta \cdot \frac{x_i}{\sqrt{n}} \cdot \left( \sqrt{1+\frac{10 \cdot c_a \sqrt{n}}{7\cdot x_i}} + \sqrt{1- \frac{10 \cdot c_a \sqrt{n}}{7\cdot x_j}} \right) \Bigg) \\
    \end{align*}
    By the definition of strong opinions  and the inequality $\sqrt{1-z} + \sqrt{1+z} \leq 2$ for $z \in (-1,+1), z  \neq 1$ it follows that
    \begin{align*}
        X_i(t+1)-X_j(t+1) 
        &\geq \frac{n}{d} \cdot \left( \frac{10}{7} \cdot c_a \cdot \frac{x_i +x_j}{\sqrt{n}} - 2 \cdot \delta \cdot \frac{x_i}{\sqrt{n}} \right) \\
        &\geq \sqrt{n} \cdot \frac{x_{max}}{d} \cdot \left( 2 \cdot \frac{9}{10} \cdot (10/7) \cdot c_a  - 2 \cdot \delta  \right) \\
        &\geq 2\cdot \sqrt{n} \cdot \frac{x_{max}}{d} \cdot \left( \frac{9}{7} \cdot c_a  -   \delta  \right)
    \end{align*}
    Again by the event $\mathcal{E}$ and the fact that $\Decided\leq x_{max}$ (see \cref{obs:decided}) it follows that
    \begin{align*}
        X_i(t+1)-X_j(t+1) 
        \geq \sqrt{n} \cdot \frac{2 \cdot x_{max}}{(1+6/\sqrt{\log n}) \cdot x_{max}} \cdot  \left( \frac{9}{7}  \cdot c_a  -  \frac{2}{7}\cdot c_a  \right) 
        \geq c_a \cdot \sqrt{n} . 
    \end{align*}
    An application of the union bound yields that the difference between $X_i(t+1)$ and $X_j(t+1)$ holds with probability $1-8\exp(-4\cdot \polyaConstSmall \cdot c_a^2 / 49)$.
    At last we combine this with (\ref{eq:drift_prep_union_bound_decisionPhase})  via an application of the law of total probability to deduce that the first statement follows with constant probability.
    Due the choice of $c_a$ it holds, with at least constant probability, that 
    \begin{align*}
        \MoveEqLeft \Pr \left[\lvert X_i(t+1)-X_j(t+1) \rvert \geq c_a \cdot \sqrt{n} \right]  \\
        &=
        \Pr \left[ \lvert X_i(t+1)-X_j(t+1) \rvert \geq c_a \cdot \sqrt{n} ~\Big|~ \mathcal{E}
        \right]\cdot \Pr\left[\mathcal{E}\right] 
        \\&\quad+
         \Pr \left[ \lvert X_i(t+1)-X_j(t+1) \rvert \geq c_a \cdot \sqrt{n} ~\Big|~ \bar{\mathcal{E}}
        \right]\cdot \Pr\left[\bar{\mathcal{E}}\right]
        \\
        &\geq  \Pr \left[ \lvert X_i(t+1)-X_j(t+1) \rvert \geq c_a \cdot \sqrt{n} ~\Big|~ \mathcal{E}
        \right]\cdot \Pr\left[\mathcal{E}\right] \\
        &\geq \left(1-8 \cdot \exp\left(-4\cdot \polyaConstSmall \cdot c_a^2 / 49\right)\right) \cdot \left( \left( \left(\exp\left(-10\cdot  c_a^2 \right) - n^{-2}\right) \cdot \exp\left(-20\cdot  c_a^2 \right) \right) - 2n^{-2}\right)
    \end{align*}
    \medskip
    
    We continue with the proof of the second statement and assume  $x_i-x_j\geq c_a \cdot \sqrt{n} $.
    We apply \cref{lemma:one-phase-conc} to both opinions with $\delta = (x_i -x_j)/4\cdot\sqrt{n} $ and yield 
    \begin{align*}
        \Pr\left[X_i(t+1) \leq \frac{x_i^2}{\Decided} - \frac{x_i}{\Decided} \cdot \sqrt{n} \cdot\delta \right] &\leq 7\cdot\exp\left(-\concConstEps \cdot \delta^2\right) 
            = 7 \cdot\exp\left(-\concConstEps \cdot \frac{(x_i-x_j)^2}{16n}\right) \\
        \Pr\left[X_j(t+1) \geq \frac{x_j^2}{\Decided} + \frac{x_j}{\Decided} \cdot \sqrt{n}\cdot \delta\right] &\leq7\cdot\exp\left(-\concConstEps \cdot \delta^2\right) 
            = 7 \cdot\exp\left(-\concConstEps \cdot \frac{(x_i-x_j)^2}{16n}\right)
    \end{align*}
     Now, as $i$ and $j$ are assumed to be strong opinions, it follows that $(x_i + x_j) / \Decided > 9/5$.
     An application of the union bound yield the second statement, with probability at least $ 1 - 14\cdot \exp\left(-\concConstEps \cdot (\lvert x_i - x_j \rvert)^2 / 4n\right)$,  due to 
    \begin{equation*}
        X_i(t+1)-X_j(t+1)
            \geq \frac{x_i + x_j}{\Decided}\cdot \left(1-\frac{\sqrt{n}\cdot\delta}{x_i-x_j} \right) \cdot (x_i - x_j)
            \geq \frac{9}{5} \cdot \left(1-\frac{1}{4}\right) \cdot(x_i - x_j) \geq (5/4) \cdot  (x_i-x_j). \qedhere
    \end{equation*}
\end{proof}

% The proofs of \cref{lem:weak-to-super-weak,lem:all-super-weak} rely on the concentration results from \cref{lemma:one-phase-conc,lemma:weak-phase-conc,lemma:ratio}.
% Additionally, in \cref{lem:all-super-weak} we use the following observation: if opinion $i$ is the only one that is not super-weak, then it has support $n\cdot (1-\LittleO{1})$ and it will eventually win, \whp.
\begin{lemma}
\label{lem:weak-to-super-weak}
    Fix $\X(t) = \x(t)$ and an opinion $j$.
    If opinion $j$ is weak, then it will become super-weak in $\BigO{\log\log{n}}$ phases with probability at least $1-\BigO{n^{-1.9}}$.
    If opinion $j$ is super-weak, then it will remain super-weak in at least $\BigOmega{\log^2{n}}$ following phases with probability more than $1-\BigO{n^{-1.9}}$. 
\end{lemma}
\begin{proof}[Proof]
    First we show that a weak opinion will become super-weak.
    Let opinion $j$ be weak but not super-weak opinion in $\x$, i.e., $c_w \cdot \sqrt{n\log n} < x_j < 0.9 \cdot x_{max}$. 
    As opinion $j$ is weak, it follows that the difference $x_{max} -x_j \geq n/(10\cdot k) = \omega(\sqrt{n \log n})$ is large enough and we can apply \cref{lemma:ratio}. This yields \2whp that 
    \begin{equation}
    \label{eq:w-t-sw}
        \frac{X_{max}(t+1)}{X_j(t+1)} \geq \left( \frac{x_{max}}{x_j}\right)^{1.5}.
    \end{equation}
    This result has two implications. For one, it states that the ratio between the largest opinion and opinion $j$ grows substantially. Second, it also implies that opinion $j$ cannot become strong in $\X(t+1)$. To see this, we combine (\ref{eq:w-t-sw}) together with the fact that $j$ is weak in $\x$
    \[
        X_j(t+1) \overset{(\ref{eq:w-t-sw})}{\leq } X_{max}(t+1)  \cdot \left(\frac{x_j}{x_{max}} \right)^{1.5} < X_{max}(t+1) \cdot \left(\frac{9}{10}\right)^{1.5} < X_{max}(t+1) \cdot 0.9.
    \]
    Hence, $j$ is either weak or super-weak in phase $t+1$. If $j$ is super-weak, we are done. Otherwise we may again apply \cref{lemma:ratio}. We follow this approach for $t' = \log_{1.5} \log_{9/10} n$ phases,  at which point, either (i) $j$ already became super-weak in some phase $< t + t'$, or (ii) it follows by the growth of the ratio that with probability at least $1-\BigO{n^{-1.9}}$
    \[
        \frac{X_{max}(t+t')}{X_i(t+t')}   \geq \left(\frac{10}{9} \right)^{1.5^{t'}} \geq n .
    \]
    This implies that, opinion $j$ must have already became super-weak.
    
    \medskip
    
    In this second part of the proof, we show that a super-weak opinion $j$ in $\x$ remains super-weak.
    We apply \cref{lemma:weak-phase-conc} on this opinion $j$ and yield 
    \begin{equation}
		\Pr \left[ X_j(t+1)  > \frac{n}{\Decided} \cdot (12c_w^2 + 74\polyaConstLarge) \log  n \right] < 4n^{-2}.
	\end{equation}
	Using that $\Decided = \Ex{\norm{\config{Y}(t)}_1} \geq n/k \geq \sqrt{n}  \log n$  then implies 
 	\begin{equation*}
		\Pr \Big[ X_j(t+1)  > (12c_w^2 + 74\polyaConstLarge) \cdot \sqrt{n}  \Big] < 4n^{-2}.
	\end{equation*}
	As $ (12c_w^2 + 74\polyaConstLarge) \cdot \sqrt{n}  = \LittleO{\sqrt{n \log n}}$ it follows that opinion $j$ remains super-weak at the start  phase $t+1$. A repetition of this argument, together with a union bound application yields that opinion $j$ will remain super-weak for at least $\BigOmega{\log^2{n}}$ phases with probability more than $1-\BigO{n^{-1.9}}$.
\end{proof}

\begin{lemma}
\label{lem:all-super-weak}
    Fix $\X(t) = \x(t)$.
    If all but a single opinion $i$ are super-weak, then only opinion $i$ remains in phase $t+2$ \whp. 
\end{lemma}
\begin{proof}[Proof]
    Let opinion $i$ be the only non super-weak opinion. 
    It follows that $x_i \geq n - (k-1)\cdot c_w\cdot \sqrt{n\log n} $. As we consider $k \leq \sqrt{n} /  \log n$, this implies that $x_i \geq n - n\cdot(c_w /  \sqrt{\log n}) = n\cdot (1-\LittleO{1})$. Furthermore, as $\Decided= \Ex{\norm{\config{Y}(t)}_1} = \sum_{j=1}^{k} x_j^2 / n \geq x_i^2 / n$ this also implies that $\Decided= \BigOmega{n}$.
	
	Now we fix some opinion $j$ that is super-weak. We apply \cref{lemma:weak-phase-conc} and use that $\Decided = \BigOmega{n}$ which immediately yields 
	\[
		\Pr \left[ X_j(t+1)  = \omega( \log n) \right] < 4n^{-2}.
	\] 
    Such an opinion $j$ will, \whp, not have a single decided agent at the start of \Boosting part $t+1$ because
 	\begin{align*}
		\Pr \Big[ Y_j(t+1) > 0 ~\Big|~ X_j(t+1) = \BigO{\log n} \Big] & <  1 - \left(1 - \frac{\BigO{\log n}}{n}\right)^{\BigO{\log n}} \\
		& <  \frac{\polylog n}{n}.
	\end{align*} 
	Therefore, $j$ will vanish before phase $t+2$ \whp.
	A simple union bound argument now yields that, \whp, \emph{none}  of the opinions which are super-weak at the beginning of phase $t$ will survive until phase $t+2$. 
\end{proof}

\subsection{Consensus for $\sqrt{n}/\log{n} <k \leq \sqrt{n}/c_k$}
\label{sec:details-analysis-case2}

\lemDriftSecondCaseK*
\begin{proof}[Proof]
    First we establish, $X_{max}(t+t_1)\ge \sqrt{n\log n}$ for some $t_1=\BigO{\log n}$ with the drift result from \cref{lemma_drift_markov_chain}.
    In contrast to the proof of \cref{thm:one-weak} we apply the result on the largest opinion.
    To be more precise, we map the largest opinion with support $X_{max}(t)$ to the state space of $W(t) = \lfloor X_{max}(t)/(\sqrt{n}) \rfloor \in \{0,\dots,\sqrt{\log{n}}\}$.
    Observe that $W(t)= \sqrt{\log{n}}$ implies  that $X_{max}(t) \geq \sqrt{n\log{n}}$.
    Now we deal with both requirements of the drift result.
    The first requirement in is surely true due to the regime of $k$ we consider. That is, we have that  $X_{max}(t)\ge n/k \ge c_k\sqrt{n} \ge \sqrt{n}$ and hence 
    \[
        \Pr[ W(t+1) \geq 1] = \Pr[\lfloor X_{max}(t+1)/(\sqrt{n}) \rfloor \geq 1]
    \]
    holds with at least constant probability.
    The second requirement  is fulfilled by the first result in \cref{lem:k-to-sqrtlog}.
    As long as $X_{max}(t)$ is smaller than $\sqrt{n\log{n}}$, it holds for a suitable constant $c_2 > 0$ that 
    \begin{align*}
        \Pr[W(t+1) \geq \min\{(1+1/60) W(t), \log{n}\}] 
        &\geq \Pr[X_{max}(t+1) \geq \min\{(1+1/60) x_{max}, \sqrt{n\log{n}}\}] \\
        &\geq 1- 7\exp(-(\concConstEps/625) \cdot (x_{max}^2 / n)) \\
        &\geq 1- \exp(-c_2 \cdot \lfloor x_{max}/(\sqrt{n}) \rfloor) \\
        &= 1- \exp(-c_2 \cdot W(t)) .
    \end{align*}
    Thus, due to the drift result, the support of the largest opinion is at least $\sqrt{n\cdot\log{n}}$ within $\BigO{\log{n}}$ phases \whp.

    Now, we continue with a configuration which satisfies $X_{max}(t+t_1) \geq \sqrt{n\log n} $.
    Here, we repeatedly apply the second statement of \cref{lem:k-to-sqrtlog}.
    It states that the largest opinion increases by an additional term of size at least $(1/60)\cdot \sqrt{n\log{n}}$.
    Thus, for $t_2=\BigO{\log n}$ phases, it follows that $X_{max}(t+t_1+t_2)\ge \sqrt{n}\log ^{3/2}n$  \whp.
\end{proof}

\lemLogToEnd*
\begin{proof}[Proof] 
Starting with the configuration $\x$ at \Decision part $t$, we consider two cases.

	\paragraph{Case 1} $\Decided(t) = \Ex{\norm{\config{Y}(t)}_1} \leq \frac{1}{2}\sqrt{n} \log^{3/2} n$. Let $i$ be an opinion which provides $x_{max}$. As in this case $x_i \geq 2 \cdot \Decided(t)$, we will show that $i$ gathers additional support until \Decision part $t+1$. To that end, we apply \cref{lemma:one-phase-conc}  with $\delta = \sqrt{(5/\concConstEps) \log n}$. This immediately yields 
	\begin{equation}
	\label{eq:lem18-1}
	\Pr \left[X_i(t+1) > \frac{x_i^2}{\psi(t)} - \sqrt{5/\concConstEps} \cdot \frac{x_i}{\psi(t)} 			\sqrt{n \log n} \right]  \geq 1  - n^{-4}.
	\end{equation}
	Remember, we have  $x_i \geq 2 \Decided(t)$.  We use this in the second step of the following calculation 
	\[
	\frac{x_i^2}{\Decided(t)} - \sqrt{5/\concConstEps} \cdot \frac{x_i}{\Decided(t)} \sqrt{n \log n} = \frac{x_i}{\Decided(t)} \left(x_i - \BigO{\sqrt{n \log n}} \right) \geq 2x_i - \BigO{\sqrt{n \log n}} > x_i (1 +  1/2).
	\]
	As for arbitrary random variables $X$ and $x \leq x'$, it holds that $\Pr[X > x] \leq \Pr[X > x']$, we get 
	\begin{equation*}
		 \Pr \left[X_i(t+1) > x_{i} (1 + 1/2) \right] \geq \Pr \left[X_i(t+1) > \frac{x_i^2}{\psi(t)} - \sqrt{5/\concConstEps} \cdot \frac{x_i}{\psi(t)} 			\sqrt{n \log n} \right] \overset{(\ref{eq:lem18-1})}{\geq}  1- n^{-4}.
	\end{equation*}
	Clearly it holds that $X_{max}(t+1) \geq X_i(t+1)$. Therefore above result implies that, \whp,  $X_{max}(t+1) > x_{i} (1 + 1/2) = x_{max} (1 + 1/2)$. Next, we apply \cref{lem:max-not-shrinking}, which states that, \whp, the maximum shrinks by at most an $(1-\LittleO{1})$ factor with every phase. Therefore, \whp, the following holds 
 \begin{align*}
	X_{max}(t+3) \overset{\cref{lem:max-not-shrinking}}{>} X_{max}(t+2) (1-\LittleO{1}) \overset{\cref{lem:max-not-shrinking}}{>} X_{max}(t+1) (1-\LittleO{1}) \\
	> x_{max}(1 + \frac{1}{2}) \cdot (1- \LittleO{1}) > x_{max} (1 + \frac{1}{4}).
\end{align*}
In the second line we used that, \whp, $X_{max}(t+1) > x_{max}(1 + 1/2)$ as argued above.
Note that this inequality chain implies the first statement of the lemma.

	\paragraph{Case 2} $\Decided(t) > \frac{1}{2}\sqrt{n} \log^{3/2} n$.
	In this case, we define $L(t)$ as the set of opinions $j$ with support at most $x_j \leq 2\sqrt{n} \log n$. We consider a fixed opinion $j \in L(t)$ and will show that $X_j(t+1) = \BigO{\sqrt{n \log n}}$. In other words, such an opinion $j$ will shrink by an $\BigOmega{\sqrt{\log n}}$ factor.
	To that end, we model $Y_j(t) \sim \Bin(x_j, x_j / n)$ as usual (see \cref{obs:bin}). As $j \in L(t)$, we have $\Ex{Y_j(t)} \leq 4 \log^2 n$. We apply Chernoff bounds  (\cref{lemma:chernoff_poisson_trials}) with $\delta =2/ \sqrt{\log n}$ and get 
	\begin{equation}
	\label{eq:lem18-3}
		\Pr \left[Y_j(t) > 4\log^2 n  \left( 1 + \frac{2}{\sqrt{\log n}} \right) \right] \leq \exp \left(- \frac{\delta^2 \cdot 4\log^2 n}{3} \right) =  \exp \left(-\frac{16\log n}{3} \right) < n^{-5}.
	\end{equation} 
	Similar, we model $\norm{\config{Y}(t)}_1 = \sum_{i=1}^{k} Y_i(t)$ as the sum of Poisson trials where $\Ex{\norm{\config{Y}(t)}_1} = \Decided(t) \geq \frac{1}{2} \sqrt{n} \log^{3/2} n$. Just as above we employ Chernoff bounds with $\delta = 2/ \sqrt{\log n}$ and derive
	\begin{equation}
	\label{eq:lem18-4}
		\Pr \left[\norm{\config{Y}(t)}_1 <\frac{1}{2} \sqrt{n} \log^{3/2} n (1 - \LittleO{1}) \right] < n^{-5}.
	\end{equation}
	We now define the following event for opinion $j \in L(t)$
	\begin{equation*}
		\mathcal{E}_j :\Leftrightarrow \left\{ Y_j(t) < 4\log^2 n \cdot  ( 1 + \LittleO{1}) ~ \land ~ \norm{\config{Y}(t)}_1  > \frac{1}{2} \sqrt{n} \log^{3/2} n \cdot (1 - \LittleO{1}) \right\}.
	\end{equation*}
	A union bound application together with (\ref{eq:lem18-3}) and (\ref{eq:lem18-4}) yields that $\Pr[\mathcal{E}_j] > 1 - 2n^{-5}$. We will now track opinion $j$ throughout \Boosting part $t$. We consider fixed $Y_j(t) = y_i$ and define $d=\norm{\config{Y}(t)}_1$, and as noted in \cref{obs:polya} we model $X_j(t+1) \sim \PE(y_i, d-y_i, n -d)$. From the concentration inequality in \cref{thm:polya_bound_small_support} we then get that 
	\begin{equation}
    \label{eq:thm31-k}
    	\Pr \left[X_{j}(t+1) > \frac{n}{d}\cdot(3y_j + \BigO{\log n}) ~\Big|~ \mathcal{E}_j \right] < 2n^{-2}.
    \end{equation}
	As we only consider values of $y_i$ and $d$ that fulfill $\mathcal{E}_j$ we have that
	\[
		\frac{n}{d}\cdot(3y_i + \BigO{\log n}) < 24 \sqrt{n \log n} \cdot (1 + \LittleO{1}) < 25 \sqrt{n \log n}.
	\]
	We use this inequality in the first step and then (\ref{eq:thm31-k}) in the second to deduce that 
	\[
		\Pr \left[X_{j}(t+1) >25 \sqrt{n \log n} ~\Big|~ \mathcal{E}_j \right]  < \Pr \left[X_{j}(t+1) > \frac{n}{d}\cdot(3y_i + \BigO{\log n}) ~\Big|~ \mathcal{E}_j \right]  \overset{(\ref{eq:thm31-k})}{<} 2n^{-2}.
	\]
	An application of the law of total probability yields
	\begin{align*}
		 \Pr \left[X_{j}(t+1)  \leq 25 \sqrt{n \log n} \right]  &\geq \Pr \left[X_{j}(t+1) \leq 25 \sqrt{n \log n} ~\Big|~ \mathcal{E}_j \right] \cdot \Pr \left[\mathcal{E}_j \right] \\
		& \geq (1-2n^{-2}) \cdot \Pr \left[\mathcal{E}_j \right] > (1-2n^{-2}) \cdot (1-2n^{-5}) \geq (1 - 3n^{-2}).
	\end{align*}
	Finally, we apply union bounds over all opinions $j \in L(t)$ which yields that
	\[
		\Pr \left[ \forall j \in L(t): X_j(t) \leq 25 \sqrt{n \log n} \right] \geq (1 -3 n^{-2} \cdot |L(t)|) > (1 - n^{-1}).
	\]
	Additionally, we again apply \cref{lem:max-not-shrinking}, which states that $X_{max}(t+1) > x_{max} (1-\LittleO{1})$ \whp. When using union bounds, we get that the following event holds \whp
	\[
		\mathcal{E} :\Leftrightarrow \left\{ \forall j \in L(t): X_j(t) \leq 25 \sqrt{n \log n} ~ \land ~  X_{max}(t+1) > x_{max} (1 -\LittleO{1}) \right\}.
	\]
	
	\medskip
	
	We now continue our analysis for one further phase. To that end, we fix the configuration $\X(t+1) = \bar{\x}$ and assume that $\bar{\x}$ fulfills the event $\mathcal{E}$. Again we distinguish two cases.
	
	\paragraph{Case 2.1} $\Decided(t+1) := \Ex{Y_i(t+1)} \leq \frac{1}{2}\sqrt{n} \log^{3/2} n$. This is mostly a repetition of case 1, replacing $t$ with $t+1$ and $\x$ with $\bar{\x}$. Let $i$ be an opinion which provides $ \bar{x}_{max}$. As in case 1, we apply \cref{lemma:one-phase-conc} and get for $\delta = \sqrt{(5/\concConstEps) \cdot \log n}$ that 
	\begin{equation}
	\label{eq:lem18-6}
			\Pr \left[X_i(t+2) > \frac{\bar{x}_i^2}{\psi(t+1)} - \sqrt{5/\concConstEps} \cdot \frac{\bar{x}_i}{\psi(t+1)} 			\sqrt{n \log n} \right]  > 1  - n^{-4}.
	\end{equation}
	As $\bar{\x}$ fulfills the event $\mathcal{E}$, we have $\bar{x}_i \geq x_{max} (1-\LittleO{1}) > 2 \Decided(t) \cdot (1 - \LittleO{1})$. We can use this to show 
	\begin{align*}
	\frac{\bar{x}_i^2}{\Decided(t+1)} - \sqrt{5/\concConstEps} \cdot \frac{\bar{x}_i}{\Decided(t+1)} \sqrt{n \log n} &= \frac{\bar{x}_i}{\Decided(t+1)} \left(\bar{x}_i - \BigO{\sqrt{n \log n}} \right) \\
	&\geq (2 - \LittleO{1}) \cdot  (\bar{x}_i - \BigO{\sqrt{n \log n})} > \bar{x}_i (1 +  1/2).
	\end{align*}
	In combination with (\ref{eq:lem18-6}), this then implies $\Pr \left[X_{i}(t+2) > \bar{x}_i (1 + 1/2) \right] > 1-n^{-4}$.
	Just as in case 1, we then use $X_{max}(t+2) \geq X_i(t+2)$ and apply \cref{lem:max-not-shrinking} to argue that the maximum shrinks by at most an $(1-\LittleO{1})$ factor throughout phase $t+3$. This way, we get \whp
	\begin{align*}
		X_{max}(t+3) \overset{\cref{lem:max-not-shrinking}}{> } X_{max}(t+2) \cdot (1 - \LittleO{1})  \overset{\cref{lem:max-not-shrinking}}{> } \bar{x}_i \cdot (1+ 1/2) \cdot (1- \LittleO{1}) \\\ \overset{\mathcal{E}}{>} x_{max} \cdot (1 + 1/2) \cdot (1- \LittleO{1}) > x_{max} (1 + 1/4).
	\end{align*}
	In the penultimate step, we used that $\bar{x}_i = \bar{x}_{max}$ and that  $\bar{\x}$ fulfills event $\mathcal{E}$ (i.e., $\bar{x}_{max} > x_{max} (1 - \LittleO{1})$). The first statement of \cref{lem:log32-to-end} follows.
	
	\paragraph{Case 2.2}  $\Decided(t+1) > \frac{1}{2}\sqrt{n} \log^{3/2} n$. This case is similar to case 2. Remember, we currently consider the fixed configuration $\bar{\x}$ at \Decision part $t+1$. As we assume that $\bar{\x}$ fulfills $\mathcal{E}$, we have that every opinion $j \in L(t)$ has $\bar{x}_j \leq 25 \cdot \sqrt{n \log n}$. We will now show that, \whp, $X_j(t+2) = \BigO{\sqrt{n / \log n}}$ for every such opinion. We start by fixing some $j \in L(t)$. Just as in case 2 we model $Y_j(t+1) \sim \Bin(\bar{x}_j, \bar{x}_j / n)$. As $\Ex{Y_j(t+1)} \leq 625 \log n$, a Chernoff bound application with $\delta = 1/5$ yields 
	\begin{equation}
	\label{eq:lem18-21}
		\Pr \left[Y_j(t+1) > 625 \log n \cdot (1 + 1/5) \right] \leq \exp \left(-\frac{25 \log n}{3} \right) < n^{-5}.
	\end{equation}
	Additionally, just as argued for (\ref{eq:lem18-4}), we have $ \Ex{\norm{\config{Y}(t+1)}_1}= \Decided(t+1) > \frac{1}{2}\sqrt{n} \log^{3/2} n$ and therefore a Chernoff bound application yields
	\begin{equation}
	\label{eq:lem18-22}
		\Pr \left[\norm{\config{Y}(t+1)}_1 <\frac{1}{2} \sqrt{n} \log^{3/2} n (1 - \LittleO{1}) \right] < n^{-5}.
	\end{equation}
	Next, we track the evolution of opinion $j$ throughout \Boosting part $t+1$. We first  define
	\begin{equation*}
		\mathcal{E}_j :\Leftrightarrow \left\{ Y_j(t+1) \leq  750 \log n ~ \land ~ \norm{\config{Y}(t)}_1  > \frac{1}{2} \sqrt{n} \log^{3/2} n \cdot (1 - \LittleO{1}) \right\}.
	\end{equation*}
	The events (\ref{eq:lem18-21}) and (\ref{eq:lem18-22}), together with a union bound application, imply that $\Pr[\mathcal{E}_j] > 1 - 2n^{-5}$. We fix $Y_j(t+1) = \bar{y}_j$ and define $ \bar{d}=\norm{\config{Y}(t+1)}_1$. Just as in case 2, we model $X_j(t+2)$ with a Pólya-Eggenberger distribution and apply  \cref{thm:polya_bound_small_support} to derive (the constant \polyaConstLarge  originates from \cref{thm:polya_bound_small_support})
	\begin{equation}
    \label{eq:thm31-k2}
    	\Pr \left[X_{j}(t+2) > \frac{n}{\bar{d}}\cdot(3\bar{y}_j + \polyaConstLarge \cdot \log n) ~\Big|~ \mathcal{E}_j \right] < 2n^{-2}.
    \end{equation}
    As we condition on the event $\mathcal{E}_j$, we have $\bar{y}_j \leq 750 \log n$ and $d = \sqrt{n} \log^{3/2} n \cdot (1 - \LittleO{1})$. Therefore 
    \[
    	\frac{n}{\bar{d}} \cdot (3 \bar{y}_i + \polyaConstLarge \cdot \log n) < c' \cdot \sqrt{n / \log n},
    \]
    for the constant $c' = 2625 + \polyaConstLarge$. This further implies the first step in the following inequality chain 
    \[
    	 \Pr \left[X_j(t+2) > c' \cdot \sqrt{n / \log n}  ~\Big|~ \mathcal{E}_j  \right] \leq \Pr \left[X_{j}(t+2) > \frac{n}{\bar{d}}\cdot(3\bar{y}_j + \BigO{\log n}) ~\Big|~ \mathcal{E}_j \right] \overset{(\ref{eq:thm31-k2})}{<} 2n^{-2}.
    \]
    Just as in case 2, we apply the law of total probability to deduce that 
    \begin{align*}
    	\Pr \left[X_j(t+2) \leq c' \cdot \sqrt{n / \log n} \right] &\geq  \Pr \left[X_j(t+2) \leq c' \cdot \sqrt{n / \log n} ~\Big|~ \mathcal{E}_j \right] \cdot \Pr \left[\mathcal{E}_j \right] \\
    	&\geq (1-2n^{-2}) \cdot (1-2n^{-5}) > (1 - 3n^{-2}).
    \end{align*}
    When applying union bounds over all $j \in L(t)$ this implies that, \whp, 
    \[
    	\forall j \in L(t) : ~X_j(t+2) < c' \sqrt{n / \log n}
    \]
    The remaining analysis will now differ from case 2. We  show that after \Decision part $t+2$, most opinions in $L(t)$ will vanish. Observe that
    \[
    	\Pr \left[Y_j(t+2) = 0 ~|~ X_j(t+2) \leq c' \sqrt{n / \log n} \right] \geq \left(1 - \frac{c' \sqrt{n}}{n \cdot \sqrt{\log n}} \right)^{c' \sqrt{n / \log n}} \geq 1 - \frac{(c')^2}{\log n},
    \]
	where the probability on the right describes the situation where not a single agent of opinion $j$ manages to sample an agent of opinion $j$ throughout the \Decision part. 
	For every $j \in L(t)$ we now define an indicator random variable $Z_j$, where $Z_j= 1$ iff $Y_j(t+2) > 0$, and $Z_j = 0$ otherwise. This way, $Z:= \sum_{j \in L(t)} Z_j$ describes the number of opinions in $L(t)$ that survive. Above, we established that $\Pr[Z_j = 1] \leq (c')^2 / \log n$ which implies $\Ex{Z} \leq (c')^2 \cdot |L(t)| / \log n$. In order to apply the Chernoff bound in \cref{General_Upper_Chernoff_bound}, we further bound $\Ex{Z} \leq \max \{ (c')^2 \cdot |L(t)| / \log n~,~ \log n\} =: \mu_z$. By \cref{General_Upper_Chernoff_bound} we get for $\delta = 1/2$ that 
	\[
		\Pr \left[Z > \max \{ (c')^2 \cdot |L(t)| / \log n~,~ \log n\} \cdot (1 + 1/2)\right] < \exp\left(-\frac{\delta^2 \mu_z}{ 2+ \delta} \right) = n^{-\BigOmega{1}}
	\]
Next, we count the number of opinions that survive \Decision part $t+2$. We assume the worst-case of every opinion $i \not \in L(t)$ surviving and combine this with above result. This way, \whp, at most 
\begin{equation}
\label{eq:lem18-10}
	 (k - |L(t)|)  + \frac{3}{2} \max \left\{ \frac{(c')^2}{\log n} \cdot |L(t)| ~,~ \log n \right\}
\end{equation}
opinions remain. Remember, initially we defined $L(t)$ as the set of opinions of size at most $2\sqrt{n} \log n$. By a counting argument it follows that $k - |L(t)| \leq \sqrt{n} / (2 \log n)$ as at most $\sqrt{n} / (2 \log n)$ opinions can have support greater than $2\sqrt{n} \log n$. Additionally, due to the regime of $k$ we consider in this section, we have $|L(t)| \leq k < \sqrt{n} / c_k$ for some arbitrary large constant $c_k >0$. We use this to loosen the bound on the remaining opinions in (\ref{eq:lem18-10}) to  
\[
	 (k - |L(t)|)  + \frac{3}{2} \max \left\{ \frac{(c')^2}{\log n} \cdot |L(t)| ~,~ \log n \right\} \leq \frac{\sqrt{n}}{2 \log n} + \frac{3 (c')^2 \sqrt{n}}{2 c_k \log n} \leq \frac{\sqrt{n}}{\log n}.
\]
 The last step follows from the fact that $c_k=4 \cdot (2625 + c_p)^2 \geq  3(c')^2 = 3 \cdot (2625 + c_p)^2$ (see \cref{def:Constant}). 
Therefore, \whp, only $\sqrt{n} / \log n$ remain after \Decision part $t+2$ and the second statement of \cref{lem:log32-to-end} follows.
\end{proof}

\subsection{Consensus for $k > \sqrt{n}/c_k$}
\label{sec:details-analysis-case3}

\proCaseThree*
\begin{proof}[Proof]
Fix an arbitrary configuration $\X(t)=\x$ and assume that it has $k(t) > \sqrt{n} / c_k$ opinions with non-zero support. Let $\mathcal{K}_s$ denote the set of opinions that have support at most $2c_k \sqrt{n}$ in $\x$.
Conceptually, this set contains opinions of small support that are likely to vanish.
We fix some $i \in \mathcal{K}_s$ and model the number of decided agents of opinion $i$ as  $Y_i(t) \sim \Bin(x_i, x_i / n)$ (see \cref{obs:bin}). This way, we have that 
\begin{equation}
\label{eq:last-prop1}
	\Pr \left[Y_i(t) = 0 \right] = \left( 1 - \frac{x_i}{n} \right)^{x_i} \geq \left( 1 - \frac{2c_k}{\sqrt{n}}\right)^{2c_k \sqrt{n}} \geq e^{-4c_k^2} := p
\end{equation}
Hence, with constant probability $1>p>0$, such a small opinion $i$ will not have a single decided agent at the start of \Boosting part $t$.
We now define an indicator random variable $Z_i$ for each opinion $i \in \mathcal{K}_s$, where $Z_i =1$ iff $Y_i(t) = 0$ and $Z_i = 0$ otherwise.
This way, $Z := \sum_{i \in \mathcal{K}_s} Z_i$ describes the number of small opinions that have not a single decided agent at the start of \Boosting part $t$.
Note that the random variables $Z_i$ are independent and by (\ref{eq:last-prop1}) we have $\Ex{Z} \geq |\mathcal{K}_s| \cdot p$. 
Before applying Chernoff bounds on $Z$, we take a closer look at $|\mathcal{K}_s|$. By a counting argument, we have that
$|\mathcal{K}_s| > k(t) - \sqrt{n} / (2c_k)$ as at most $\sqrt{n} / (2c_k)$ opinions can have support of more than $2c_k \sqrt{n}$. Initially we assumed that $k(t) > \sqrt{n} /c_k$, this further implies that $|\mathcal{K}_s| > k/2$. Therefore, we can continue the bound on $\Ex{Z}$ as 
\[
	\Ex{Z} \geq |\mathcal{K}_s|  \cdot p \geq k(t) \cdot \frac{p}{2}
\]
A Chernoff bound application now easily shows that, \whp, $Z > k(t) \cdot \frac{p}{3} = \BigOmega{k(t)}$. Throughout the following \Boosting part $t$, no agent will adopt any of these $Z$ opinions. Hence, if there exists at least one decided agent of another opinion, then these $Z$ opinions will vanish forever \footnote{Note, if there is not a single decided agent after the decision part of phase $t$, then the configuration of opinions does not change throughout phase $t$, i.e., $\X(t+1) = \X(t)$}.
In other words if $\norm{\Y}_1 = \sum_{i=1}^{k(t)} Y_i(t) > 0$ \emph{and} $Z > k(t) \cdot \frac{p}{3} = \BigOmega{k(t)}$, then we have $k(t+1) \leq k(t) - k(t) \cdot \frac{p}{3}$.
Observe that 
\[
	\Pr \Big[ \norm{\Y}_1  > 1 \Big] = 1 - \prod_{i=1}^{k(t)} \left(1 - \frac{x_i}{n} \right)^{x_i} > 1 - \left(1 - \frac{1}{n}\right)^{n} = 1 - \frac{1}{e}.
\]
When combining this result with the bound on $Z$ via union bounds, we have that
\[
 \Pr \left[Z > \frac{k(t)}{3} \cdot p \land \norm{\Y}_1 > 0 \right] \geq 1 - \frac{1}{e} - n^{-\BigOmega{1}} > \frac{1}{2}
\]
Summarizing, we showed that, with probability at least 1/2, we have that $k(t+1) < k(t) - \frac{k(t) \cdot p}{3} =  k(t) \cdot (1 - p/3)$, i.e., the number of remaining opinions shrinks by a constant factor until \Decision part $t$.

In the following we will call a phase $\hat{t}$ successful, if either $k(\hat{t}+1) < k(\hat{t}) (1 - p/3)$ or $k(\hat{t}) < \sqrt{n} / c_k$. We just showed that with probability at least $1/2$ a phase $\hat{t}$ is successful. Even if a phase $\hat{t}$ is not successful, then $k(\hat{t} + 1) \leq k(\hat{t})$ as the number of distinct opinions is non-increasing over time. By a Chernoff bound application, we have that in a sequence of $3 \log_{1 / (1 - p/3)} n = \BigTheta{\log n}$ phases at least $\log_{1 / (1 -p/3)} n$ will be successful. 
From phase $0$ to $3 \log_{1 / (1 - p/3)} n$ we will therefore have $\log_{1 / (1 - p/3)} n$ successful phases \whp and
\begin{align*}
    k\cdot(3 \log_{1 / (1 - p/3)} n) &\leq \max\{ \sqrt{n} / c_k ~,~ k(0) \cdot (1 - p/3)^{\log_{1 / (1 - p/3)} n} \} \\
   &\leq \max\{ \sqrt{n} / c_k ~,~ n \cdot (1 - p/3)^{\log_{1 / (1 - p/3)} n} \} = \sqrt{n} / c_k.
\end{align*}

\medskip

It remains to show that insignificant opinions at the start of phase $0$ remains significant until phases $3 \log_{1 / (1 - p/3)} n$. To avoid repetition, we refer to the proof of \cref{pro:case-1} where we show that  $\Sigset{\X(t + t')} \subseteq \Sigset{\X(t)}$ for $t'=\BigO{\log n}$.
This statement follows from \cref{lem:initial-bias}, which states that insignificant opinions remains insignificant from one \Decision part to the next.
Finally, we note that $|\Sigset{\X(t')}| > 0$ must hold as the opinion with the maximum support is always significant.

\end{proof}

\subsection{Proof of \cref{thm:main-result-population-model}}
\label{sec:details-proof-of-main-result}

We first show that the correctness of the synchronization follows from \cite{DBLP:conf/soda/AlistarhAG18,DBLP:journals/rsa/PeresTW15}.
There it is shown that for a polynomial number of phases and for any pair of agents $u$ and $v$ the distance between $\Clock u$ and $\Clock v$ w.r.t.\ the circular order modulo $6\tau\log{n}$ is less than $\tau \log{n}$, \whp.
The choice of $\tau$ also ensures that every undecided agent is able to adopt an opinion in the boosting part of a phase, \whp.

\begin{proof}[Proof of Synchronization Properties]
In every interaction every agent is either in the decision part or the boosting part of a fixed phase.
We call an agent $u$ \emph{active in a decision part} as long as $\DecisionFlag u = \tfalse$.
An agent $u$ is \emph{active in a boosting part} as long as $\Clock u \leq 5\tau\log{n}$.
Furthermore, we define $P_u(\theta)$ as the number of the phase to which agent $u$ belongs in interaction $\theta$.
Intuitively, we aim to show that the leaderless phase clock separates the phases of agents such that no agent is active in a decision part while another agent is active in the boosting part at the same time.
Recall that the leaderless phase clock works as follows.
The clock of agent $u$ uses the variable $\Clock u$ which can take values in $\set{0,\ldots, 6\tau \log n-1}$ for a suitably chosen constant $\tau$.
The circular order modulo $m$, $a \leq_{(m)} b$, is defined as $a \leq_{(m)} b \equiv (a \leq b ~\text{\textsc{xor}}~ \abs{a - b} > m/2)$,
and the distance w.r.t.\ the circular order modulo $m$ is defined as $ \min\set{\abs{a - b},~ m - \abs{a - b}}$.
In every interaction $(u,v)$, the smaller of the two values $\Clock u$ and $\Clock v$ is increased by one modulo $6\tau\log{n}$.
Here, smaller refers to the circular order modulo $6\tau\log{n}$.
For the correctness of our protocol it is sufficient that the following synchronization properties hold for a polynomial number of interactions.
\begin{enumerate}
\item For any pair of agents $u$ and $v$ we have $P_u(\theta)=P_v(\theta)\pm 1$.
\item Assume agent $u$ with $P_u(\theta) = t$ is interacting at time $\theta$, and $u$ is active in the decision part of phase $t$.
Then there exists no agent $v$ that is already active in the boosting part of phase $t$ or still active in the boosting part of phase $t-1$.
\item Assume agent $u$ with $P_u(\theta) = t$ is interacting at time $\theta$, and $u$ is active in the boosting part of phase $t$.
Then there exists no agent $v$ that is already active in the decision part of phase $t+1$ or still active in the decision part of phase $t$.
\item Let $Z(t)$ be defined as the interval of interactions during which all agents $u$ are together and active in the boosting part of the same phase $t$, i.e., \[ Z(t) = \bigcap_{u}\set{t | P_u(\theta) = t \text{ and } 2\tau\log{n} \leq \Clock u(t) \leq 5\tau\log n}.\]
Then for each $1 \leq t \leq \poly{\log{n}}$ we have $\abs{Z(t)} > n \tau \log{n}$.
\end{enumerate}

The first condition directly follows from \cite{DBLP:conf/soda/AlistarhAG18,DBLP:journals/rsa/PeresTW15}.
There it is shown that for a polynomial number of phases and for any pair of agents $u$ and $v$ the distance between $\Clock u$ and $\Clock v$ w.r.t.\ the circular order modulo $6\tau\log{n}$ is smaller than $\tau \log{n}$, \whp.

To show the second and third conditions, it suffices to show that a) no agent becomes active in the boosting part of a phase $t$ while another agent is still active in the decision part of phase $t$, and b) no agent becomes active in the decision part of a phase $t+1$ while another agent is still active in the boosting part of phase $t$.

To show a), fix a phase $t \leq \poly\log{n}$ and let $\theta$ be the first interaction in which any agent $u$ has $P_u(\theta) = t$ and $\Clock u = 2\tau\log{n}-1$.
As before, observe that $\Clock u$ and $\Clock v$ differ by less than $\tau\log{n}$ for any pair of agents $u$ and $v$ at interaction $\theta$ \whp.
Hence at interaction $\theta$, no other agent $v$ has a clock value $\Clock v \leq \tau\log{n}$ \whp.
It follows that every other agent $v$ has already set $\DecisionFlag v = \ttrue$ at interaction $\theta$ and thus is not active in the decision part \whp.
This guarantees (\whp) a clean separation between decision parts and boosting parts.
(Technically, it is also necessary that agent $v$ has been at least once the left agent in an interaction pair $(v,w)$. This condition follows from a simple Chernoff bound, since every agent was part of at least $\tau\log{n}$ many interactions and in each interaction an involved agent is the left agent with probability $1/2$.)

To show b), fix again a phase $t \leq \poly\log{n}$ and let $\theta$ be the first interaction in which any agent $u$ has $P_u(\theta) = t$ and $\Clock u = 6\tau\log{n}-1$.
As before, observe that $\Clock u$ and $\Clock v$ differ by less than $\tau\log{n}$ for any pair of agents $u$ and $v$ at interaction $\theta$ \whp.
Hence at interaction $\theta$, no other agent $v$ has a clock value $\Clock v \leq 5\tau\log{n}$ \whp, and thus no other agent is active in interaction $\theta$ \whp.
This now guarantees (\whp) the clean separation between boosting parts and decision parts.

It remains to show the fourth condition.
Fix a phase $t \leq \poly\log{n}$ and let $z_{\min} = \min Z(t)$ and $z_{\max} = \max Z(t)$.
At interaction $z_{\min}$, there exists an agent $u$ with $\Clock u = 2\tau\log{n}$. 
Hence no agent can have a clock value larger than or equal to $3\tau\log{n}$ \whp (since $\Clock u$ and $\Clock v$ differ by less than $\tau \log{n}$ \whp, see above).
Analogously, at interaction $z_{\max}$, there exists an agent $u$ with $\Clock u = 5\tau\log{n}-1$. As before, no agent can have a clock value smaller than $4\tau\log{n}$ \whp.
It takes at least $n\tau\log{n}$ interactions for all agents to advance their clocks from $3\tau\log{n}$ to $4\tau\log{n}$ and hence $\abs{Z(t)} \geq n\tau\log{n}$ \whp.

\medskip

The fourth condition guarantees that all undecided agents become decided again at the end of a boosting phase.
This holds since the interval of interactions during which all agents are in the boosting part of a phase is long enough for a so-called broadcast to succeed, and the way how agents become decided can be seen as a simple broadcast process.
It is folklore that for a sufficiently large constant $\tau$ a broadcast succeeds within $n \tau \log{n}$ interactions \whp (see, e.g., the notion of \emph{one-way epidemics} in \cite{DBLP:journals/dc/AngluinAE08a}).
Note that all agents becoming decided in the boosting part is also crucial to modeling the boosting part by a \emph{Pólya-Eggenberger distribution}.
\end{proof}

Above, we showed that the clocks properly separate the parts of each phase and guarantee
 \whp, long enough phases (of length $\BigO{\log n}$ time) such that \cref{obs:bin, obs:polya} hold \whp. In the rest of the section we assume that this indeed holds.
% \begin{proof}[Proof of Part I of \cref{thm:main-result-population-model}]
% \cref{pro:case-3} shows that the number of opinions is reduced to $\sqrt{n}/\log n$ within $\BigO{\log n}$ phases, \whp. The proposition also shows that at least one of the opinions which is significant in the initial configuration is among the 
% $\sqrt{n}/\log n$ remaining opinions, \whp. In turn,  \cref{pro:case-2} shows that the number of opinions is further reduced to $\sqrt{n}/c_k$ (for a constant $c_k$) within $\BigO{\log n}$ additional phases, \whp. Again, at least one of the opinions which is significant in the initial configuration (and in the configuration with $\sqrt{n}/\log n$ opinions) is still significant at that point, \whp. Finally, \cref{pro:case-1} shows that after an additional $\BigO{\log n}$ phases only one of the significant  opinions remains, \whp. Then Part 1 of the theorem follows with the observation that one phase consists of $\BigO{\log n}$ parallel steps.
% \end{proof}
The proof for the second part of \cref{thm:main-result-population-model} closely resembles the proofs conducted in \cite{DBLP:conf/podc/GhaffariP16a,DBLP:conf/icalp/BerenbrinkFGK16}.
Note that it is not sufficient to adapt the proof of Part 1 of the theorem since the number of phases given in \cref{thm:main-result-population-model} for the biased case is $\BigO{\log n}$, which is too high. The proof of the second part relies on the two technical \cref{lem:23, lem:24}, which we present at the end of this section.

\begin{proof}[Proof of Part II of \cref{thm:main-result-population-model}]
In case $k \leq \sqrt{n} / \log n$ it follows directly from \cref{lem:23} that after $\BigO{\log \log_{\alpha} n}$ phases all agents agree on the initial majority opinion \whp. In case of $k > \sqrt{n} / \log n$ we first relay on \cref{lem:23}, where we establish that after $\BigO{\log \log_{\alpha} n + \log \log n}$ time the initial majority grows to size $(5/8)\cdot n$ \whp. Then, we apply \cref{lem:24}, which shows that after further $\BigO{\log \log n}$ phases again all agents agree on this majority opinion \whp.
The runtime of Part II of the theorem follows as each phases lasts for $\BigO{\log n}$ time.
This concludes the proof.
\end{proof} 

 \begin{lemma}
 \label{lem:23}
      Fix $\X(0) = \x(0)$ and let $i^*$ be an opinion that provides the initial maximum.
      Assume $\x(0)$ has an additive bias of at least $\sDelta\sqrt{n\log{n}}$ and a multiplicative bias of $\alpha$.
        \begin{enumerate}
            \item If $k \leq \sqrt{n} / \log n$, then all agents agree on opinion $i^*$ in $\BigO{\log \log_{\alpha} n}$ phases, \whp.
            \item If $k > \sqrt{n} / \log n$, then for $t^*=\BigO{\log \log_\alpha n + \log \log n}$ we have $X_{i^*} (t^*) > (5/8) \cdot n$ \whp.
        \end{enumerate}
 \end{lemma}
 \begin{proof}
    In this proof, we will use the following notions with respect to some configuration $\X(t)$: Remember that $X_{max}(t)$ is the support of a largest opinion and let $X_{sec}(t)$ denote the support of a second-largest opinion in $\X(t)$. Furthermore, we define
    \[
        \alpha(t) := X_{max}(t)/ X_{sec}(t) \quad \text{ and } \quad \gamma(t) := \min \Big\{\alpha(t) ~,~ X_{max}(t) / (c_w \sqrt{n \log n}\Big\}.
    \]
    To facilitate the proof we first show three intermediate results: \cref{obs:part-A, obs:part-B, obs:part-C}.
    \begin{observation}
    \label{obs:part-A}
        Fix $\X(t) = \x(t)$. If $x_{max}(t) -  x_{sec}(t) > \sDelta \sqrt{n \log n}$ then
        \begin{equation}\label{eq:exponent}
             \alpha(t+1) > \gamma(t)^{3/2}
    \end{equation}
    \end{observation}
    \begin{proof} We start by considering a fixed opinion $j$ with $x_{max} > x_j \geq c_w \sqrt{n \log n}$.
    An application of \cref{lemma:ratio} immediately yields with probability $1-n^{-2}$ that 
    \begin{equation}
    \label{eq:lemma21-square}
     \frac{X_{max}(t+1)}{X_j(t+1)} \overset{\cref{lemma:ratio}}{\geq }\left( \frac{x_{max}}{x_j} \right)^{1.5} \geq \alpha(t)^{1.5} \geq \gamma(t)^{3/2}.
    \end{equation}
    On the other hand, consider now some fixed opinion $j$ with $x_j < c_w \sqrt{n \log n}$.
    We lower and upper bound the support of the largest opinion and $j$, respectively. We define $c_1 :=\sqrt{(\ln 7+2\log n)/(\concConstEps\log n)}$ and apply \cref{lemma:one-phase-conc} with $\delta=c_1\sqrt{\log n}$ and \cref{lemma:weak-phase-conc} with $c=c_w$ we get that  with probability at least $1-5n^{-2}$
    \begin{align*}
        \frac{X_{max}(t+1)}{X_{j}(t+1)}
        &\geq \frac{ \frac{x_{max}^2}{\Decided} \cdot \left( 1- \frac{c_1 \cdot \sqrt{n\log n}}{x_{max}} \right) }{ \frac{n}{\Decided} \cdot (12c_w^2 + 74\polyaConstLarge) \cdot \log n } = \left( \frac{x_{max}}{\sqrt{(12c_w^2 + 74\polyaConstLarge) \cdot n\log{n}}} \right)^2 \cdot  \left( 1- \frac{c_1 \cdot \sqrt{n\log n}}{x_{max}} \right) \\
        &=   \left( \frac{x_{max}}{\sqrt{ n\log{n}}} \right)^{3/2} \cdot \left( \frac{x_{max}}{\sqrt{ n\log{n}}} \right)^{1/2} \cdot \frac{1}{(12c_w^2 + 74\polyaConstLarge)} \cdot   \left( 1- \frac{c_1 \cdot \sqrt{n\log n}}{x_{max}} \right) \\
        &\overset{(a)}{\geq} \left( \frac{x_{max}}{\sqrt{ n\log{n}}} \right)^{3/2} \cdot  \sqrt{\sDelta}  \cdot \frac{1}{(12c_w^2 + 74\polyaConstLarge)} \cdot   \left( 1- \frac{c_1}{\sDelta} \right) \\
        &\overset{(b)}{\geq} \left( \frac{x_{max}}{\sqrt{ n\log{n}}} \right)^{3/2} \cdot \frac{\sqrt{\sDelta}}{24c_w^2 + 148\polyaConstLarge}  
        \overset{(c)}{\geq} \left( \frac{x_{max}}{\sqrt{ n\log{n}}} \right)^{3/2} \cdot \left(\frac{1}{c_w}\right)^{3/2}
        \geq \gamma^{3/2}.
    \end{align*}
    in which for (a) we use $ x_{max}/(c_w \sqrt{n\log{n}}) < x_{max}/x_j$ and  $x_{max} \geq \sDelta \sqrt{n\log{n}}$, for (b) we utilize $1-c_1/\sDelta>1/2$ and for (c) we consider  $\sDelta=(160\cdot c_w)^2+(148\cdot \polyaConstLarge)^2\ge ((24c_w^2+148\polyaConstLarge)/c_w^{3/2})^2$ (see \cref{def:Constant}).
    Summarizing, we show for any fixed opinion $j$ with $x_j < x_{max}$ we have with probability at least $1-O(n^{-2})$. 
    \[
        \frac{X_{max}(t+1)}{X_j(t+1)} \geq \gamma(t)^{3/2}.
    \]
    A union bound over all opinions $j$ yields that, \whp, $\alpha(t+1) > \gamma(t)^{3/2}$ as desired.
    \end{proof}
    
    \begin{observation}
    \label{obs:part-B}
        Fix $\X(t) = \x(t)$.  If (i) $x_{max}(t) - x_{sec}(t) \geq \sDelta \cdot \sqrt{n \log n}$ , (ii) $\gamma(t) \geq \sDelta / c_w = 1 + \Omega(1)$, and (iii)  $x_{max}(t)^2 / n <  x_{sec}(t)/3$, then \whp
        \[
            \gamma(t+1) \geq \gamma(t)^{5/4}.
        \]
    \end{observation}
    \begin{proof} We again fix a configuration $\X(t) = \x$ at the start of some phase $t$ and assume that $\x$ fulfills the requirements of the observation. We will show that \whp,
    \[
        \frac{X_{max}(t+1)}{ c_w \cdot\sqrt{n \log n}} \geq \gamma(t)^{5/4}.
    \]
    Note that, together with the result of \cref{obs:part-A} this also implies $\gamma(t+1) \geq \gamma(t)^{5/4}$.
    
    To show this, we distinguish between two cases. First consider the case of  $x_{sec} < c_w \cdot \sqrt{n \log n}$.
    We apply \cref{lemma:one-phase-conc} together with $\delta = c_1 \cdot\sqrt{n \log n}$ for $c_1 =\sqrt{(\ln 7+2\log n)/(\concConstEps\log n)}$. This yields for $\Decided = \sum_{j=1}^{k} x_j^2 / n$  and with probability $1-5n^{-2}$ that 
    \[
        X_{max}(t+1) \overset{\cref{lemma:one-phase-conc}}{>} \frac{x_{max}^2}{\Decided} \cdot \left(1 - \frac{c_1 \sqrt{n \log n}}{x_{max}} \right) > \frac{x_{max}^2}{\Decided} \left( 1 - \frac{c_1}{\sDelta}\right) > \frac{x_{max}^2}{\Decided} \left(1 - \frac{1}{100} \right).
    \]
    Here we used in the second step that $x_{max} > \sDelta \sqrt{n \log n}$ as per assumption (i) and the last step follows from the definition of the constants $c_1$ and $\sDelta$ (see \cref{def:Constant}).
    To further simplify above result, we note the following. First,  $\Decided = \sum_{j=1}^{k} x_j^2 / n \leq x_{max}^2 / n +  x_{sec}$ is true for any configuration $\x$. Second, because we additionally assume $x_{sec} <  c_w \sqrt{n \log n}$ and $x_{max}^2 / n < x_{sec}/3$ this further implies $\Decided \leq (1 + 1/3) c_w \cdot\sqrt{n \log n}$. When using this we get
    \[
        X_{max}(t+1) \geq \frac{x_{max}^2}{(1 + 1/3) \cdot c_w \cdot \sqrt{n \log n}} \cdot \left(1 -  \frac{1}{100} \right) > \frac{x_{max}^2}{c_w \cdot\sqrt{n \log n}} \cdot \frac{5}{8} = x_{max} \cdot  \gamma(t) \cdot \frac{5}{8}.
    \]
    When dividing by $c_w \cdot\sqrt{n \log n}$ on both ends of this inequality chain, we get
    \[
        \frac{X_{max} (t+1)}{c_w \cdot\sqrt{n \log n}} \geq \frac{x_{max}}{c_w \cdot \sqrt{n \log n}} \cdot \gamma(t) \cdot \frac{5}{8} \geq \gamma(t)^{2} \cdot \frac{5}{8} > \gamma(t)^{5/4}.
    \]
    In the last step we rely on assumption (ii), which implies that $\gamma(t) > \sDelta / c_w \gg 100$.
    
    The remaining case we need to consider is $x_{sec} \geq c_w \sqrt{n \log n}$. Let $j$ denote an opinion with $x_j = x_{sec}$.
    First, we argue  that $X_{j}(t+1) > x_{j} / 2$.
    To that end we again apply \cref{lemma:one-phase-conc} with $\delta = c_1 \cdot\sqrt{n \log n}$  which yields with probability $1-5n^{-2}$ that
    \[
        X_{j}(t+1) \overset{\cref{lemma:one-phase-conc}}{>} \frac{x_{j}^2}{\Decided} \cdot \left(1 - \frac{c_1 \cdot\sqrt{n \log n}}{x_{j}} \right) > \frac{x_{j}^2}{\Decided} \cdot \frac{2}{3}.
    \] 
    In the second step, we used $x_{j} > c_w \cdot\sqrt{n \log n}$ and  $c_w > 3c_1$ (see \cref{def:Constant}).
    Just as before, we argue that $\Decided \leq x_{max}^2 / n +  x_{j} < (1 + 1/3)\cdot x_{j}$ is implied by assumption (i).  This implies
    \[
        X_{j}(t+1) > \frac{x_{j}^2}{\Decided} \cdot \frac{2}{3} \geq  x_{j} \cdot \frac{1}{(1 + 1/3)} \cdot \frac{2}{3}  \geq \frac{x_{j}}{2}
    \]
    When first using that $x_{j} \ge c_w \cdot\sqrt{n \log n}$ followed by above inequality in the next step, we get 
    \begin{equation}
    \label{eq:lemma21-22}
        \frac{X_{max}(t+1)}{c_w \cdot\sqrt{n \log n}} \geq \frac{X_{max}(t+1)}{x_{j}} > \frac{1}{2} \cdot \frac{X_{max}(t+1)}{X_{j}(t+1)}.
    \end{equation}
     Recall, throughout the proof of \cref{obs:part-A}, we established in inequality  (\ref{eq:lemma21-square}) that $X_{max}(t+1)/X_j(t+1)>(x_{max}/x_j)^{1.5}$ in case opinion $j$ has $x_j > c_w \cdot\sqrt{n \log n}$. As this is indeed the case, we have \2whp
    \[
        \frac{X_{max}(t+1)}{X_{j}(t+1)} \overset{(\ref{eq:lemma21-square})}{>} \left( \frac{x_{max}}{x_{j}} \right)^{3/2} \geq \gamma(t)^{3/2}
    \]
    Now, when combining this with (\ref{eq:lemma21-22}) we get \2whp that
    \[
       \frac{X_{max}(t+1)}{c_w \sqrt{n \log n}}  \geq \frac{1}{2} \cdot \left(\frac{X_{max}(t+1)}{X_{j}(t+1)}\right)^{3/2} = \frac{1}{2} \gamma(t)^{3/2} > \gamma(t)^{5/4}.
    \]
    Where we use in the last step that per assumption (ii) $\gamma(t) > \sDelta / c_w \gg 100$ is large enough (see \cref{def:Constant}).
    \end{proof}
    
    \begin{observation}
    \label{obs:part-C}
    Assume $\X(0)$ has an additive bias of at least $\sDelta \sqrt{n \log n}$. Then, \whp, it holds for all $t < \log^2 n$ that 
    \begin{enumerate}
        \item $X_{max}(t) - X_{sec}(t) > \sDelta \sqrt{n \log n}$ and 
        \item $\Sigset{\X(t)} = \Sigset{\X(0)}$
    \end{enumerate}
     The first statement implies that the bias does not fall below $\sDelta \sqrt{n \log n}$ in the first $\log^2 n$ rounds.
    \end{observation}
    \begin{proof} 
    To show this observation, we first consider some fixed phase $t$ and assume that $X_{max}(t) - X_{sec}(t) > \sDelta \sqrt{n \log n}$. Note that this is equivalent to  assuming that $|\Sigset{\X(t)}| = 1$, or in other words, assuming that there exists exactly one significant opinion. From \cref{lem:initial-bias} it now follows that $\Sigset{\X(t+1)} \subseteq \Sigset{\X(t)}$ \whp. On the other hand, we know that  $\Sigset{\X(t+1)} \neq \emptyset$ as the largest opinion is always significant. When combining these two observations, we therefore get \whp that $\Sigset{\X(t+1)} = \Sigset{\X(t)}$. 
    An inductive application of this approach yields that, \whp,
    \[
        \forall 0 < t \leq \log^2 n: \Sigset{\X(t)} = \Sigset{\X(0)}.
    \]
    This immediately yields the second results of \cref{obs:part-C}. The second result follows as $|\Sigset{\X(t)}| = |\Sigset{\X(0)}|= 1$, which implies the existence of an additive bias of at least $\sDelta \sqrt{n \log n}$ in phase $t$.
    \end{proof}
    
    We are now ready to start with the proof of \cref{lem:23}.
      We first repeatedly apply the result of \cref{obs:part-A} until we hit a phase $t_1$ with $\gamma(t_1) < \alpha(t_1)$. From the definition of $\gamma(t_1)$ this is equivalent to  $X_{sec}(t_1) < c_w \sqrt{n \log n}$.
      Recall, \cref{obs:part-A} states that if in some phase $t$ we have $X_{max}(t) -  X_{sec}(t) > \varepsilon \sqrt{n \log n}$, then $\alpha(t+1) > \gamma(t)^{3/2}$ \whp. In \cref{obs:part-C}, we established that this requirement on the additive bias is fulfilled in the first $\log^2 n$ phases \whp. Therefore, the result of \cref{obs:part-A} is applicable to the first $t' = \log_{3/2} \log_{\alpha(0)} n = \BigO{\log n}$ phases \whp. That is, \whp, we have 
      \[
	    \forall 0 \leq t < t':~ \alpha(t+1) \geq \gamma(t)^{3/2}.
	  \]
	  Now, assume that $t_1 > t'$, i.e., we have $\gamma(t_1) < \alpha(t_1)$ for the first time in some phase after $t'$. Therefore, in every phase $t$ with $t \leq t'$ we have $\gamma(t) = \alpha(t)$. This implies $\alpha(t') \geq \alpha(0)^{(3/2)^{t'}} > n$. This further implies $X_{sec}(t') < c_w \sqrt{n \log n}$ and from the definition of $\gamma(t')$, it follows that $\gamma(t') < \alpha(t')$. This is a contradiction to our assumption of $t_1 >t'$. Hence $t_1 < \log_{5/4} \log_{\alpha(0)} n = \BigO{\log n}$  \whp
	  
	  \medskip
	  
	  Starting with phase $t_1$, we split our analysis into two cases. First assume that $k < \sqrt{n} / \log n$. Per definition of phase $t_1$ we have $X_{sec}(t_1) < c_w \sqrt{n \log n}$. This implies that all opinions besides the first are super-weak in $\X(t_1)$. From \cref{lem:all-super-weak} it follows that in phase $t_1 + 2$ only a single opinion prevails \whp. As $t_1 = \BigO{\log n}$ \whp, the second statement of \cref{obs:part-C} implies that this opinion must be the initially significant opinion. The first statement of \cref{lem:23} follows.
	  
	  \medskip
	  
	   The case of $k \geq \sqrt{n} / \log n$ is more involved. In this case, we need to follow a different approach from phase $t_1$ onward.
	   In the following we define three events. We say that a phase $t$ fulfills
	   \begin{enumerate}
	       \item event $\mathcal{E}_1(t)$ iff $X_{max}(t) - X_{sec}(t) > \sDelta \sqrt{n \log n}$
	       \item event $\mathcal{E}_2(t)$ iff $\gamma(t) \geq \sDelta / c_w$
	       \item event $\mathcal{E}_3(t)$ iff $X_{max}(t)^2 / n < X_{sec}(t) / 3$.
	   \end{enumerate}
	   Observe, if $t$ fulfills all three events, then we may apply \cref{obs:part-B} and get $\gamma(t+1)  \geq \gamma(t)^{5/4}$.
	   We now show that we will soon arrive at a phase $t_1 + t_2$ such that $\mathcal{E}_3(t)$ does not hold.
	   
	   Recall that we have $t_1 = \BigO{\log n}$ \whp. It follows from \cref{obs:part-C} that $\mathcal{E}_1(t_1)$ is fulfilled. Additionally, it follows from the definition of $\gamma(t_1)$ that $\gamma(t_1) = \frac{X_{max}(t_1)}{c_w \sqrt{n \log n}} \geq \sDelta / c_w > 1$. Therefore $\mathcal{E}_2(t_1)$ also holds. Now, in case $\mathcal{E}_3(t_1)$ does not hold, we are finished. Otherwise, we may apply \cref{obs:part-B} and get \whp that $\gamma(t_1+1) \geq \gamma(t_1)^{5/4}$. 
	   
	   Note that, in such case, $\mathcal{E}_2(t_1+1)$ will also be satisfied because $\gamma(t_1 + 1) \geq \gamma(t_1)^{5/4} \geq \sDelta / c_w$.
	   Furthermore, \cref{obs:part-C} still guarantees that, \whp, $\mathcal{E}_1(t_1 + 1)$ still holds. In case $\mathcal{E}_3(t_1 + 1)$ is fulfilled we again apply \cref{obs:part-B} and get $\gamma(t_1 +2 ) \geq \gamma(t_1)^{(5/4)^2}$.
        We repeat this approach for $t' = \log_{5/4} \log_{\sDelta/c_w} n = \BigO{\log \log n}$ phases. Even if there is no $t$ with $0 \leq t < t'$ such that $\mathcal{E}_3(t_1 + t)$ is violated, we have \whp that
	   	\[
	        \gamma(t_1 + t') \geq \gamma(t_1)^{(5/4)^{t'}} > \left(\frac{\sDelta}{c_w} \right)^{(5/4)^{t'}} = n.
	   \]
	   At this point $\mathcal{E}_3(t_1  + t')$ must be violated. 
	   It follows that $t_2 \leq t' = \BigO{\log \log n}$ holds \whp.
	   
	    Now fix the configuration $\X(t_1 + t_2) = \x$ and assume that $\mathcal{E}_3(t_1 + t_2)$ does not hold. This implies that $x_{max}^2 / n \geq x_{sec} / 3$. As $\Decided < x_{max}^2 / n +  x_{sec}$ this further implies $\Decided \leq (x_{max}^2 / n) \cdot  (1 +  1/3)$. We now apply \cref{lemma:one-phase-conc} with $\delta = c_1 \cdot \sqrt{n \log n}$ which yields that
	    \[
	        X_{max}(t_1 + t_2 + 1) \geq \frac{x_{max}^2}{\Decided} \left( 1 - \frac{c_1 \sqrt{n \log n}}{x_{max}}\right) \geq \frac{x_{max}^2}{\frac{x_{max}^2}{n} (4/3)} \left(1 - \frac{1}{100} \right)
            > n \cdot \frac{5}{8}
	    \]
	    Summarizing, in phase $t_1 + t_2 + 1 = \BigO{\log \log_{\alpha(0)} n + \log \log n)}$ we have $X_{max}(t_1 + t_2 +1) > n \cdot (5/8)$ \whp. Additionally, just as in the case of $k \leq \sqrt{n} \log n$, we argue that \cref{obs:part-C} implies that this maximum must be the initially largest opinion.
 \end{proof}

 \begin{lemma}
 \label{lem:24}
     Assume $k > \sqrt{n} / \log n$. Fix $\X(t) = \x(t)$. Assume  $x_1(t)\ge (5/8) \cdot n$. Then all agents agree on this opinion within $\BigO{\log\log{n}}$ phases, \whp.
 \end{lemma}
 \begin{proof} 
 
     We start by showing that $4$ phases following $t$, at most $\sqrt{n} / \log n$ opinions will have non-zero support. Let $i$ be an opinion which provides $ \x_{max}$. From the assumption we have $ x_i > (5/8)\cdot n$. Let $L(t)$ denote the set of opinions with support  at most $\sqrt{n} \log n$ and fix an opinion $j \in L(t)$. Since $x_i > (5/8)\cdot n$, it follows that $\Decided := \Ex{\norm{\Y(t)}_1} = \sum_{j=1}^{k} x_j^2 / n \geq \frac{x_i^2}{n} > (25/64)\cdot n$.
     We now distinguish two cases, depending on the size of $x_j$, and show that $X_j (t+1) = \BigO{\log^2 n}$.
     First, assume that $x_j < \sqrt{2/\concConstEps} \cdot \sqrt{n \log n}$ (remember, the constant \concConstEps  is stated in \cref{def:Constant}). In this case, we apply \cref{lemma:weak-phase-conc} and use that $\psi = \BigTheta{n}$ which immediately yields that 
     \begin{equation}
     \label{eq:thm-lemma-2-eq0}
		\Pr \left[ X_j(t+1) = \omega(\log n) \right] < 4n^{-2}.		 
	\end{equation}
     In the case of $x_j >  \sqrt{2/\concConstEps} \cdot \sqrt{n \log n}$, we apply \cref{lemma:one-phase-conc} together with $\delta =x_j / 2\sqrt{n} $ and get 
	\begin{equation}
	\label{eq:thm-lemma-2-eq1}
		\Pr \left[ X_j(t+1) > \frac{x_j^2}{\Decided} \cdot \frac{3}{2} \right] = \Pr \left[ X_j(t+1) > \frac{x_j^2}{\Decided} \left(1 + \frac{\delta\cdot \sqrt{n}}{x_j}\right) \right] < 7 n^{-2}. 
	\end{equation}
	When using that $\Decided > (25/64)\cdot n$ and $x_j < \sqrt{n} \log n$ (remember, $j \in L(t)$) we have $(x_j^2 / \Decided) \cdot (3/2) \leq 4\log^2 n$. The inequality in (\ref{eq:thm-lemma-2-eq1}) then implies that 
	\[
		\Pr \left[ X_j(t+1) \geq 4 \log^2 n \right] < 7 n^{-2}.
	\]
	Hence, in any case, we have that $X_j(t+1) < 4 \log^2 n$ with probability at least $1 - \BigO{n^{-2}}$. By a union bound application, we get that this holds for every $j \in L(t)$ \whp Note that this also implies that $L(t) \subseteq L(t+1)$ \whp
    Additionally, it follows from \cref{lemma:one-phase-conc} that, \whp,  $X_i(t+1) > x_i(1-\LittleO{1})$.
    Now observe that 
     \[
        \Pr[Y_j(t+2) = 0 ~|~ X_j(t+1) < 4 \log^2 n] > \left(1 - \frac{4 \log^2 n}{n}\right)^{4 \log^2 n} > 1 - \frac{\polylog n}{n}.
     \] 
    This implies that some fixed opinion $j \in L(t)$ vanishes after \Decision part $t+2$ \whp. As $L(t) \subseteq L(t+1)$ \whp this approach can be repeated, and we deduce that $Y_j(t+3) = 0$ with probability at least $1 - \polylog n / n^2$.  By a union bound application it follows \whp for every $j \in L(t)$ that $Y_j(t+3) = 0$, and by a counting argument we have that all but $\sqrt{n} / \log n$ opinions lie in $L(t)$. Therefore, at the start of \Decision part $t+4$, only $\sqrt{n} / \log n$ opinions remain.
    
    Additionally, consider how the largest opinion $i$ evolves from the \Decision part of phase $t$ until $t+4$.
    By \cref{lemma:one-phase-conc} we have for $\delta = \sqrt{\log n}$ that, \whp, 
    \begin{equation}
    \label{eq:theorem-part2-2}
        X_i(t+1) > \frac{x_i^2}{\Decided} - \frac{x_i}{ \Decided} \sqrt{n \log n}
    \end{equation}
    The inequality $\Decided = \sum_{i=1}^{k} x_i^2 / n \geq x_{max}$ is true for every configuration $\x$ and additionally we assumed $x_i = x_{max}$. We use this to lower-bound the right-hand side of  (\ref{eq:theorem-part2-2}) and get that $X_i(t+1) > x_i (1 - \LittleO{1})$ \whp Note, as $x_i \geq (5/8)\cdot n$, this easily implies that $X_i(t+1) > n/2$ and therefore $i$ will remain the opinion with the largest support.
    When repeating this argument three more times from \Decision part $t+1$ until $t+4$, we therefore get that $X_i(t+4) > (5/8)\cdot n \cdot (1 -\LittleO{1}) > (24/64) \cdot n$ \whp.
    
    Summarizing, we get by a union-bound application that in \Decision part $t+4$ (i) less than $\sqrt{n} / \log n$ distinct opinions remain, and (ii) the largest opinion $i$ has at least $(24/64)\cdot n$ support \whp. The small amount of remaining opinions allows us the apply results from the analysis in \cref{sec:analysis-case1}. According to the definition of weak in \cref{sec:analysis-case1}, all opinions besides $i$ must be weak in phase $t+4$. Let $j$ with $j \neq i$ be such an opinion. As $j$ is weak, it follows from \cref{lem:weak-to-super-weak} (which states that weak opinions become super-weak and stay super-weak) that opinion $j$ will be super-weak at some phase $t + 4 +t'$ with  $t'= \BigO{\log \log  n}$ and probability at least $1-\BigO{n^{-1.9}}$. By a union bound application it follows \whp that \emph{every} such opinion $j$ is super-weak in \Decision part $t+4+t'$. Next we apply \cref{lem:all-super-weak} which implies that after further $2$ phases only opinion $i$ remains and the result follows.
 \end{proof}

\section{Analysis for the Gossip Model}
\label{sec:analysis-gossip-model}

Our algorithm in the gossip model is similar to the one in the population model.
As before, our algorithm runs in multiple phases that consist of a \Decision part (one round) and a \Boosting part (multiple rounds). 
Analogously to the population model, we use $\X(t)$ to denote the configuration at the start of phase $t$ (before the \Decision part) and $\Y(t)$ to denote the number of decided agents of each opinion (at the beginning of the \Boosting part).
As before, we fix in our analysis some configuration $\X(t) = \x$ at the beginning of phase $t$.
This allows us to model $\Y(t)$ as a vector of independent random variables with Binomial distributions $\Y_i(t) \sim \Bin(x_i, x_i / n)$, as in the analysis for the population model.

%Our strategy for the analysis is to show that $X_i(t+1)$ obeys similar bounds as in the population model.
As in the population model, the decision part and the boosting part are analyzed separately.
The bounds derived for the population model can also be applied here w.r.t.\ the decision part.
Hence, the analysis of the decision part is simply a repetition of the analysis from \cref{sec:analysis-population-model}. 
However, the boosting part is different since we can no longer use Pólya-Eggenberger distributions to model the outcome of the entire \Boosting part of a phase.
Instead, we consider the rounds of a phase one by one (see \cref{lem:sync-boosting-concentration}).
Let $\Y(t,r)$ denote the configuration of decided agents in round $r$ of the \Boosting part of phase $t$.
Then we have $\Y(t) = \Y(t,1)$ and $\X(t+1) = \Y(t, \T{BC}+1)$.
When we fix $\Y(t,r) = \y$, we can model $\Y_i(t,r+1) \sim \y_i + \Bin(n - \norm{\y}_1, y_i / n)$.
This holds since every undecided agent $u$ (out of $n - \norm{\y}_1$ many undecided agents) contacts one other random agent $v$ and adopts opinion $i$ if $v$ is one of the $\y_i$ many agents that have opinion $i$.
This means we can model $\X_i(t+1)$ by a sum of not independent binomial random variables.

%In \cref{lem:sync-boosting-concentration} we show a similar concentration bound for the round-based process in the gossip model as in \cref{thm:polya_bound_total_balls} and \cref{thm:polya_bound_small_support} for the Pólya-Eggenberger process, which we apply in the analysis of  the population model.
%However, the constants have slightly different values.
%Summarizing,
%(i) the outcome of one fixed \Decision part follows the same distribution in both models, and
%(ii) \cref{lem:sync-boosting-concentration} gives (almost) the same concentration guarantees throughout one fixed \Boosting part as in the population model.
%Hence, the analysis of the population model also applies to the gossip model.

\begin{proof}[Proof of \cref{thm:main-result-gossip-model}] 
\label{proof:gossip-model-theorem}
%The analysis from \cref{sec:analysis-population-model} also applies to \cref{alg:consensus-gossip-model} in the gossip model. 
The proof of this theorem is mainly a repetition of the proof of \cref{thm:main-result-population-model} by 
replacing the corresponding concentration results used in \cref{sec:analysis-population-model} (\cref{thm:polya_bound_total_balls},  \cref{thm:polya_bound_small_support}) with the bounds derived in this section (\cref{lem:sync-boosting-concentration}).
This way \cref{pro:case-1}, \cref{pro:case-2} and \cref{pro:case-3} again guarantee that after $\BigO{\log n}$ phases consensus is reached \whp. 
Similarly, we have by \cref{lem:initial-bias} that an initial bias of $\sDelta \sqrt{n\log n} $ for large enough constant $\sDelta > 0$ suffices for the initially largest opinion to win \whp.
Additionally, \cref{lem:initial-bias} guarantees that only significant opinions can become the consensus opinion \whp.
Furthermore, if a bias of at least $\sDelta \sqrt{n\log n} $ is present, the same arguments as in the proof of \cref{thm:main-result-population-model} yield that consensus is reached after (i) $\BigO{\log \log_\alpha n}$ phases if $k \leq \sqrt{n} /  \log n$, and (ii) $\BigO{\log \log_\alpha n + \log \log n}$ phases otherwise.
As each phase consists of $\T{BC} = \BigO{\log  k + \log \log n}$ rounds, all statements of \cref{thm:main-result-gossip-model} follow.
\end{proof}

In the next lemma we show a similar concentration bound for the boosting part in the gossip model as in \cref{thm:polya_bound_total_balls} and \cref{thm:polya_bound_small_support} for the population model.
Note that $\delta = \BigOmega{1}$ in the assumptions of \cref{lem:sync-boosting-concentration} whereas $\delta$ can be arbitrarily small in \cref{thm:polya_bound_total_balls}.
However, that slightly more general result is not necessary for the analysis.  \cref{lem:sync-boosting-concentration} divides the boosting part of the phase in up to 4 different sub-parts
which are analyzed in \cref{lem:sync-to-linear} to \cref{lem:sync-4-last-steps}, respectively.

\begin{lemma}
\label{lem:sync-boosting-concentration} 
    Fix the configuration $\Y(t) = \y$ at the start of the \Boosting part of phase $t$ and some opinion $i$. Then, for any $\delta$ with $c_1 < \delta < \sqrt{y_i}$ 
        \[
            \Pr \left( X_i(t+1) < \frac{y_i}{\norm{\y}_1} \cdot n - \sqrt{y_i} \frac{n}{\norm{\y}_1} \delta \right) < c_2 \cdot \exp(-c_3 \delta^2) \text{ and }
        \]
        \[
            \Pr \left( X_i(t+1) > \frac{y_i}{\norm{\y}_1} \cdot n + \sqrt{y_i} \frac{n}{\norm{\y}_1} \delta\right) < c_2 \cdot \exp(-c_3 \delta^2).
        \]
        Furthermore, we have 
        \[
            \Pr \left( X_i(t+1) > \frac{n}{\norm{\y}_1} (2 y_i + c_4 \log n) \right) < n^{-2}.
        \]
    Here $c_1,c_2,c_3,c_4 > 0$ are suitable constants.
\end{lemma}

\newcommand{\rone}{r_{(1)}}
\newcommand{\rtwo}{r_{(2)}}
\newcommand{\rthree}{r_{(3)}}
\newcommand{\rfour}{r_{(4)}}

\newcommand{\F}[1]{Y_{#1}}
\newcommand{\f}[1]{y_{#1}}

\begin{proof}%[Proof of \cref{lem:sync-boosting-concentration}]
We first prove the first statement.
To improve the readability of our analysis, we introduce the following notation w.r.t.\ the \Boosting part of a fixed phase $t$ and opinion $i$.
For some round $r\geq 1$ of the \Boosting part of phase $t$, we abbreviate $\F{r} := Y_i(t,r)$. Furthermore, we denote by the random variable $D_r = \norm{Y_i(t,r)}$ the number of decided agents at the beginning of round $r$. Similarly, we define $U_r = n- D_r$ to be the number of undecided agents at the beginning of round $r$.
As we start from a fixed configuration $\Y(t) = \y$ at the start of the \Boosting part of phase $t$, we assume the initial values $\F{1},D_1$ and $U_1$ to be fixed as $\F{1} = \f{1}$, $D_1 = d_1$ and $U_1 = u_1$.

When fixing the process at the start of some round $r$ with  $\F{r} = \f{r}$, $D_r = d_r$ and $U_r = u_r$, we may model $\F{r+1} \sim \Bin(u_r,\f{r} / n)$, $D_{r+1} \sim \Bin(u_r, d_r / n)$ and $U_r \sim \Bin(u_r, u_r / n)$. The proof of the first point of \cref{lem:sync-boosting-concentration} mainly relies on this observation. The idea is to compute $\F{r+1}, D_{r+1}$ and $U_{r+1}$ round-by-round and to accumulate the error terms and the error probabilities that arise from Chernoff bounds. Throughout the analysis of the first statement of \cref{lem:sync-boosting-concentration}, we assume $\delta$ to be a fixed value with $c_1  < \delta < \sqrt{\f{1}}$ for some large enough constant $c_1>0$. In order to further facilitate the calculation, we split the boosting part of phase $t$ into sub-parts with the following boundaries.
\begin{enumerate}
    \item $\rone$, the first round such that $U_{\rone} < n/\lambda$ where $\lambda >0$ is a large constant.
    \item $\rtwo$, the first round such that $U_{\rtwo} < n/\lambda$ and 
    %the conditions of $\rtwo$ or $\rthree$ hold %$(\delta > \sqrt{\F{\rtwo} \cdot U_{\rtwo} / n} \cdot 10$ or $\delta > U_{\rtwo} / \sqrt{n} \cdot 10)$
    \begin{enumerate}
    \item 
    %$\rtwo$ is the first round such that 
     $\delta > \sqrt{\F{\rtwo} \cdot U_{\rtwo}/ n} \cdot 10$, or
    \item 
    $\delta > U_{\rtwo} / \sqrt{n} \cdot 10$. %in which case we denote $\rtwo$ by $\rthree$
    \end{enumerate}
    \item $\rthree$, the first round such that $U_{\rthree} < n/\lambda$ and   $\delta > U_{\rthree} / \sqrt{n} \cdot 10$.
    
    \item $\rfour$, the first round such that $U_{\rfour} = 0$.
\end{enumerate}
Our process moves through the sub-parts $[1,\rone), [\rone,\rtwo), [\rtwo, \rthree)$ and $[\rthree, \rfour]$ in sequence. Note, it may happen that $\rtwo = \rthree$, in which case the third sub-part is skipped. For each sub-part, we provide a concentration result in \cref{lem:sync-to-linear, lem:sync-to-muk,lem:sync-til-uk,lem:sync-4-last-steps}, respectively. Throughout any fixed sub-part, the fraction $\F{r} / D_r$ deviates by at most an $(1 \pm \frac{1}{8} \frac{\delta}{\sqrt{\f{1}}})$ factor
with probability  $1 - c\exp(-c'\delta^2)$ for some fitting constants $c,c'> 0$. This implies that
\begin{align*}
    X_i(t+1) &= \F{\rfour} = n \cdot \frac{\F{\rfour}}{D_{\rfour}} < n \cdot \frac{\f{1}}{d_1} \prod_{i=1}^{4} \left(1 + \frac{1}{8} \frac{\delta}{\sqrt{\f{1}}} \right) \leq n \frac{\f{1}}{d_1} (1 + \frac{\delta}{\sqrt{\f{1}}}) = \frac{\f{1}}{d_1} n + \frac{\sqrt{\f{1}}}{d_1} n \cdot \delta
    \\
    &= \frac{y_i}{\norm{\y}} n + \frac{\sqrt{y_i}}{\norm{\y}} n \delta.
\end{align*}
By using the chain rule, it follows that this deviation is not exceeded in \emph{any} sub-part with probability at least $1 - 4 \cdot c \exp(-c' \delta^2) \geq 1 - c_2 \exp(-c_3 \delta^2)$ for suitable constants $c_2 =4c$ and $c_3 = c'$.
A symmetric statement for the lower bound on $\F{\rfour}$ holds analogously.

We continue with the proof of the second statement of  \cref{lem:sync-boosting-concentration}. 
We start by fixing the configuration $\Y(t) = \y$ at the beginning of the \Boosting part of phase $t$ and some opinion $i$. Let $d= \norm{\y}$ and $c_4 = 4 c_3^{-1}$, where $c_3$ is the constant from \cref{lem:sync-boosting-concentration}.
We distinguish three cases depending on the size of $d$ and $y_i$.

\paragraph{Case 1: $d \leq c_4 \log n$} The result follows as the statement trivially holds, since
\[
    \Pr\left[ X_i(t+1) > \frac{n}{d} (2a + c_4 \log n) \right] \leq \Pr \left[ A > n \right] = 0.
\]
\paragraph{Case 2: $y_i \geq \frac{c_4}{2} \log n$} In this case, we apply the first statement of \cref{lem:sync-boosting-concentration} with $\delta = \sqrt{y_i} - 1$, which gives us
\[
    \Pr \left[X_i(t+1) > \frac{n}{d} 2y_i \right] <  n^{-2}.
\]
\paragraph{Case 3: $d > c_4 \log n$ and $y_i < \frac{c_4}{2} \log n$}
We show this case via a coupling. Define $\y'$ with $y'_i = \frac{c_4}{2} \log n$ and $\norm{\y'} = \norm{\y}$. That is, this artificial configuration has the same amount of decided agents but more agents of opinion $i$.
We couple a process that starts the \Boosting part with configuration $\y'$ with our original process. If both processes are subject to the same random choices, we observe that (i) the number of agents becoming decided in every round is the same in both processes, and (ii) the probability for an agent to adopt opinion $i$ is always higher in the coupled process.
Let $X'(t+1)$ be the support of opinion $i$ at the end of the \Boosting part of phase $t$ in the coupled process. It follows that $X'_i(t+1) \succeq X_i(t+1)$ and therefore  
\[
    \Pr \left[X_i(t+1) > \frac{n}{d} c_4 \log n \right] \leq  \Pr \left[X'_i(t+1) > \frac{n}{d} c_4 \log n \right] = \Pr \left[A' > \frac{n}{d'} 2y_i' \right] < n^{-2}.
\]
\paragraph{Case 4}
In the final case we follow a similar approach as in the second case and use the first statement of \cref{lem:sync-boosting-concentration} to bound $X'_i(t+1)$.
\end{proof}

In what follows, we list all the statements required to bound the phases as required by the proof of the first statement of \cref{lem:sync-boosting-concentration}.

\medskip

\paragraph{Additional Notation and Definitions}
In order to make the analysis more readable, we define the abbreviation $a = b(1 \pm c)$ to denote $a \in [b(1-c), b(1+c)]$. Additionally, for two intervals $[s,t], [u,v]$ we say $[s,t] \asymp [u,v]$ iff $u \leq s$ and $t \leq v$ (or in other words $[s,t] \in [u,v]$).

\begin{lemma}
\label{lem:sync-to-linear}
    Let $\rone$ be the first round such that $U_{\rone} < n/\lambda$ for some constant $\lambda > 0$. Then, we have 
	\[
		\frac{\F{\rone}}{D_{\rone}} = \frac{\f{1}}{d_1}  \left(1 \pm \frac{1}{8} \cdot \frac{\delta}{\sqrt{\f{1}}} \right)
	\]
	with probability $1-7\cdot e^{-c\delta^2}$ where $c > 0$ is a constant depending on $\lambda$.
\end{lemma}
\begin{proof}
	We fix our process at the start of round $r$ and assume that $\F{r} = \f{r} \geq \f{1} \cdot \beta^{r}$ for $\beta = (1+\frac{47}{48\lambda}) > 1$ and $U_r = u_r \geq n/\lambda$.
    We model $\F{r+1}$  as $\F{r+1} = \f{r} +  \Bin(u_r, \frac{\f{r}}{n})$. Next, we define $\Delta(r,\delta) = z \cdot \frac{\delta \cdot \sqrt{r}}{\sqrt{\f{r}}}$ where $z =\frac{1}{48} \cdot \min\{ 1\,, 1 / \sum_{r=1}^{\infty} \frac{\sqrt{r}}{\sqrt{\beta}^{r-1}} \} = \BigTheta{1}$. When using that $\f{r} > \f{1} \cdot \beta^{r-1}$ we observe that $\Delta(r,\delta) \leq 1/48 < 1$. We apply Chernoff bounds with error term $\Delta(r,\delta)$ to derive
	\begin{align*}
		\Pr\left[ \Bin(u_r, \frac{\f{r}}{n}) = \frac{u_r \cdot \f{r}}{n} \cdot  \left(1 \pm \Delta(r,\delta) \right)\right] &\geq 1 -  2\exp \left(-\frac{z^2 \cdot \delta^2 \cdot r \cdot u_r}{3n} \right) \\
		&\geq 1 - 2\exp(-z^2 \cdot \delta^2 \cdot r \cdot \lambda / 3) \\
		&=  1 - 2 \exp(-c_\lambda \cdot \delta^2 r), 
	\end{align*}
	where we used $u_r \geq n/\lambda$ in the second line and set $c_\lambda= z^2 \cdot \lambda / 3 = \BigTheta{1}$ in the second.
	Therefore, we have with probability at least $1-2\exp(-c_\lambda \delta^2 r)$ that
	\[
		\F{r+1} = \f{r} + \frac{u_r \cdot \f{r}}{n}(1 \pm \Delta(r,\delta)).
	\]
	Observe that this together with $\Delta(r,\delta) \leq 1/48$ implies $\F{r+1} \geq \f{r} \cdot (1 + \frac{47u_r}{48n}) \geq \f{r} (1+ \frac{5}{6\lambda}) = \f{r} \beta \geq \f{1} \cdot  \beta^{r}$. Furthermore, above result also implies the following weaker bound of 
	\begin{equation}
	\label{eq:sync-1-ak}
		\F{r+1} = \f{r} \left(1 + \frac{u_r}{n}\right) \cdot \left(1 \pm \Delta(r,\delta) \cdot \frac{u_r}{n}\right) \asymp \f{r} \cdot \left(1 + \frac{u_r}{n}\right) \cdot \left(1 \pm \Delta(r,\delta) \right).
	\end{equation}
	
	Similar, we can also model $D_{r+1}$ by $D_{r+1} = d_{r} + \Bin(u_{r}, \frac{d_r}{n})$ and bound it with the help of Chernoff bounds. As $d_r \geq \f{r}$ a repetition of the above yields that -- with probability at least $1-2\exp(-c_\lambda \delta^2 r)$ -- we have
	\begin{equation}
	\label{eq:sync-2-dk}
		D_{r+1} = d_{r} \left(1 + \frac{u_r}{n}\right) \cdot (1 \pm \Delta(r,\delta)).
	\end{equation}
Finally, we employ union bounds and combine the bounds (\ref{eq:sync-1-ak}) and (\ref{eq:sync-2-dk}) to deduce that
	\[
		\Pr\left[ \frac{\F{r+1}}{D_{r+1}} = \frac{\f{r}}{d_r} (1 \pm 3\Delta(r,\delta)) \text{ and } \F{r+1} \geq \f{1} \cdot \beta^{r} \right] > 1 - 4\exp(-c_\lambda \delta^2(r+1)).
	\]
Note that by $\F{r+1} \geq \f{1} \cdot \beta^{r}$ we established that one of the assumption that we made w.r.t.\ round $r$ also holds in round $r+1$.
Therefore, starting from $r=1$ we can inductively apply above bounds until $r={\rone}$ where $U_{\rone} < n/\lambda$ for the first time and get
\begin{align}
	\Pr \Bigg[\frac{\F{\rone}}{D_{\rone}} \leq \frac{\f{1}}{d_1} \prod_{r=1}^{{\rone}-1} \left(1 + 3\Delta(r,\delta) \right) \text{ and }   \frac{\F{\rone}}{D_{\rone}} \geq \frac{\f{1}}{d_1} \prod_{r=1}^{{\rone}-1} \left(1 - 3\Delta(r,\delta) \right)  \label{eq:sync-1-ak-dk} \\
	\text{ and }	\forall 1 \leq r \leq {\rone}: \F{r} \geq \f{1} \beta^{r-1} \Bigg] 
	> 1 - \sum_{r=1}^{{\rone} -1} 4 \exp(-c_\lambda\delta^2 (r+1)). \nonumber
\end{align}
In order to closer examine the error terms, we start by observing that
\[
	\sum_{r=1}^{{\rone} -1} 3\Delta(r,\delta) = 3z\delta \sum_{r=1}^{{\rone} -1}\sqrt{\frac{r}{\f{r}}} < 3z \frac{\delta}{\sqrt{\f{1}}} \sum_{r=1}^{\infty}\sqrt{\frac{r}{\beta^{r-1}}} \leq \frac{1}{16} \cdot \frac{\delta}{\sqrt{\f{1}}} < \frac{1}{16}
\]
where in the second step we used that $\f{r} \geq \f{1} \beta^{r-1}$ and in the third step we further reduced the expression by substituting $z$.
Now we taker a closer look at the first product in the upper bound on $\F{\rone} / D_{\rone}$ in (\ref{eq:sync-1-ak-dk}) and employ the Weierstrass product inequality 
\[
	\prod_{r=1}^{{\rone}-1} \left(1 + 3\Delta(r,\delta) \right) \leq \frac{1}{1- \sum_{r=1}^{\rone -1} 3\Delta(r,\delta)} \leq 1 + 2 \sum_{r=1}^{{\rone} -1}3\Delta(r,\delta) \leq 1 + \frac{1}{8} \cdot \frac{\delta}{\sqrt{\f{1}}}.
\]
When it comes to error term of the lower bound on $\F{\rone} / D_{\rone}$ a similar approach yields
\[
	 \prod_{r=1}^{{\rone}-1} (1 - 3\Delta(r,\delta)) > 1 -  \sum_{r=1}^{{\rone} -1}3 \Delta(r,\delta) > 1 - \frac{1}{16} \frac{\delta}{\sqrt{\f{1}}} > 1 - \frac{1}{8}\frac{\delta}{\sqrt{\f{1}}}.
\]
Finally, we look at the error probability of (\ref{eq:sync-1-ak-dk}) and define $h_r = \exp(-c_\lambda \delta^2 r)$. For $\delta^2 \cdot c_\lambda \geq 1$, we have that $h_{r+1} \leq h_r \cdot e^{-1}$. Hence 
\begin{equation}
\label{eq:sync-1-error-sum}
	\sum_{r=1}^{{\rone} - 1} 4 \exp(-c_\lambda \delta^2 r) = 4 \sum_{r=1}^{{\rone} -1} h_r \leq 4h_0 \cdot \sum_{r=1}^{\infty}\frac{1}{e^{r}} < 7 \cdot h_0 = 7 \cdot \exp(-c_\lambda \delta^2)
\end{equation}
and the lemma statement follows.
\end{proof}

\begin{lemma}
\label{lem:sync-to-muk}
    Let $\rone$ be the first round such that $U_{\rone} < n/\lambda$ for some large constant $\lambda > 0$.
    Additionally, let $\rtwo$ be the first round such that $\delta > \sqrt{\F{\rtwo} \cdot U_{\rtwo} / n} \cdot 10$ \emph{or} $\delta > U_{\rtwo} / \sqrt{n} \cdot 10$. Then, 
	\[
		\frac{\F{\rtwo}}{D_{\rtwo}} = \frac{\F{\rone}}{D_{\rone}} \left( 1 \pm \frac{1}{8}\frac{\delta}{\sqrt{\F{\rone}}} \right)
	\]  
	with probability at least $1- c \exp(-c'\delta^2)$. Here $c,c' >0$ are constants.
\end{lemma}
\begin{proof}
	We fix the configuration at the start of some round $r \geq \rone$ and let $\F{r} = \f{r}$, $D_{r} = d_r$, $U_r = u_r$ and assume that $\delta \leq \sqrt{\mu_r} \cdot 10$ and $\delta \leq u_r / \sqrt{n} \cdot 10$. Here $\mu_r = \sqrt{\f{r} \cdot u_r /n}$ denotes the expected number of agents that join opinion $i$ throughout round $r$.  Just as in the proof of \cref{lem:sync-to-linear}, we model $\F{r+1} = \f{r} + \text{Bin}(u_r, \f{r} / n)$ we get by a Chernoff bound application that
	\begin{align}
	\label{eq:sync-2-a}
	\F{r+1} = \f{r} + \frac{u_r \cdot \f{r}}{n} \left(1 \pm \Delta_{\f{}}(r,\delta) \right) \text{ where } \Delta_{\f{}}(r,\delta) = z \cdot \begin{cases}
								\frac{\delta \cdot \sqrt{r - \rone + 1}}{\sqrt{\mu_r}} & \text{ if } \delta \cdot \sqrt{r - \rone + 1} \leq  \sqrt{\mu_r} \\
								\frac{\delta}{\sqrt{\mu_r}} &\text{ otherwise}
							\end{cases}
	\end{align}
	with  probability $1-2\cdot \exp(-\Delta_{\f{}}(r,\delta)^2 \cdot \mu_r / 3)$ and for some arbitrary constant $z < \frac{1}{96}$. Note that the above also implies the slightly weaker bound $\F{r+1} = \f{r} \cdot (1 + \frac{u_r}{n}) \left(1 \pm \Delta_{\f{}}(r,\delta) \frac{u_r}{n} \right)$. In a similar way, we model $D_{r+1}$ as $d_r + \text{Bin}(u_r, d_r/n)$ and get 
	\begin{equation}
	\label{eq:sync-2-d}
		D_{r+1} = d_r \left(1 + \frac{u_r}{n} \right) \left(1 \pm \Delta_{\f{}}(r,\delta) \frac{u_r}{n} \right)
	\end{equation}
	also with probability $1-2\cdot \exp(-\Delta_{\f{}}(r,\delta)^2 \cdot \mu_r / 3)$ as $d_r \geq \f{r}$ leads to even stronger concentration when applying Chernoff bounds to $\text{Bin}(u_r, d_r /n)$.
	
	Finally, we consider the decrease of $U_{r+1}$. By a Chernoff bound application, we have with probability at least $1 - \exp(-\frac{u_r^2}{3n})$ that 
	\begin{equation}
	\label{eq:sync-2-u}
		U_{r+1} < \frac{u_{r}^2}{n} \cdot 2.	
	\end{equation}
	Hence, it follows from a union bound application that the events in (\ref{eq:sync-2-a}),(\ref{eq:sync-2-d}) and (\ref{eq:sync-2-u}) hold with probability at least $1 - 4 \exp(-\Delta_{\f{}}(r,\delta) \cdot  \mu_r / 3) -  \exp(-\frac{u_r^2}{3n})$ in total. A repetition of above approach for increasing $r$ until the first round $\rtwo$ is reached such that either (i) $\delta > \sqrt{\mu_{\rtwo}} \cdot 10$, or (ii) $\delta > u_{\rtwo} / \sqrt{n}$ yields the following 
	\begin{align}
	    \label{eq:sync-2-mega-event}
		Pr\Big[\forall r \text{ with } \rone \leq r < \rtwo: &\F{r+1} = \f{r+1}, D_{r+1} = d_{r+1} \text{ and } U_{r+1} = u_{r+1} \text{ with }  \\ \nonumber
		&\f{r+1} = \f{r} \cdot (1 + \frac{u_r}{n}) \left(1 \pm \Delta_{\f{}}(r,\delta) \frac{u_r}{n} \right), \\ \nonumber
		&d_{r+1} = d_r \cdot (1 + \frac{u_r}{n}) \left(1 \pm \Delta_{\f{}}(r,\delta) \frac{u_r}{n} \right), \\ \nonumber
		&u_{r+1} \leq \frac{u_{r}^2}{n} \cdot 2 \Big] \\ \nonumber
		> &1 - 4\sum_{r=\rone}^{\rtwo - 1} \exp \left(-\Delta_{\f{}}(r,\delta)^2 \cdot \frac{u_r \cdot \f{r}}{n} \right) - 4 \sum_{r=\rone}^{\rtwo - 1} \exp(-\frac{u_r^2}{n}) \nonumber
	\end{align}
	This event implies the following statements for all $r$ with $\rone \leq r < \rtwo$:
	\begin{enumerate}
		\item $\f{r+1} < 3 \f{r}$ \label{enum:sync-2-item-1}
		\item $u_{r+1} \leq \frac{n}{(\lambda/2)^{2^{(r+1) - \rone }}}$
		\item $\Delta_{\f{}}(r, \delta) \neq z \cdot \frac{\delta \cdot \sqrt{r - \rone + 1}}{\sqrt{\mu_r}}$ implies that $r= \rtwo - 1$ \label{enum_sync-2-item-3}
		\item $\sum_{r=\rone}^{\rtwo - 1} \Delta_{\f{}}(r,\delta) \frac{u_r}{n}< 2z \cdot \frac{\delta}{\sqrt{\f{\rone}}} \leq 1/5$.
		\label{enum_sync-2-item-4}
	\end{enumerate}
	The first statement is straight forward and follows from the recursive bound on $\f{r}$ in (\ref{eq:sync-2-mega-event}). Similar, the second statement follows from the recursive bound on $u_r$. It can be shown, e.g., per induction over $\ell$  with the following hypothesis (note that the base case of $\ell=0$ holds as $u_{\rone} \leq n/\lambda$)
	\[
	    u_{\rone + \ell}\leq \frac{n}{\lambda^{2^\ell}} \cdot 2^{2^{\ell} -1}.
	\]
	Regarding the third statement, we assume that  $\Delta_{\f{}}(r,\delta) \neq z \cdot \frac{\delta \cdot \sqrt{r - \rone + 1}}{\sqrt{\mu_r}}$ for $r < \rtwo -1$. According to the definition of $\Delta_{\f{}}(r,\delta)$, this implies that $\delta \sqrt{r - \rone + 1} > \sqrt{\mu_r}$. We start by observing that $\mu_{r+1}$ shrinks significantly when compared to $\mu_r$
	\[
		\mu_{r+1} = \f{r+1} \cdot \frac{u_{r+1}}{n} \leq 6 \f{r} \frac{u_r}{n} \cdot \frac{u_r}{n} = \mu_{r} \cdot 6 \frac{u_r}{n} \leq \mu_r \frac{6}{(\lambda /2)^{2^{r - \rone}}}.
	\]
	Here we used \cref{enum:sync-2-item-1} as well as the bound on $u_{r+1}$ of (\ref{eq:sync-2-mega-event}) in the first step and \cref{enum_sync-2-item-3} in the last step. From this and  $\delta \sqrt{r - \rtwo + 1} > \sqrt{\mu_r}$ it follows that $\delta > \sqrt{\mu_{r + 1}} \cdot 10$ in case $\lambda > 0$ is a large enough constant. This in turn implies that $r + 1 = \rtwo$ and the third statement follows. Finally, we show the fourth statement 
	\begin{align*}
		\sum_{r = \rone}^{\rtwo - 1} \Delta_{\f{}}(r,\delta) \frac{u_r}{n} < z \cdot \sum_{r=\rone}^{\rtwo - 1} \left( \frac{\delta \cdot \sqrt{r - \rtwo + 1}}{\sqrt{\f{r}}} \cdot \sqrt{\frac{u_r}{n}} \right) + z \cdot \frac{\delta}{\sqrt{\f{\rtwo - 1}}} \sqrt{\frac{u_{\rtwo - 1}}{n}}  \\
		< z \cdot \frac{\delta}{\sqrt{\f{\rone}}} \sum_{\ell = 0}^{\infty} \frac{\sqrt{\ell  + 1}}{(\lambda/2)^{2^{\ell}}} + z \cdot \frac{\delta}{\sqrt{\f{\rone}}} \frac{1}{\lambda} < 2z \cdot \frac{\delta}{ \sqrt{\f{\rone}}} \leq \frac{1}{48}.
	\end{align*} 
	In the first step, we use \cref{enum_sync-2-item-3} to split the sum. In the third step we assume $\lambda > 0$ to be a large enough constant s.t. the infinity series can be bounded by $1$.
In the last step we just used that $\delta < \sqrt{\f{\rone}}$ and $z < 1/96$.

We are now ready to finalize the proof. The recursive bounds on $\f{r+1}$ and $d_{r+1}$ in (\ref{eq:sync-2-mega-event}) imply that
\[
	\frac{\F{\rtwo}}{D_{\rtwo}} = \frac{\f{\rone}}{d_{\rone}} \prod_{r= \rone}^{\rtwo - 1}\left(1 \pm 3 \Delta_{\f{}}(r,\delta) \frac{u_r}{n} \right),
\]	
which can be simplified with the help of \cref{enum_sync-2-item-4} and the Weierstrass product inequality. This yields the desired concentration statement
\[
	\frac{\F{\rtwo}}{D_{\rtwo}} = \frac{\f{\rone}}{d_{\rone}}\left(1 \pm 6 \sum_{r = \rone}^{\rtwo - 1}\Delta_{\f{}}(r,\delta) \frac{u_r}{n}  \right)  \asymp \frac{\f{\rone}}{d_{\rone}} \left( 1 \pm \frac{1}{8} \frac{\delta}{\sqrt{\f{\rone}}}\right).
\]
Remember, the probability for this (see (\ref{eq:sync-2-mega-event})) is 
\begin{equation}
\label{eq:sync-2-prob}
1 - 4\sum_{r=\rone}^{\rtwo - 1} \exp \left(-\Delta_{\f{}}(r,\delta)^2 \cdot \frac{u_r \cdot \f{r}}{n} \right) - 4 \sum_{r=\rone}^{\rtwo - 1} \exp(-\frac{u_r^2}{3n})
\end{equation}
with the help of \cref{enum_sync-2-item-3} we may again split the first sum and get
\begin{align}
\label{eq:sync-2-prob-1}
	 4\sum_{r=\rone}^{\rtwo - 1} \exp \left(-\Delta_{\f{}}(r,\delta)^2 \cdot \frac{u_r \cdot \f{k}}{n} \right) \leq 4 \sum_{r=\rone}^{\rtwo - 1} \exp \left(-z^2 \delta^2 (r - \rtwo + 1)) \right) + 4\exp(-z^2 \delta^2) \\
	 \leq 7 \exp(-\delta^2) +  4\exp(-z^2 \delta^2) < 11 \exp(-z^2 \delta^2) \nonumber.
\end{align}
Next, we bound the second sum in (\ref{eq:sync-2-prob}). Using the recursive bound on $u_r$ in (\ref{eq:sync-2-mega-event}) and  $u_r < n/\lambda$ it yields that 
\[
u_{\rtwo - 1} \leq (u_{\rtwo-1-r})\cdot(\frac{2}{\lambda})^{r} \text{ for $r$ with } 0 \leq r < \rtwo - \rone.
\]
From this we get that 
\[
\frac{u_{\rtwo -1 -r}^2}{n} \ge (\lambda/2)^{2r}\cdot \frac{u_{\rtwo - 1}^2}{n}> (r+1) \cdot\frac{u_{\rtwo - 1}^2}{n} > (r+1) \cdot \frac{\delta^2}{100} \text{ for $r$ with } 0 \leq r < \rtwo - \rone
\]
where we used that $\delta \leq u_{\rtwo - 1} / \sqrt{n} \cdot 10$ in the last step.
Therefore, we can bound the second sum in (\ref{eq:sync-2-prob}) by 
\begin{align}
\label{eq:sync-2-undecided-prob}
4\sum_{r=\rone}^{\rtwo-1}\exp\Big(-\frac{u_r^2}{3n}\Big)\le 4 \sum_{r=0}^{\rtwo-\rone - 1}\exp\Big( -\frac{u_{\rtwo - 1}^2}{3n}\Big)^{(r+1)} \le \sum_{r=0}^{\rtwo - \rone - 1} \exp(-\delta^2 \cdot (r+1) / 300) 
 \\ \nonumber 
 \leq  c \cdot  \exp(-c \cdot \delta^2)
\end{align}
where $\delta = \BigOmega{1}$ guarantees that the last sum in the first line increases at most as fast as a geometric series, and $c >0$ is some fitting constant.
\end{proof}

\begin{lemma}
\label{lem:sync-til-uk}
    Let $\rtwo$ be the first round such that $U_{\rtwo} < n/\lambda$ for some large constant $\lambda  >0$ and (a) $\delta > \sqrt{\F{\rtwo} \cdot U_{\rtwo} / n} \cdot 10$ \emph{or} (b) $\delta > U_{\rtwo} / \sqrt{n} \cdot 10$.
    Additionally, let $\rthree$ be the first round where $\delta > U_{\rthree} / \sqrt{n} \cdot 10$. Then, we have 
	\[
		\frac{\F{\rthree}}{D_{\rthree}} = \frac{\F{\rtwo}}{D_{\rtwo}} \left(1 \pm \frac{1}{8} \cdot \frac{\delta}{\sqrt{\F{\rtwo}}} \right)
	\]
	with probability $1-c\exp(-c' \delta^2)$ where $c,c'$ are constant.
\end{lemma}
\begin{proof}
    At the start of round $\rtwo$ at least one of the conditions (a) or (b) must hold. If (b) is fulfilled, then $\rtwo = \rthree$ and we are finished because $\F{\rthree} / D_{\rthree} = \F{\rtwo} / D_{\rtwo}$.

    Otherwise, in $\rtwo$ only (a) holds and we need to analyze the process  until we reach the round $\rthree$ where $\delta > U_{\rthree} / \sqrt{n} \cdot 10$.
    We will proof the upper and lower bound separately.
\paragraph{Upper Bound}
	We consider some arbitrary round $r \geq \rtwo$ and fix $\F{r} = \f{r}$, $U_r = u_r$ and $D_r = d_r$. Furthermore, we define $\mu_r = \sqrt{a_r u_r / n}$ and $\delta_r = \delta \cdot \sqrt{r - \rtwo +1}$. We assume that (i) $\delta_r \geq 10  \sqrt{\mu_r}$, and (ii) $\delta^2 \leq u_r^2 / n$. As in the proofs of the  \cref{lem:sync-to-muk,lem:sync-to-linear}, we model $\F{r+1}= \f{r} + \Bin(u_r, \frac{\f{r}}{n})$. We apply \cref{lem:chernoff-superexponential} to bound the far-right tail of this  Binomial distribution. As we assumed that $\delta_r^2 > 100 \cdot \mu_r$ this yields 
    \begin{equation}
        \F{r+1} < \f{r} + \frac{\f{r} \cdot u_r}{n} + \frac{\delta_r^2}{1 + \ln(\delta_r^2 / \mu_r)}
    \end{equation}
    with probability at least $1 - \exp(-c_1 \delta_r^2)$ for some constant $c_1 > 0$. For $\Delta_{\f{}}(r,\delta) = \frac{\delta_k^2}{\f{k} \cdot (1 + \ln(\delta_r^2 / \mu_r))}$
    this also implies the slightly weaker bound of 
	\begin{equation}
	\label{eq:sync-3-a}
		\F{r+1} < \f{r}(1 + \frac{u_r}{n}) ( 1 + \Delta_{\f{}}(r,\delta)).
	\end{equation}
	In a similar way, we model $D_{r+1} = d_r + \Bin(u_r, \frac{d_r}{n})$ and apply Chernoff bounds to derive that 
    \begin{align}
    \label{eq:sync-3-d}
        D_r > d_r + \frac{u_r \cdot d_r}{n} (1 - \Delta_d(r,\delta)) \text{ where } \Delta_d(r,\delta) = \frac{1}{15} \cdot \begin{cases}
                    \frac{\delta_r}{\sqrt{d_r \cdot u_r / n}} &\text{if $\delta_r < \sqrt{d_r \cdot u_r / n}$} \\
                    \frac{\delta}{\sqrt{d_r \cdot u_r / n}} &\text{otherwise}
                \end{cases}
    \end{align}
    with probability at least $1-\exp(-\Delta_d(r,\delta)^2 \cdot \frac{u_r d_r}{n})$. Note that the bound in  (\ref{eq:sync-3-d}) also implies the following weaker bound 
    \begin{equation}
    \label{eq:sync-3-d2}
        D_r > d_r \left( 1 + \frac{u_r}{n}\right) \left( 1 - \Delta_d(r,\delta) \frac{u_r}{n}\right).
    \end{equation}
    Finally, just as in the proof of \cref{lem:sync-to-muk}, we apply Chernoff bounds and deduce that with probability $1 - \exp(-\frac{u_r^2}{n})$ we have 
    \begin{equation}
    \label{eq:sync-3-u}
        U_{r+1} < \frac{u_r^2}{n} \cdot 2.
    \end{equation}
    When employing union bounds, we deduce that the events in (\ref{eq:sync-3-a}),(\ref{eq:sync-3-d}) and (\ref{eq:sync-3-u}) occur at the same time with probability at least $1 - \exp(c_1 \delta_r^2) -  \exp(-\Delta_d(r,\delta)^2 \frac{u_r d_r}{n}) - \exp(-\frac{u_r^2}{3n})$. 
    Note that from (\ref{eq:sync-3-a}) and $\delta < \sqrt{\f{\rtwo}}$ it easily follows that $\F{r+1} < 3 \f{r}$. This, combined with (\ref{eq:sync-3-u}) and $u_r \leq n/\lambda$, implies 
    \[
        \mu_{r+1} = \F{r+1} \cdot \frac{U_{r+1}}{n} < 6 \cdot \f{r} \frac{u_r}{n} = 6 \cdot \frac{u_r}{n} \mu_r \leq \frac{6}{\lambda} \mu_r.
    \]
    Our assumption (i) stated that $\delta_r > \sqrt{\mu_r} \cdot 10$. Note that $\delta_{r+1} > \delta_{r} > \sqrt{\mu_r} \cdot 10 > \sqrt{\mu_{r}} \cdot 10 \cdot \sqrt{\frac{6}{\lambda}} > \mu_{r+1}$ as long as $\lambda >6$.
    In other words, assumption (i) can be established inductively.
    
    Therefore, starting with $r=\rtwo$, we can repeat this approach until assumption (ii) is violated. This way, we have the following, where $\rthree$ denotes the first round s.t. $\delta > U_{\rthree} / \sqrt{n}$
    \begin{align}
	    \label{eq:sync-3-mega-event}
		\Pr \Big[\forall r \text{ with } \rtwo \leq r < \rthree: &\F{r+1} = \f{r+1}, D_{r+1} = d_{r+1} \text{ and } U_{r+1} = u_{r+1} \text{ with }  \\ \nonumber
		&\f{r+1} < \f{r} \cdot (1 + \frac{u_r}{n}) \left(1 + \Delta_{\f{}}(r,\delta) \right), \\ \nonumber
		&d_{r+1} > d_r \cdot (1 + \frac{u_r}{n}) \left(1 - \Delta_d(r,\delta) \frac{u_r}{n} \right), \\ \nonumber
		&u_{r+1} < \frac{u_{r}^2}{n} \cdot 2 \Big] \\ \nonumber
		> &1 - \sum_{r=\rtwo}^{\rthree - 1} \exp \left(-c_1 \delta_{r}^2 \right) -  \sum_{r=\rone}^{\rthree - 1} \exp(-\Delta_d(r,\delta) \frac{u_r d_r}{n}) - \sum_{r=\rtwo}^{\rthree -1} \exp(-\frac{u_r^2}{3n}) \nonumber
	\end{align}
    These recursive bounds allow for the following statements to be made for $r$ s.t. $\rtwo \leq r < \rthree$
    \begin{enumerate}
        \item $u_{r} <  \frac{n}{(\lambda/2)^{2^{r - \rtwo}}}$ \label{enum:sync-3-item-1}
        \item $\delta_r^2 \geq \mu_r \cdot \frac{100}{6} \cdot e^{2^{r - \rtwo} + 7}$ \label{enum:sync-3-item-2}
        \item $\Delta_d(r,\delta) \neq \frac{1}{10} \frac{\delta_r}{\sqrt{d_r u_r / n}}$ implies that $r = \rthree -1$ \label{enum:sync-3-item-3}
        \item $\sum_{r=\rtwo}^{\rthree -1} \Delta_{\f{}}(r,\delta) < \frac{1}{8} \frac{\delta^2}{\sqrt{\f{\rtwo}}}$ \label{enum:sync-3-item-4}
        \item $\sum_{r=\rtwo}^{\rthree -1} \Delta_d(r,\delta) \frac{u_r}{n} < \frac{1}{10} \frac{\delta}{\sqrt{d_r}}$ \label{enum:sync-3-item-5}
    \end{enumerate}
    Just as in the proof of \cref{lem:sync-to-muk}, the first statement follows from the recursive bound $u_{r+1}<\frac{u_r^2}{n}$ in (\ref{eq:sync-3-mega-event}). Regarding the second point, for $r=\rtwo$ the statement holds by definition of round $\rtwo$. For $\rtwo < r < \rthree$, we have that
    \[
        \mu_r = \f{r} \frac{u_r}{n} \leq 6 \f{r-1} \frac{u_{r-1}}{n} \cdot \frac{u_{r-1}}{n} = 6 \mu_{r-1} \cdot \frac{u_{r-1}}{n} \leq \dots \leq \mu_{\rtwo} \prod_{r=\rtwo}^{r - 1} 6 \frac{u_{r}}{n} \leq \mu_{\rtwo} \frac{6}{(\lambda /2)^{2^{r - \rtwo - 1}}}.
    \]
    In the second step we used again that (\ref{eq:sync-3-mega-event}) implies $\f{r+1} < 3\f{r}$. In the final step, we used \cref{enum:sync-3-item-1} and crudely bounded the product by its smallest factor (note that large enough $\lambda> 0$ ensures that none of the factors are larger than one). We now use this intermediate result and deduce that
    \[
        \delta_r^2 > \delta^2 \geq \mu_{\rtwo} \cdot 100 \geq \mu_{r} \cdot 100 \cdot  \frac{(\lambda/2)^{2^{r-\rtwo - 1}}}{6} > \mu_r \cdot \frac{100}{6} \cdot e^{2^{r - \rtwo + 7}},
    \]
    where in the last step we assume $\lambda$ to be a large enough constant. We continue with the proof of \cref{enum:sync-3-item-3}. Assume that $\Delta_d(r,\delta) \neq \frac{1}{10}\frac{\delta_r}{\sqrt{d_r u_r / n}}$ for some $r$ with $\rtwo \leq r <\rthree - 1$. This implies that $\delta^2 \geq (d_r u_r /n) / (r - \rtwo + 1)$. Moreover, we have that
    \[
        \frac{u_{r+1}^2}{n} \leq \frac{u_r^2}{n} = \frac{d_r \cdot u_r}{n} \cdot \frac{u_r}{d_r} \cdot  \leq \frac{d_r \cdot u_r}{n} \cdot \frac{1}{(1- 1/\lambda)} \cdot  \frac{1}{(\lambda/2)^{2^{r-\rtwo}}} < \frac{d_r\cdot u_r}{n} \cdot \frac{1}{r - \rtwo  + 1} \leq \delta^2,
    \]
    where we crudely bounded $d_r \geq n(1- 1/\lambda)$ and employed \cref{enum:sync-3-item-1} in the third step. The penultimate step follows for large enough constant $\lambda > 0$.
    This way, we just established that $\delta > \sqrt{u_{r+1}} / \sqrt{n}$, which in turn implies $r +1 = \rthree$ and \cref{enum:sync-3-item-3} follows.
    We proceed with \cref{enum:sync-3-item-4} and observe 
	\[
		\sum_{r=\rtwo}^{\rthree - 1} \Delta_{{\f{}}}(r,\delta) \leq \frac{\delta^2}{\f{\rtwo}} \cdot \sum_{r=\rtwo}^{\rthree-1} \frac{r - \rtwo +1}{1 + \ln(\frac{\delta_r^2}{\mu_r})}
		\leq \frac{\delta^2}{\f{\rtwo}} \sum_{r=\rtwo }^{\infty} \frac{r - \rtwo +1}{1 + 2^{r - \rtwo + 7} + \ln(100/6)}  \leq \frac{\delta^2}{\f{\rtwo}} \cdot  \frac{1}{20},
 	\] 
    where we employed \cref{enum:sync-3-item-2} in the second step to bound $\delta_r^2 / \mu_r$. The infinite sum clearly converges to some constant as the denominator easily dominates the numerator. We continue with the final statement. We have 
    \begin{align*}
	\sum_{r=\rtwo}^{\rthree - 1} \Delta_d(r,\delta) \cdot \frac{u_r}{n} \leq \left( \sum_{r=\rtwo}^{\rthree - 2} \frac{1}{15} \frac{\delta_r}{\sqrt{d_r}} \cdot \sqrt{\frac{u_r}{n}} \right) + \frac{1}{15} \frac{\delta}{\sqrt{d_{\rtwo}}} \cdot \frac{u_{\rtwo}}{n} \\
	\leq \frac{1}{15} \frac{\delta}{\sqrt{d_r}} \left(1 + \sum_{r=\rtwo}^{\infty} \sqrt{\frac{r - \rtwo + 1}{(\lambda/2)^{2^{r - \rtwo}}}}\right)
	\leq \frac{1}{10} \frac{\delta}{\sqrt{d_r}}.
\end{align*}
    Here we employed \cref{enum:sync-3-item-1} in order to split the sum in the first step and bound $u_r$ with the help of \cref{enum:sync-3-item-1} in the second step. In the last step we assume $\lambda >0$ to be a sufficiently large constant.

    We are now ready to translate the statement in (\ref{eq:sync-3-mega-event}) into a concentration result. We first use the recursive bounds on $\f{r+1}$ and $d_{r+1}$ followed by an application of the Weierstrass product inequality and finally use \cref{enum:sync-3-item-4,enum:sync-3-item-5} which yields 
    \[
        \frac{\F{\rthree}}{D_{\rthree}} < \frac{\f{\rtwo}}{d_{\rtwo}} \prod_{r=\rtwo}^{\rthree -1} \frac{(1 + \Delta_{\f{}}(r,\delta))}{(1-\Delta_d(r,\delta) \frac{u_k}{n})} \leq \frac{\f{\rtwo}}{d_{\rtwo}} \cdot  \frac{1 + 2 \sum_{r=\rtwo}^{\rthree -1} \Delta_{\f{}}(r,\delta)}{1 -  \sum_{r=\rtwo}^{\rthree -1} \Delta_d(r,\delta) \frac{u_r}{n}} < \frac{\f{\rtwo}}{d_{\rtwo}}\left( 1 + \frac{1}{8} \cdot \frac{\delta}{\sqrt{\f{r}}} \right).
    \]
    What remains is to take a closer look at the probability of (\ref{eq:sync-3-mega-event}). Along the lines of (\ref{eq:sync-1-error-sum}) in the proof of \cref{lem:sync-to-linear}, it follows that $\sum_{r=\rtwo}^{\rthree -1} \exp(-c_1 \delta_r^2) < 7 \exp(-c_1 \delta^2)$ when using that $\delta^2 > 1/c_1 = \BigOmega{1}$. The second sum may be bounded along the lines of (\ref{eq:sync-2-prob-1}) in the proof of \cref{lem:sync-to-muk}, where we employ \cref{enum:sync-3-item-3} to split the sum. A bound for the third sum can be developed along the lines of   (\ref{eq:sync-2-undecided-prob}) which is contained in the proof of  \cref{lem:sync-to-muk}. In total, this allows us to lower bound the probability of (\ref{eq:sync-3-mega-event}) by  $1 - c' \cdot \exp(-c' \delta
   ^2)$ for some constant $c' > 0$.
    
    \paragraph{Lower Bound} Per assumption we know that at the start of round $\rtwo$ we have that $100 \cdot \mu_{\rtwo}< \delta^2$. In the worst case, no remaining agent adopts opinion $i$. Therefore, we have for any round $r \geq \rtwo$ that 
\[
	\frac{\F{r}}{D_{k}} \geq \frac{\f{\rtwo}}{d_{\rtwo} +  u_{\rtwo}} = \frac{\f{\rtwo}}{d_{\rtwo}} \left( \frac{1}{1 + \frac{u_{\rtwo}}{d_{\rtwo}}}\right) > \frac{\f{\rtwo}}{d_{\rtwo}}(1 - \frac{u_{\rtwo}}{d_{\rtwo}}).  
\]
As $100 \cdot \mu_{\rtwo}= 100 \cdot \frac{\f{\rtwo} u_{\rtwo}}{n} \leq \delta^2$ and $d_{\rtwo} \geq n(1 - 1/\lambda)$, we have 
\[
	\frac{u_{\rtwo}}{d_{\rtwo}} \leq  \frac{u_{\rtwo}}{(1- \frac{1}{\lambda})n} \leq \frac{1 +\frac{2}{\lambda}}{100} \cdot \frac{\delta^2}{\f{\rtwo}}  < \frac{1}{8} \cdot \frac{\delta^2}{\f{\rtwo}}
\]
and conclude the proof of the lower bound.
\end{proof}

% \begin{lemma}
% \label{lem:sync-til-uk}
%     Let $\rtwo$ be the first round such that $\delta > \sqrt{\F{\rtwo} \cdot U_{\rtwo} / n} \cdot 10$ and $U_{\rtwo} < n/\lambda$ for some  large constant $\lambda  >0$.
%      Additionally, let $\rthree$ be the first round where $\delta > U_{\rthree} / \sqrt{n} \cdot 10$. Then, we have 
% 	\[
% 		\frac{\F{\rthree}}{D_{\rthree}} = \frac{\F{\rtwo}}{D_{\rtwo}} \left(1 \pm \frac{1}{8} \cdot \frac{\delta}{\sqrt{\F{\rtwo}}} \right)
% 	\]
% 	with probability $1-c\exp(-c' \delta^2)$ where $c,c'$ are constant.
% \end{lemma}

\begin{lemma}
\label{lem:sync-4-last-steps}
    Let $\rthree$ be the first round where $\delta > \frac{U_{\rthree}}{\sqrt{n}} \cdot 10$  and $U_{\rthree} < n/\lambda$ for a large enough constant $\lambda > 0$. Additionally, let $U_{\rfour}$ be the first round with $U_{\rfour} = 0$. Then, we have 
	\[
		\frac{\F{\rfour}}{D_{\rfour}} = \frac{\F{\rfour}}{n}= \frac{\F{\rthree}}{D_{\rthree}} \left(1 \pm \frac{1}{8} \cdot \frac{\delta}{\sqrt{\F{\rthree}}} \right)
	\]
	with probability $1-c\exp(-c' \delta^2)$ where $c,c' >$ are constants.
\end{lemma}
\begin{proof}
    We abbreviate $r={\rthree}$. Fix $U_r = u_r$, $\F{r} = \f{r}$, and $D_r = d_r$. First we consider the concentration of $\F{r+1} / D_{r+1}$. We distinguish two cases.
    
    \paragraph{Case 1: $\delta \leq \sqrt{\f{r} u_r / n} \cdot 10$} As usual we model $\F{r+1} = \f{r} + \Bin(u_r, \frac{\f{r}}{n})$. Chernoff bounds yield with probability at least $1-2\exp(-\delta^2 /300)$ that 
    \[
        \F{r+1} = \f{r} + \frac{\f{r} u_r}{n}(1 \pm   \frac{1}{10}\frac{\delta}{\sqrt{\f{r} u_r / n}}) \asymp \f{r} (1 + \frac{u_r}{n}) (1 \pm  \frac{1}{10} \frac{\delta \sqrt{\frac{u_r}{n}}}{\sqrt{\f{r}}}) \asymp \f{r} (1 + \frac{u_r}{n}) \left(1 \pm \frac{1}{10 \sqrt{\lambda}} \cdot \frac{\delta}{\sqrt{a_r}}\right)
    \]
    A similar approach yields that $D_{r+1} = d_{r} (1 + \frac{u_r}{n}) \left(1 \pm \frac{1}{\lambda} \cdot \frac{\delta}{\sqrt{\f{r}}}\right)$ holds also with probability at least $1 -2\exp(-\delta^2 / 3)$ and a union bound application yields 
    \[
        \Pr \left[ \frac{\F{r+1}}{D_{r+1}} = \frac{\f{r}}{d_r} (1 \pm \frac{4}{10 \cdot \sqrt{\lambda}} \frac{\delta}{\sqrt{\f{r}}})  \right] \geq 1 - 4 \exp(-\delta^2 /3).
    \]
    \paragraph{Case 2: $\delta > \sqrt{\f{r} u_r / n} \cdot 10$} In this case, the lower bound can be shown with completely the same argument as the lower bound of \cref{lem:sync-til-uk}. Which yields that (with probability $1$)
    \[
        \frac{\F{r}}{D_r} \geq \frac{\f{r}}{d_r} \left(1 - \frac{1+\frac{2}{\lambda}}{100} \frac{\delta^2}{\sqrt{\f{r}}} \right) \geq \frac{\f{r}}{d_r} \left(1 - \frac{1}{50}\frac{\delta}{\sqrt{\f{r}}}\right).
    \]
    For the upper bound, we again use that $\F{r+1} = \f{r} + \Bin(u_r, \f{r} / n)$
    A Chernoff bound application yields that, with probability $1 - \exp(-\delta^2 / 300)$, 
    \[
        \F{k+1} = \f{r} + \Bin(u_r, \f{r} / n) < \f{r} + \frac{u_r \cdot \f{r}}{n} + \delta^2 / 100 < \f{r} + \frac{\delta^2}{50}
    \]
    as $\delta^2 / 100 > \sqrt{\f{r} u_r / n}$. Therefore, it follows that
    \[
       \frac{\F{r+1}}{D_{r+1}} < \frac{\f{r} + \delta^2 / 50}{d_r} \leq \frac{\f{r}}{d_r}(1 + \frac{\delta^2}{50 \cdot \f{r}}) < \frac{\f{r}}{d_r} (1 + \frac{1}{50}\frac{\delta}{\sqrt{\f{r}}}).
    \]
    Hence, in any case, $\F{r+1} / D_{r+1}$ is concentrated around $\f{r} / d_r$ with relative error of order $(1 \pm \delta / \sqrt{\f{r}})$.
    
    Now, consider $U_{r+1} = \Bin(u_r, u_r / n)$. According to our assumption we have $\delta > \frac{u_r}{\sqrt{n}} \cdot 10$ which implies that  $\Ex{\Bin(u_r, u_r / n)} \leq \delta^2 / 100$ and a Chernoff bound application yields that, with probability at least $1-\exp(-\delta^2 / 300)$, we have  $U_{r+1} \leq \delta^2 / 50$. Continuing from round $r+1$ until $\rfour$, we assume that either all or none of the remaining agents adopt opinion $i$. As only $\delta^2 / 50$ agents remain this causes an additional multiplicative error term of at most $(1 \pm \frac{\delta^2}{50 \cdot \f{r}})$ and the result follows.
\end{proof}

\section{A Note on Uniformity}
\label{sec:uniform-protocol}

\Vardef{Stage}
\Vardef{Bit}
\Vardef{Hit}
\Vardef{OddRound}
\Vardefx{InitialOpinion}{initial\std{O}pinion}
\def\Init{\mathsf{init}}
\def\Count{\mathsf{count}}
\def\Sync{\mathsf{sync}}
\def\Run{\mathsf{run}}

In our basic protocol for the gossip model, all agents run (synchronously) one round of the \Decision part followed by \BigTheta{\log{n}} many rounds of the \Boosting part. In order to distinguish between those two parts, agents use a counter modulo \BigTheta{\log{n}}. 
Counting modulo \BigTheta{\log{n}}, however, requires knowledge of $n$.
In order to answer the open question posed in \cite{DBLP:journals/dc/BecchettiCNPST17}, we present a protocol which does not require knowledge of $n$ (or any function of $n$).
Such a protocol is called \emph{uniform}, since any such protocol can be applied to any population size without adaption.
In the following, we therefore assume a restrictive variant of the gossip model, where agents neither know $n$ nor $k$, nor do they have access to a source of randomness. (This does not apply to the interactions, though, which are still uniformly at random).

The main challenge is to get an approximation of $\log{n}$ which can then be used to replace $\T{BC}$ in \cref{alg:consensus-gossip-model}.
For this, we use the entropy encoded in the opinions of the agents when choosing a random communication partner.
That is, each agent may remember whether it has seen its own opinion when communicating with a random partner.
All agents that are marked after a two step selection process start a broadcast (we will see that this is at most a $1/4$-th fraction of all agents).
All agents $u$ that start a broadcast have a bit $\Bit u = \tfalse$.
Whenever an agent $v$ receives a broadcast by interacting with another agent $u$ that has either started or received a broadcast in an earlier round, it sets its own bit such that $\Bit v = \lnot \Bit u$.
This way, the agents can easily be divided into two sets, each of which consist of a constant fraction of all agents.
Applying now a counting procedure as in \cite{DBLP:conf/soda/AlistarhAEGR17}, one of the agents obtains a value $T$ of order $T = \Theta(\log n)$, which is then sent to all other agents in the system via a maximum broadcast.
This way, all agents can act in a completely synchronous manner, and our protocol can be deployed as described in \cref{sec:gossip-model} by setting $\T{BC}$ to the value of $T$. 

On an intuitive level, the agents run through three stages before executing the actual protocol in stage 4.
In the first stage, the initialization stage, agents will generate a synthetic coin (cf.\ \cite{DBLP:conf/soda/AlistarhAEGR17}). This coin is stored by each agent $u$ in the variable $\Bit u$. It is generated by a broadcast process.
In the second stage, the agents initialize a counter $T_u$. This counter $T_u$ is incremented by agent $u$ for as many rounds as agent $u$ interacts with another agent $v$ with $\Bit v = \Bit u$.
Note that a similar approach for approximating $\log{n}$ up to constant factors was already described in \cite{DBLP:conf/soda/AlistarhAEGR17}.
In the third stage, the synchronization stage, all agents $u$ broadcast the maximum counter value $T_u$.
In round $1000\cdot T_u$, the agents leave the synchronization stage and execute a variant of \cref{alg:consensus-gossip-model}.
The first three stages will make sure that each agent $u$ has \whp the same value $T_u = \BigTheta{\log{n}}$.
The final stage, the running stage, executes the actual protocol, with two minor changes: First, agents count the rounds modulo $1000\cdot T_u$. (Note that the factor of $1000$ was chosen for convenience and is not tight.)
Secondly, if agent $u$ encounters while executing the running stage another agent $v$ with $T_v > T_u$, it aborts the running stage, reverts its current opinion to the initial opinion, and goes back to the synchronization stage.

Formally, every agent $u$ has variables $\Stage u \in \set{\Init, \Count, \Sync, \Run}$, $T_u \in \mathbb{N}$, $\Round u \in \set{0, \dots, T}$, $\Opinion u, \InitialOpinion u \in \set{1, \dots, k}$, $\Bit u$, $\Hit u$, $\Undecided u \in \set{\ttrue , \tfalse}$.
The variables $\Opinion u$, $\Round u$ and $\Undecided u$ have the same meaning as in \cref{alg:consensus-gossip-model}, $\InitialOpinion u$ holds the initial opinion of agent $u$, $T_u$ (initially 0) is used to store an approximation of $\log n$, and the flags $\Bit u$ and $\Hit u$ (both initially \tfalse) are used in the first and second stage to set up $T_u$.
The variable $\Stage u$ (initialized to $\Init$) determines which set of instructions is used:
\begin{itemize}
\item The initialization stage ($\Stage u = \Init$) is defined in \cref{alg:uniform-1},
\item the counting stage ($\Stage u = \Count$) is defined in \cref{alg:uniform-2},
\item the synchronization stage ($\Stage u = \Sync$) is defined in \cref{alg:uniform-3}, and
\item the running stage ($\Stage u = \Run$) is defined in \cref{alg:uniform-4}.
\end{itemize}

\begin{figure}[p]
\noindent\begin{minipage}[t]{0.5\textwidth-1em}
\begin{lstalgo}[H]{Initialization stage\label{alg:uniform-1}}
Actions performed when agents $(u, v)$ interact
with $\Stage u = \Init$:

if $\Round u$ is even then$\label{alg:uniform-1-1}$
	if $\InitialOpinion u = \InitialOpinion v$ then
		$\Hit u \gets \ttrue$
	else
		$\Hit u \gets \tfalse$

if $\Round u$ is odd then
	if $\Hit u \land \Opinion u \neq \Opinion v$ then
		$\Stage u \gets \Count\label{alg:uniform-1-12}$

if $\Stage v \in \set{\Count, \Sync, \Run}$ then
	$\Stage u \gets \Count$
	$\Bit u \gets \lnot \Bit v$

$\Round u \gets \Round u + 1$
\end{lstalgo}
\end{minipage}\hfill\begin{minipage}[t]{0.5\textwidth-1em}
\dcmlinenumbersfalse
\begin{lstalgo}[H]{Counting stage\label{alg:uniform-2}}
Actions performed when agents $(u, v)$ interact
with $\Stage u = \Count$:

if $\Stage v \neq \Init$ then
	if $\Bit u = \Bit v$ then
		$T_u \gets T_u + 1$
	else
		$\Stage u \gets \Sync$

$\Round u \gets \Round u + 1$








\end{lstalgo}
\end{minipage}

\noindent\begin{minipage}[t]{0.5\textwidth-1em}
\begin{lstalgo}[H]{Synchronization stage\label{alg:uniform-3}}
Actions performed when agents $(u, v)$ interact
with $\Stage u = \Sync$:

if $T_u < T_v$ then
	$T_u \gets T_v$
	$\Round u \gets \Round v$
	$\Opinion u \gets \InitialOpinion u$

$\Round u \gets \Round u + 1$
if $\Round u \geq 1000 \cdot T_u$ then
	$\Round u \gets 0$
	$\Stage u \gets \Run$

	
	

\end{lstalgo}
\end{minipage}\hfill\begin{minipage}[t]{0.5\textwidth-1em}
\dcmlinenumbersfalse
\begin{lstalgo}[H]{Running stage\label{alg:uniform-4}}
Actions performed when agents $(u, v)$ interact
with $\Stage u = \Run$:

if $T_u < T_v$ then
	$\Stage u \gets \Sync$
	execute $\Sync$ stage

if $T_u = T_v$ then
    /* Decision Part: {\normalfont $\Round u = 0$} */
    if $\Round u = 0$ then
    	if $\Opinion u \neq \Opinion v$ then
    		$\Undecided u \gets \ttrue$
    	else
    	    $\Undecided u \gets \tfalse$
    
    /* Boosting Part: {\normalfont $\Round u > 0$} */
    if $\Round u > 0$ and $\Undecided u$ then
    	if not $\Undecided v$ then
    		$\Undecided u \gets \tfalse$
    		$\Opinion u \gets \Opinion v$
    
    /* Synchronization by counting modulo $1000 T_u$ */
	$\Round u \gets (\Round u + 1) \bmod (1000 T_u)$
\end{lstalgo}
\end{minipage}
\end{figure}

We start our analysis by analyzing the time $\rho$ when all agents have completed the counting stage. We show that $\rho = \BigO{\log{n}}$. Then we show that at time $\rho$ the maximum value $T_u(\rho)$ over all agents $u$ is in $\BigOmega{\log{n}}$.
\begin{lemma}
Let $\rho = \min\set{r | \forall u : \Stage u(r) \neq \Init \wedge \Stage u(r) \neq \Count}$.
Then, \whp
\begin{enumerate}
\item $\rho < 50\log{n}$, \label{statement-rho-1}
\item $\max_u\set{T_u(\rho)} > 0.1\log{n}$. \label{statement-rho-2}
\end{enumerate}
\end{lemma}

\begin{proof}
The proof consists of three parts. Note that a similar analysis was already conducted in \cite{DBLP:conf/soda/AlistarhAEGR17} in a different setting. For completeness, we give the proof adapted to our model.

\paragraph{Part 1}
Let $r_0$ be the first time step, in which some agents change their stage from $\Init$ to $\Count$. Then, the set of these agents has size at least $1$ and at most $n/4(1+o(1))$ \whp (for the details see beginning of part II).
We prove now an intermediate result: there is a time when a certain constant fraction of agents concludes the initialization stage such that at least $1/32\cdot n$ many agents have their bit set to \ttrue and to \tfalse, respectively.

Let $I(t)$ be the set of agents $u$ that have $\Stage u \neq \Init$ at the beginning of round $t$, i.e.,
\[ I(t) = \set{u | \Stage u(t) \neq \Init}, \]
and define $\phi$ to be the last round $t$ in which $\abs{I(t)}/n < 1/5$ (note that $\phi$ can be $0$), i.e., 
\[ \phi = \max\set{t | \abs{I(t)} < \frac{1}{5} \cdot n } .\]
Then we have (by definition) 
\begin{equation} \label{eq:rho-0}
\abs{I( \phi + 1)} \geq n/5 .
\end{equation}

Now observe that an agent leaves the initialization stage in round $\phi$ when it interacts with another agent in $I(\phi)$ or it successfully passes \cref{alg:uniform-1-1} to \cref{alg:uniform-1-12} in \cref{alg:uniform-1} in rounds $\phi-1$ and $\phi$. 
The probability to encounter another agent with $\Stage u \neq \Init$ is at most $1/5$.
Additionally, if $\phi$ is odd, an agent $u$ with $Stage u = \Init$ leaves the initialization stage with probability at most $1/4$ due to hitting the own opinion in round $\phi -1$ and a different opinion in round $\phi$.
These events might not be independent; nevertheless, by a union bound we get that each agent leaves the initialization stage in round $\phi$ with probability at most $9/20$.
Since the agents act independently, an application of Chernoff bounds gives that \whp
\begin{equation} \label{eq:rho-1} \abs{ I( \phi + 1)} < \left(\frac{13}{20} + o(1)\right) \cdot n . \end{equation}

We now consider round $\phi + 2$.
Let $I_b(t)$ be defined as the set of agents $u$ that have $\Stage u \neq \Init$ \textbf{and} $\Bit u(t) = b$ in round $t$.
W.l.o.g.\ assume $\abs{I_0(\phi + 1)} \geq \abs{I_1(\phi + 1)}$
By \eqref{eq:rho-0} we have $\abs{I( \phi + 1)} \geq n/5$, and hence $\abs{I_0( \phi + 1)} \geq n/10$ (since we assumed $\abs{I_0(\phi + 1)} \geq \abs{I_1(\phi + 1)}$).
By \eqref{eq:rho-1}, at least $(7/20 - o(1)) \cdot n$ many agents are not in $I(\phi + 1)$.
Each of these agents interacts in round $\phi + 1$ with an agent in $I_0(\phi + 1)$ with probability at least $1/10$.
Again, an application of Chernoff bounds gives that at least
$(7/200 - o(1)) \cdot n \geq 1/32 n$ many agents join $I_1(\phi + 2)$ in round $\phi + 1$.
Hence for any $r \geq \phi + 2$ we have $I_1(r) \geq 1/32 n$, and by assumption
$\abs{I_0(\phi + 1)} \geq \abs{I_1(\phi + 1)}$ we have $I_0(r) \geq 1/10 n \geq 1/32 n$. (The case $\abs{I_0(\phi + 1)} \geq \abs{I_1(\phi + 1)}$ follows analogously by exchanging $I_0$ with $I_1$ and vice versa.)
This implies that for any $r \geq \phi + 2$ and any $b \in \set{\ttrue, \tfalse}$ we have \whp
\begin{equation}
\frac{\abs{I_b(r)}}{n} \in [1/32, 1-1/32] . \label{eq:rho-2}
\end{equation}

\paragraph{Part 2} We now show the first statement, $\rho < 50\log{n}$.

We start by showing that at least one agent enters the counting stage within at most $2\log{n}$ rounds.
Let $E_{r,u}$ be the event that in round $2r$ agent $u$ sets $\Stage u$ to $\Count$ in \cref{alg:uniform-1-12} of \cref{alg:uniform-1}.
The event $E_{r,u}$ occurs if and only if two events occur:
agent $u$ interacts with an agent $v$ of the same opinion in round $2r-1$, and agent $u$ interacts with an agent $w$ of a different opinion in round $2r$. For both, the probability is at least $1/n$, and the events are complimentary. (We assume that $k > 1$ and agents may interact with themselves.) Hence we observe that $\prob{E_{r,u}} > 1/n\cdot(1-1/n)$.
Note that all interactions of agents are independent and hence the $E_{r,u}$ are independent.
Then the probability that no agent enters the counting stage within $2\log{n}$ rounds is \[
\prod_{r'=1}^{r}\prod_{u}\left(1 - \prob{E_{r',u}} \right) \leq \left( 1 - \frac{1}{n}\cdot\left(1-\frac{1}{n}\right) \right)^{2n\log{n}} < 1/n. \]

Let now $r_0$ be the first round in which any of the events $E_{r,u}$ occurs for any agent $u$, i.e., $r_0 = \min\set{r | E_{r,u}}$.
Observe that any agent $u$ with $\Stage u = \Init$ proceeds to the counting stage once it interacts with another agent $v$ where $\Stage v \neq \Init$.
This means, the way how agents proceed to the counting stage can be modeled by a broadcast process starting in round $r_0$.
It follow from \cite{DBLP:conf/focs/KarpSSV00} that at time $6\log{n}$ every agent has concluded the initialization phase.

It remains to show that every agent concludes the counting stage after at most $44\log{n}$ additional rounds \whp.

Let $T_u$ be the counter of agent $u$ when it has concluded the counting stage.
As before, let $\phi$ be defined as the last round $t$ in which $\abs{I(t)}/n < 1/5$.
Let furthermore $p_{u,r}$ be the probability that $u$ interacts with another agent $v$ with $\Bit u(r) = \Bit v(r)$ in a round $r \geq \phi+2$.
In \eqref{eq:rho-2}, we obtained bounds on $\abs{I_b(r)}/n$ for any $r \geq \phi + 2$ and any $ b \in \set{\ttrue, \tfalse}$ and from these bounds we observe $p_{u,r} \in [1/32, 1-1/32]$ for $r \geq \phi+ 2$.

Observe that the number of successes of agent $u$ can be majorized by $T_u \preceq \phi + 2 + T_u^{(\max)}$ as follows.
The first expression $\phi+ 2 \leq 6\log{n}$ is an upper bound on $\phi + 2$. (In the first $\phi + 1$ rounds, we do not have a bound on $\abs{I_1(r)}/\abs{I_0(r)}$ and hence bound the number of successes by the number of rounds.)
The second expression is a geometrically distributed random variable $T_u^{(\max)} \sim  G(1/32)$. (Note that $G(1/32)$ models a probability of $1/32$ as \emph{failure probability}. In the case of a failure, the counting stage ends.)
With this random variable $T_u^{(\max)}$ we obtain
\[
\prob{T_u \leq 50\log{n} } \geq \prob{T_u^{(\max)} \leq 44\log{n}} 
= 1 - \left(1-\frac{1}{32}\right)^{44\log{n}} \geq 1 - 1/n^2.
\]
We take a union bound over all $n$ agents and get that $\max_u\set{T_u(\rho)}\leq 50\log{n}$ \whp.

\paragraph{Part 3}
We finally show the second statement, $\max_u\set{T_u(\rho)} > 0.1\log{n}$.

Let $T_u$ be the counter of agent $u$ when it has concluded the counting stage.
As before, let $\phi$ be defined as the last round $t$ in which $\abs{I(t)}/n < 1/5$, and recall that in round $\phi + 1$ at least $(7/20 - o(1)) \cdot n \geq 1/3 \cdot n$ many agents are not in $I(\phi + 1)$.
Let $R$ be the set of these agents with $\abs{R} \geq 1/3 n$.
The agents in $R$ leave the initialization stage in round $\phi + 1$ at the earliest.
In particular, they run the counting stage starting in a round $r \geq \phi + 2$.
Let now $p_{u,r}$ be the probability that one of these agents $u \in R$ interacts with another agent $v$ with $\Bit u(r) = \Bit v(r)$ and recall that $p_{u,r} \in [1/32, 1-1/32]$ for any $r \geq \phi + 2$ and any $ b \in \set{\ttrue, \tfalse}$.
Hence the number of successes of agent $u$ can be minorized by $T_u \succeq T_u^{(\min)}$ where $T_u^{(\min)} \sim G(1-1/32)$ is a geometrically distributed random variable with parameter $1-1/32$.
We compute
\[
\prob{T_u \leq 0.1\log{n} } = 1 - \left(1-\left(1-\frac{1}{32}\right)\right)^{0.1\log{n}}
= 1 - \left(\frac{1}{2}\right)^{0.5\log n}
= 1 - 1/\sqrt{n}.
\]
This means that the probability that an agent $u \in R$ reaches a counter value $T_u$ of at least $0.1\log{n}$ is at least $1/\sqrt{n}$.
Recall that the agents act independently. A Chernoff bound
on the number of agents in $R$ that reach a counter value of at least $0.1\log{n}$ gives the desired result: at least one agent $u$ reaches 
$T_u \geq 0.1\log{n}$ and hence at time $\rho$ when all agents have concluded the counting stage it holds that $\max_u\set{T_u(\rho)} > 0.1\log{n}$.
\end{proof}

The technical lemma above allows us to show the main results for this section.
\begin{proposition}
The protocol defined via \cref{alg:uniform-1,alg:uniform-2,alg:uniform-3,alg:uniform-4} is uniform. Within $\BigO{\log{n}}$ rounds, every agent $u$ adopts the same value of $T_u = T \in [0.1\log{n},50\log{n}]$. In round $1000T$, every agent starts to run \cref{alg:uniform-4}, synchronizing phases modulo $1000 T$.
\end{proposition}
\begin{proof}
From the description of the pseudo codes it is clear that the protocol does not depend on $n$.
To prove the correctness of the synchronization, we first show the following invariant. Recall that $\rho$ is the time when all agents have completed the counting stage.

Let $T = \max_u{T_u(\rho)}$ be the maximal value obtained by any agent in the counting stage. For any agent $u$ with $T_u = T$ it holds in any round $r$ that $\Round u(r) = r \bmod 1000T$. (Notice that $T_u$ may change during the execution of the protocol.)

Observe that there are two possible ways how an agent $u$ can obtain $T_u = T$. Either agent $u$ samples $T$ in the counting stage. In this case, the agent has correctly counted rounds from the very beginning, and it proceeds to count rounds modulo $1000T$ for the rest of the time. (The reset $\Round u \gets 0$ at the end of the synchronization phase behaves identical to counting modulo $1000T$.)
Otherwise, the agent $u$ adopts $T$ from another agent $v$ in some round $r$. In this case, $u$ also adopts $\Round u \gets \Round v$ according to \cref{alg:uniform-3}. By the invariant, $v$ had $\Round v(r) = r \bmod 1000T$ in round $r$ and hence the invariant also holds for $u$.

Note that once an agent $u$ adopts $T_u = T$, it reverts itself to its initial opinion. From that time on, there are no further instructions that could alter $T_u$ or $\Round u$ (except for counting regularly in $\Round u$ modulo $1000 T$). This implies  that at the end of the synchronization stage in round $r = 1000 T$, every agent $u$ has $\Opinion u = \InitialOpinion u$.

It remains to argue that at the end of the synchronization stage in round $r = 1000 T$, every agent $u$ has $T_u(r) = T$.
Recall that $\rho \leq 50\log{n}$ and $T \geq 0.1\log{n}$.
The synchronization stage ends in round $1000T$, which is at the earliest in round $100\log{n}$.
Hence, starting with round $50\log{n}$, at least one agent knows the maximal value $T$ and starts broadcasting this value to all other agents. 
With at least $50\log{n}$ rounds remaining until round $1000 T$, every other agent $u$ with $T_u < T$ adopts the maximum $T$.
When some agent adopts $T$ in some round $r$, from then on it executes the (remainder of the) synchronization stage with correct value $T_u = T$ and correct value of $\Rounds u = r$.
Observe that the $50\log{n}$ rounds are enough for the broadcast to reach all agents \cite{DBLP:conf/focs/KarpSSV00}.
\end{proof}

\begin{remark*}
In round $1000T$, all agents start the running stage (\cref{alg:uniform-4}), fully synchronized in phases of length $1000T$.
Starting with round $1000T$, the execution of the protocol is identical to an execution of \cref{alg:consensus-gossip-model} with $\T{BC}$ set to $1000T$.
\end{remark*}

\section{Auxiliary Results}
\label{sec:auxiliary-results}

\subsection{Concentration Results}

\begin{theorem}[\cite{DBLP:books/daglib/0012859}, Theorem $4.4$, $4.5$]
\label{lemma:chernoff_poisson_trials}
    Let $X_1, \dots, X_n$ be independent Poisson trials with $\Pr[X_i = 1] = p_i$ and let $X = \sum X_i$ with $\E{X} = \mu$. Then the following Chernoff bounds hold: \\
    For $0 < \delta' \leq 1$:
    \[
        \Pr[X > (1+\delta') \mu] \leq e^{-\mu {\delta'}^2 / 3}.
    \]
    For $0 <\delta' < 1$:
    \[
        \Pr[X < (1-\delta') \mu] \leq e^{-\mu {\delta'}^2 / 2},
    \]
\end{theorem}

\begin{lemma}[Super-exponential Chernoff Bound]
\label{lem:chernoff-superexponential} 
	Let $X_1, ... ,X_n$ be $n$ independent random variables taking value in $\{ 0,1\}$ and $X = \sum_{i=1}^{n} X_i$. Then, for $\Ex{X} = \mu$ and $\delta >0$ it holds that
	\begin{align*}
		&\Pr \Big[ X > \mu +  \delta \sqrt{\mu} \Big] < \exp(-c \cdot \delta^2) &\text{ for } \delta^2 \leq \mu \\
		&\Pr \Big[X > \mu + \frac{\delta^2}{1 + \ln(\frac{\delta^2}{\mu})}  \Big]  < \exp(-c \cdot \delta^2) &\text{ for } \delta^2 > \mu
	\end{align*}
	where $c > 0$ is a universal constant.
\end{lemma}

\begin{proof}
	Let $X$ be defined as in the lemmas statement. From the Chernoff bound we have for $\lambda > 0$ that $\Pr[X > \mu(1+\lambda)] < \exp(-\min\{\lambda, \lambda^2\} \cdot \mu /3)$. It is easy to see that, for fitting constants $c_1, c_2 > 0$, this implies for $\delta >0$ that 
	\begin{align}
		\Pr \left[ X > \mu +  \delta \sqrt{\mu} \right] < \exp(-c_1 \cdot \delta^2) &\text{ if } \delta^2 \leq \mu \text{, and} \label{strong-chernoff-eq-1} \\
		\Pr \left[ X > \mu +  \delta^2 / 7 \right] < \exp(-c_2 \cdot \delta^2) &\text{ if } \delta^2 > \mu .  \label{strong-chernoff-eq-2}
	\end{align}
	Inequality (\ref{strong-chernoff-eq-1}) corresponds directly to the first statement of the lemma. However, observe that (\ref{strong-chernoff-eq-2}) only implies the second desired inequality in case $\mu < \delta^2 \leq \mu \cdot e^{6}$. This results from the fact that $\delta^2 / 7 \leq \delta^2 / ( 1 + \ln(\delta^2 / \mu))$ in this setting.
	
In order to tackle the case of $\delta^2 > \mu \cdot e^{6}$, we employ a different version of the Chernoff bound, which is tighter for large values of $\delta$. That is, by inequality (1.10.8) of \cite{Solr-1684971306}  we have for $\lambda > 0$ that $Pr[X > \mu(1+\lambda)] < (e/\lambda)^{(\lambda \cdot \mu)}$. We now define  $y= \ln(\delta^2 /\mu) > 6$ and set $\lambda = \delta^2  (1+y)^{-1} \mu^{-1} = e^{y} (1+y)^{-1}$. This way we get
\begin{equation}
\label{strong-chernoff-eq-3}
	\Pr\left[ X > \mu + \frac{\delta^2}{1 + \ln(\delta^2 / \mu)}\right] < \left( \frac{e(1+y)}{e^{y}} \right)^\frac{\delta^2}{(1+y)}.
\end{equation}
For $y \geq 6$, it holds that $(1+y) < e^{y/2 - 1}$. This implies that $e(1+y) e^{-y} <  e^{-y/2}$ and allows us to upper-bound the term on the right-hand side of (\ref{strong-chernoff-eq-3}) as follows
\[ 
\left( \frac{e(1+y)}{e^{y}} \right)^\frac{\delta^2}{(1+y)} < e^{-\frac{y}{2} \frac{\delta^2}{1+y}} <  e^{-\frac{\delta^2}{4}}.
\]
The results follow when setting $c= \min\{c_1,c_2, 1/4\}$.
\end{proof}

\begin{theorem}[General Chernoff upper Bound]\label{General_Upper_Chernoff_bound}
Let $X_1\cdots,X_n$ be independent 0-1 random variables. Let $X=\sum_{i=1}^n X_i$ and $\mu_u\ge 0$ such that $\E{X}\le \mu_u$. Then, for any $\delta'>0$
\[
\Pr[X\ge (1+\delta')\cdot \mu_u]\le e^{-\frac{\delta'^2\cdot \mu_u}{2+\delta'}}.
\]

\end{theorem}

The next result is a modified version of a drift result in \cite{DBLP:conf/spaa/DoerrGMSS11}. We adapted the proof slightly. The original proof can be found in the full version \cite{DBLP:conf/dagstuhl/DoerrGMSS09}.
\begin{theorem}[\cite{DBLP:conf/spaa/DoerrGMSS11}, Modified version of Claim 2.9]
\label{lemma_drift_markov_chain}
    Consider a Markov Chain $(W(t))_{t=1}^\infty$ with the state space $\{0,\dots,  c_4\sqrt{\log{n}}\}$  for an arbitrary constant $c_4 > 0$.
    For some constants $c_2>0$ and $\varepsilon>0$  it has the following properties:
    \begin{itemize}
        \item $\Pr[W(t+1) \geq 1 \vert W(t) = 0] = \BigOmega{1}$
        \item $\Pr[W(t+1) \geq \min\{(1+\varepsilon)W(t), m\}] \geq 1-e^{-c_2 W(t)}$
    \end{itemize}
    Then it holds for $t = \BigO{\log n}$ that 
    \[
        \Pr[W(t) \geq m] \geq 1-n^{-2}
    \]
\end{theorem}
\begin{proof}
    We follow the outline of the original proof in the full version \cite{DBLP:conf/dagstuhl/DoerrGMSS09}.
    Let $B \in \N \cup \{0,\infty\}$ be a random variable that denotes the number of consecutive successful rounds (abb.: winning streak) when starting at round $t_0$ with $W_{t_0} = 0$ until the first failure similar to a geometrically distributed random variable.
    Let $\ell^* \in \N$ be the smallest number such that $W(\ell^*) \geq c_4\sqrt{\log n}$.
    We know that $\Pr[B = 0] \leq 1-p$ and for any $1 \leq \ell \leq \ell^* $
    \[
        \Pr[B = \ell] \leq p \cdot \prod_{j=1}^{\ell-1}(1-e^{-c_2 (1+\varepsilon)^{2j}}) \cdot e^{-c_2  (1+\varepsilon)^{2\ell}} 
        \leq p' \cdot e^{-c_2 (1+\varepsilon)^{2\ell}}
        \leq c_3 \cdot \delta^\ell
    \]
    for some constant $p',\delta < 1$.
    It is easy to see that for some constant $c_q<1$
    \[
        \Pr[B = \ell \vert B < \infty] 
        \leq c_q \cdot \delta^\ell
    \]
    Thus, $\E{B \vert B < \infty} = \BigTheta{1}$.
    In a similar way it also follows for any starting value $W(t_0) =w_0 \geq 0$ that $\Pr[B = \ell \vert W(t_0) = w_0] \leq c_q \cdot \delta^\ell$, i.e., the probability holds irrespective of $W(t_0)$.
    If a winning streak holds for more than $t' =\BigTheta{\log \log n} $ phases, then we reach a phase where $W(t_0+t') = c_4 \sqrt{\log{n}}$ with probability
    \[
        \Pr[B \geq t'] \geq p \cdot \prod_{j=1}^{t'}(1-e^{-c_2 (1+\varepsilon)^{2j}}) \geq c_5
    \]
    for some constant $c_5 < 1$.
    Thus, by a standard Chernoff bound, we have to consider $\BigTheta{\log{n}}$ attempts such that at least one streak lasts for more than $t'$ phases \whp. \\
    As stated in the original proof,  at most $\BigTheta{\log{n}}$ attempts requires at most $\BigTheta{\log{n}}$ phases \whp which finishes the proof.
\end{proof}

\subsection{Pólya-Eggenberger Distribution}
\label{sec:polya-eggenberger}

The Pólya-Eggenberger process is a simple urn process that consists of $n$ steps. Initially, the urn contains $a$ red and $b$ blue balls, where $a,b \in \mathbb{N}_0$. One fixed step of the process can be described as follows. First, a ball is drawn from the urn uniformly at random with replacement. Second, an additional ball that matches the color of the drawn ball is added to the urn. The corresponding \emph{Pólya-Eggenberger distribution}, denoted by $\PE(a,b,n)$, describes the number of \emph{total} red balls that are contained in the urn after all $n$ steps. Alongside a more detailed discussion of this process, the following tail inequalities  have been shown in \cite{DBLP:conf/podc/BankhamerEKK20} \footnote{In \cite{DBLP:conf/podc/BankhamerEKK20} the Pólya-Eggenberger distribution is defined to describe the number of \emph{added} instead of \emph{total} red balls at the end of the process. We adapted \cref{thm:polya_bound_total_balls, thm:polya_bound_small_support} accordingly.}.

\begin{theorem}[Theorem 1 of \cite{DBLP:conf/podc/BankhamerEKK20}]
\label{thm:polya_bound_total_balls}
Let $A \sim \PE(a,b,n- (a+b))$,  $\mu = (a/(a+b)) n$ and $a+b \geq 1$. Then, for any $\delta$ with $0 < \delta < \sqrt{a}$ and some small constant $1 > \polyaConstSmall>0$ it holds that 
\[
    \Pr \Big(A < \mu - \sqrt{a} \cdot \frac{n}{a+b} \cdot  \delta\Big) < 4 \exp(-\polyaConstSmall \cdot \delta^2)
\]
\[
    \Pr \Big(A > \mu + \sqrt{a} \cdot \frac{n}{a+b} \cdot  \delta \Big) < 4 \exp(-\polyaConstSmall \cdot \delta^2)
\]
\end{theorem}

\begin{theorem}[simplified Theorem 47 of \cite{DBLP:conf/podc/BankhamerEKK20}]
\label{thm:polya_bound_small_support}
Let $A \sim \PE (a,b,n - (a+b))$ with $1 \leq a \leq b$. Then, for some large constant $\polyaConstLarge>1$ it holds that 
\[
    P \Big( A > \frac{n}{a+b} \cdot ( 3a + \polyaConstLarge \log n) \Big) < 2n^{-2}, 
\]
\end{theorem}

\subsection{Anti-Concentration Results}

\begin{lemma}
\label{lem:reverse-chernoff}
    Let $ X \sim \Bin(n,p)$ with $\mu=np$. Then, for $\delta$ with $n/2 > (1+\delta) \mu > \mu$, it holds that
    \[
        \Pr[X \geq (1+ \delta) \mu ] \geq \frac{1}{\sqrt{8(1+\delta) \mu}} \cdot \left(1 - \frac{\delta^2 \mu^2}{n-(1+\delta) \mu}\right) \cdot \left( \frac{e^\delta}{(1+\delta)^{(1+\delta)}}\right)^\mu .
    \]
\end{lemma}
\begin{proof}
    We start by considering some $k$ with $n/2 > k > \mu$. By Lemma 4.7.2 of \cite{A90} we have that 
    \begin{equation}
    \label{eq:anti-1}
        \Pr[X \geq k] \geq \frac{1}{\sqrt{8k}} \exp\left(-n D\left(\frac{k}{n} ~\Big|\Big|~ p\right)\right),
    \end{equation}
    where $D( \cdot || \cdot)$ denotes the Kullback-Leibler divergence with 
    \[
        D\left(\frac{k}{n} ~\Big|\Big|~ p\right) = \frac{k}{n} \ln \left(\frac{k}{\mu}\right) +  \left(1 - \frac{k}{n}\right) \ln\left( \frac{n-k}{n-\mu} \right).
    \]
    Hence, it follows that
    \begin{equation}
    \label{eq:anti-2}
        \exp\left(-n D\left(\frac{k}{n} ~\Big|\Big|~ p\right)\right) = \left(\frac{\mu}{k}\right)^k \cdot \left(\frac{n-\mu}{n-k}\right)^{n-k}.
    \end{equation}
    Next, we use that $(1+x/m)^m\ge e^x(1-x^2/m)$ for $m>1$ and $|x|<m$, which can be derived with the help of the well-known inequality $(1+\frac{1}{x})^{x+1} \geq e$ as well as the Bernoulli inequality. This implies  
    \begin{equation}
     \label{eq:anti-3}
       \left(\frac{n-\mu}{n-k}\right)^{n-k} = \left(1 + \frac{k-\mu}{n-k}\right)^{n-k} \geq e^{k-\mu} \left(1 - \frac{(k-\mu)^2}{n-k}\right).
    \end{equation}
    When combining (\ref{eq:anti-1}) with (\ref{eq:anti-2}) and then (\ref{eq:anti-3}),  the statement follows for $k=(1+\delta)\mu$.
\end{proof}

\begin{lemma}
\label{lem:reverse-chernoff-simple}
    Let $ X \sim \Bin(n,p)$ with $\mu=np$. Then, for $\delta$ with $n/2 > (1+\delta) \mu > \mu$ \textbf{and}  $\delta \mu < \sqrt{n} / 2$, it holds that
    \[
        \Pr[X \geq (1+ \delta) \mu ] \geq \frac{1}{6 \cdot \sqrt{(1+\delta) \mu}} \cdot \exp \left( -\delta^2 \mu \right)
    \]
\end{lemma}
\begin{proof}
    The result is implied by \cref{lem:reverse-chernoff}. We lower bound some factors involved in the right-hand side of  \cref{lem:reverse-chernoff}. It follows from $\delta \mu < \sqrt{n} / 2$  and  $\mu < n/2$ that
    \[
        \left(1 - \frac{\delta^2 \mu^2}{n-(1+\delta) \mu}\right) > \frac{1}{2} (1 - \LittleO{1}).
    \]
    Additionally, when using the well-known inequality $e^x \geq (1+x)$ twice, we get
    \[
        \left( \frac{e^\delta}{(1+\delta)^{(1+\delta)}}\right) \geq \frac{1}{(1+\delta)^\delta} = \left(\frac{1}{1+\delta} \right)^\delta \geq \left( \frac{1}{e^\delta} \right)^\delta = e^{-\delta^2}. \qedhere
    \]
\end{proof}

\begin{lemma}[Lemma 4 of \cite{DBLP:journals/siamcomp/KleinY15}]
\label{lemma:reverse_chernoff_v2}
    Let $X \sim \Bin(n,p)$ with $\mu = np$.
    For any $\delta \in (0,1/2]$ and $p \in (0,1/2]$, assuming $\delta^2 \mu \geq 3$, it holds that
    \[
        \Pr[X \geq (1+\delta)\mu] \geq e^{-9\delta^2 \mu}
    \]
    \[
        \Pr[X \leq (1-\delta)\mu] \geq e^{-9\delta^2 \mu}
    \]
\end{lemma}

\begin{theorem}[Theorem 1 of \cite{DBLP:journals/corr/GreenbergM13}]
\label{lemma:binomial_anti_expectation}
    Let $X \sim \Bin(n,p)with$ $ \mu = np $.
    If $1/n < p $, then 
    \[
        \Pr[X \geq \mu] > 1/4 .
    \]
\end{theorem}

\end{document}